\documentclass[%
 reprint,
 amsmath,amssymb,
 aps,
prb,
floatfix,
]{revtex4-1}

\usepackage{ifpdf}
\usepackage{amsmath} 
\usepackage{amssymb}
\usepackage{amsfonts}
\usepackage{braket}
\usepackage{color}
\usepackage[colorlinks=true,linkcolor=blue,citecolor=blue,breaklinks]{hyperref}
\usepackage{graphicx}
\usepackage{dcolumn}
\usepackage{amsthm} 
\usepackage{bm}

\newcommand{\tr}{\text{Tr}}
\newcommand{\ad}{\text{ad}}
\newcommand{\Q}{\mathcal{Q}}

\newtheorem{thm}{Theorem}[section]
\newtheorem{lem}{Lemma}[section]
\newtheorem{prop}{Proposition}[section]
\newtheorem{defi}{Definition}[section]

\begin{document}

\title{Explicit construction of quasiconserved local operator of translationally invariant nonintegrable quantum spin chain in prethermalization}

\author{Cheng-Ju Lin}
\author{Olexei I. Motrunich}
\affiliation{Department of Physics and Institute for Quantum Information and Matter, California Institute of Technology, Pasadena, California 91125, USA}

\date{\today}

\begin{abstract}
We numerically construct translationally invariant quasi-conserved operators with maximum range $M$ which best-commute with a non-integrable quantum spin chain Hamiltonian, up to $M = 12$.
In the large coupling limit, we find that the residual norm of the commutator of the quasi-conserved operator decays exponentially with its maximum range $M$ at small $M$, and turns into a slower decay at larger $M$.
This quasi-conserved operator can be understood as a dressed total ``spin-z'' operator, by comparing with the perturbative Schrieffer-Wolff construction developed to high order reaching essentially the same maximum range.
We also examine the operator inverse participation ratio of the operator, which suggests its localization in the operator Hilbert space.
The operator also shows an almost exponentially decaying profile at short distance, while the long-distance behavior is not clear due to limitations of our numerical calculation.
Further dynamical simulation confirms that the prethermalization-equilibrated values are described by a generalized Gibbs ensemble that includes such quasi-conserved operator.
\end{abstract}

\maketitle


\section{\label{sec:intro} Introduction}
Two decades ago, the eigenstate thermalization hypothesis (ETH) was proposed as a mechanism accounting for the validity of the statistical mechanics in isolated quantum systems.\cite{deutsch_quantum_1991, srednicki_chaos_1994} 
In contrast, many-body localization (MBL) refers to a class of interacting systems that fail to thermalize due to the presence of strong disorder.
Phenomenologically, MBL systems can be viewed as having an extensive number of local integrals of motion,\cite{serbyn_local_2013, huse_phenomenology_2014,PhysRevX.4.011052, chandran_constructing_2015, obrien_explicit_2016} analogous to integrable quantum systems.

Many research works have proposed systems ``in between,'' namely, systems that fail or partially fail to thermalize but are disorder-free.
For example, Ref.~\onlinecite{grover_quantum_2014} proposed a phase of matter called ``quantum disentangled liquid,'' where the system is composed of heavy degrees of freedom and light degrees of freedom, and where after a partial measurement the light degrees of freedom will localize around the heavy degrees of freedom.
Recent numerical and theoretical works provided some support for the existence of such phases of matter.\cite{garrison_partial_2017, veness_quantum-disentangled_2016}
Other studies observed that in some such systems the dynamics shows behavior similar to MBL systems.\cite{schiulaz_dynamics_2015, van_horssen_dynamics_2015, yao_quasi-many-body_2016}

While numerous proposals have tried to realize MBL in translationally invariant systems, it was argued that this cannot happen in the true sense of MBL.
However, some phenomenological aspects of MBL can still be realized in such systems.\cite{de_roeck_asymptotic_2013, de_roeck_asymptotic_2014, de_roeck_scenario_2014, de_roeck_can_2014, bols_asymptotic_2016}
That is, when one performs some dynamical simulation in such a system, the system will appear to be localized at some intermediate time scale, but will delocalize eventually.
Therefore, one may view such ``quasi-localization'' or ``asymptotic localization'' as prethermalization, where the system equilibrates to a state which is described by a Gibbs ensemble controlled by some effective Hamiltonian (instead of the original Hamiltonian) at some intermediate time, and truly thermalizes only at much later time.
Nevertheless, a recent work has proposed another model with translational invariance and has claimed to find true disorder-free localization,\cite{smith_disorder-free_2017} so this question is still open.

Prethermalization has been observed and studied in many different systems.
In particular, various works showed that systems with weak integrability breaking exhibit this phenomenon.\cite{essler_quench_2014, bertini_thermalization_2016, langen_prethermalization_2016}
In addition, prethermalization has been shown rigorously to exist in periodically driven many-body systems under strong driving frequencies using the Floquet-Magnus expansion\cite{kuwahara_floquet-magnus_2016, mori_rigorous_2016} and renormalization technique.\cite{abanin_effective_2017, abanin_rigorous_2017}
The latter also applies to time-independent many-body systems, and in particular can be used to prove rigorously the presence of exponentially long relaxation times of ``particles'' such as doublons in the Hubbard model in the strong coupling limit.\cite{PhysRevLett.115.205301,PhysRevB.82.224302,PhysRevLett.104.080401}
There are also very recent proposals utilizing the prethermalization to protect the edge modes in the topological superconductor.\cite{kemp_long_2017, else_prethermal_2017}

In fact, we can view most of the aforementioned prethermalization systems as having quantities with hierarchically different thermalization time scales or having different rates of dynamics.
Upon time evolution, the fast degrees of freedom relax very quickly, while the slow degrees of freedom evolve slowly during this initial period.
This results in the apparent prethermalization stage, where the slow degrees of freedom appear to be frozen.
These quantities with slow dynamics can be viewed as {\it quasi-conserved}.\cite{fagotti_conservation_2014, mierzejewski_approximate_2015}
Emergence of such a quasi-conserved quantity is what accounts for the prethermalization stage.
If such a quantity could develop an exact conservation law, this would extend the prethermalization to infinitely long time and would correspond to partial breakdown of the ETH, as envisioned, e.g., in Refs.~\onlinecite{grover_quantum_2014, garrison_partial_2017}.

Motivated by this point of view, in this paper we numerically systematically search for such hidden quasi-conserved quantities which cannot be directly identified from the Hamiltonian itself.
Following the ``slowest operator formalism'' introduced in Ref.~\onlinecite{kim_slowest_2015}, we numerically construct the quasi-conserved local operator for the non-integrable spin model
\begin{equation}\label{eqn:hamiltonian}
H = \sum_{j = -\infty}^\infty \big( J Z_j Z_{j+1} + h Z_j + g X_j \big) ~,
\end{equation}
where $X_j$, $Y_j$, and $Z_j$ denote Pauli matrices operating on site $j$ of the one-dimensional chain.
We constrain our slowest operator to be translationally invariant and represented as a sum of local terms.
We find that, in the large $g$ limit, there exists a quasi-conserved operator whose thermalization time scale increases exponentially as one increases its maximum range up to some point.
Furthermore, the operator can be understood as a dressed ``total spin-z operator'' (for appropriately chosen spin axes).
This operator has a very slow dynamics compared to other quantities.
We also simulate the dynamics of the quantum spin chain following a quench and confirm that this quasi-conserved quantity has a non-trivial effect.
Specifically, at intermediate times, the system equilibrates to a state which can be described by a generalized Gibbs ensemble (GGE) that includes such a quantity as an ``integral of motion.''
While our study cannot reach infinite maximum range, we find that the rate of decrease of the slowest operator with the maximum range becomes weaker beyond some point and starts resembling behavior observed in regimes of good thermalization.
A conservative interpretation of this behavior is that our system shows only prethermalization with very long time scale.
Nevertheless, the available data does not rule out a more exotic possibility that the slowest operator converges and becomes exactly conserved in the thermodynamic limit, which would indicate breakdown of the ETH.

The paper is organized as follows.
In Sec.~\ref{sec:method}, we briefly describe the formalism we use to search for the slowest operator in the translationally invariant setting.
In Sec.~\ref{sec:scaling}, we present our numerical results focusing on the scaling of the ``residual norm'' (i.e., norm of the commutator with the Hamiltonian) versus the maximum range of the operator.
In the large coupling limit, we find that the residual norm shows exponential decay at least on short distances and identify the slowest operator as quasi-conserved operator.
As a comparison, in Sec.~\ref{sec:iterative}, we use the Schrieffer-Wolff approach to perturbatively construct a quasi-conserved operator, which can be understood as a dressed total spin-$z$ operator.
We find that in the large coupling limit, the overlap between the perturbative construction and exact numerical construction of the slowest operator is almost 100\%; thus we understand the nature of the slowest operator in this regime, at least up to some value of the maximum range.
In Sec.~\ref{sec:characteristic}, we examine the operator inverse participation ratio and the weight distribution in the slowest operator at different distances, demonstrating its localization in the operator space and real space.
To verify the conjecture that this quasiconserved quantity results in prethermalization, we explicitly simulate a quench dynamics in Sec.~\ref{sec:dynamics} and confirm the importance of the quasiconserved quantity when describing the equilibrated values at intermediate time.
Finally, in Sec.\ref{sec:conclusion}, we summarize and discuss some outstanding questions.
Several appendices all focus on the Schrieffer-Wolff approach:
Appendix~\ref{app:ladderformalism} presents a ladder algebra formalism convenient for analytical calculations at low order.
Appendices~\ref{app:boundHn} and \ref{app:boundadHI} present some analytical bounds on the convergence of the Schrieffer-Wolff procedure, while Appendix~\ref{app:conv_radius} presents better bounds calculated numerically.
Finally, Appendix~\ref{app:Vmnorm} compares these bounds with exact numerical calculations, finding that the former are gross overestimations; we trace possible origins of these overestimations and consider how one might improve upon them and speculate about implications for the Schrieffer-Wolff approach.

\section{\label{sec:method} Method of the slowest operator}
Our motivation is to numerically search for the operator that ``best-commutes'' with the Hamiltonian.
We focus on translationally-invariant Hermitian operators obtained as sums of local terms and adopt the formalism of Ref.~\onlinecite{kim_slowest_2015}.
We restate this approach as a problem in the operator Hilbert space as follows.

We consider traceless, and translationally-invariant operators with maximum range $M$,
\begin{equation}
\Q^{(M)} = \sum_{j = -\infty}^\infty q_j^{(M)} ~,
\end{equation}
where $q_j^{(M)}$ is an operator with support on a region extending from site $j$ to site $j + M - 1$.
We denote the space of traceless translationally-invariant operators with maximum range $M$ as $\mathcal{T}_M$.
The operator space $\mathcal{T}_M$ is a vector space, as one can easily verify.
A natural basis for $q_j^{(M)}$ is provided by ``Pauli string operators,'' i.e., operators of the form $\prod_{k = j}^{j + M - 1} A_k$ where $A_k$ can be $I$, $X$, $Y$, or $Z$ acting on site $k$, and $A_k$ are independent for different $k$.
However, there is a ``gauge degree of freedom'' for the representation of $q_j^{(M)}$.
For instance, we can write $H = \sum_j q_j \in \mathcal{T}_2$ using $q_j = J Z_j Z_{j+1} + h Z_j I_{j+1} + g X_j I_{j+1}$ or $q_j = J Z_j Z_{j+1} + h I_j Z_{j+1} + g I_j X_{j+1}$, etc.
We fix the gauge by requiring the operator $A_k$ on the first site, $k = j$, to be non-identity in every Pauli string basis vector, i.e., $A_j$ can only be $X$, $Y$, or $Z$, while $A_{k > j}$ can be $I$, $X$, $Y$, or $Z$.
This also automatically satisfies the tracelessness condition.
The Hermiticity condition of an operator just corresponds to the condition of real coefficients in this basis.
It is now easy to see that the dimension of $\mathcal{T}_M$ is $\text{dim}(\mathcal{T}_M) = 3 \cdot 4^{M-1}$.

We define the Frobenius inner product (also know as Hilbert-Schmidt inner product) on the operator space $\mathcal{T}_M$ as
\begin{equation}
\langle \Q, \Q' \rangle = \frac{\tr[q_j^\dagger q^\prime_j]}{\tr[I^{\otimes M}]} ~,
\label{innerprod}
\end{equation}
where $q_j, q'_j$ are understood in the above gauge acting on $M$ sites only and $I^{\otimes M}$ is the identity operator also acting on $M$ sites.
One can easily see that the aforementioned Pauli-string operators are advantageous as they form an orthonormal basis under this inner product.
The above inner product defines the norm $\|\Q\|_\text{F} \equiv \sqrt{\langle \Q, \Q \rangle}$, which we can view as an ``intensive Frobenius norm'' (see below).
For example, $\|H\|_\text{F} = \sqrt{J^2 + g^2 + h^2}$.
Note that instead of the conventional definition of the operator inner product, here we only take the local piece $q_j$ in the trace calculation after the gauge fixing.
This definition has the advantage that the norm is ``intensive,'' compared to the conventional definition of Frobenius norm that would increase with the system size.
In fact, if we consider a chain of length $L$ with periodic boundary conditions and operators $\Q^{(M)} = \sum_{j=1}^L q_j$ (assuming $M < L$), we can easily verify that the above inner product is simply appropriately scaled conventional Frobenius inner product: 
\begin{equation}\label{eqn:general_inner}
\langle \Q, \Q' \rangle = \tr[\Q^\dagger Q']/(L\, \tr[I^{\otimes L}])~.
\end{equation}
In other words, Eq.~(\ref{innerprod}) is obtained from Eq.~(\ref{eqn:general_inner}) when applied to this ``gauge-fixing'' writing of the translationally-invariant operators.
If one does not use the gauge-fixing, one should use Eq.~(\ref{eqn:general_inner}) to calculate the inner product.
In what follows, we will always use only the intensive Frobenius norm, often dropping the descriptor ``intensive'' for brevity.

A natural embedding $\mathcal{T}_M \subset \mathcal{T}_N$ for $M < N$ is obtained by the tensor product with the identities, $q_j^{(N)} = q_j^{(M)} \otimes I_{j+M} \otimes \dots \otimes I_{j+N-1}$, where $\sum_{j}q_j^{(M)} \in \mathcal{T}_M$ and $\sum_{j}q_j^{(N)} \in \mathcal{T}_N$.
We will not emphasize the difference between $\sum_j q_j^{(M)}$ and $\sum_j q_j^{(N)}$, since it only depends on what operator space one is considering, while the inner product in Eq.~(\ref{innerprod}) is independent of the embedding.
We can further consider the norm closure $ \overline{\bigcup_{M \in  \mathbb{N}} \mathcal{T}_M}$, which is a mathematically well-defined Hilbert space.

The commutator with a fixed operator can be viewed as a linear map between the operator spaces.
We define the superoperator
\begin{equation}
\ad_A(O) \equiv [A, O] ~.
\end{equation}
Clearly, $\ad_H$ is a linear map from the operator space $\mathcal{T}_M$ to space $\mathcal{T}_{M+1}$, since $H \in \mathcal{T}_2$.
In fact, for any operator $A \in \mathcal{T}_r$ and $O \in \mathcal{T}_s$, we have $\ad_{A}(O) \in \mathcal{T}_{r + s - 1}$.
Using the Pauli string basis, we can write down the matrix representation $\mathbf{B}$ for $\ad_H$, which in general will be a $3 \cdot 4^M \times 3 \cdot 4^{M-1}$ matrix.
We want to find an operator in ${\cal T}_M$ that ``best commutes'' with the Hamiltonian, which we define as minimizing the {\it residual norm} $\|\ad_H(\Q^{(M)})\|_\text{F}$ under the constraint $\|\Q^{(M)}\|_\text{F} = 1$.
This corresponds to finding the smallest singular value $\sigma_0$ of $\mathbf{B}$, or the smallest eigenvalue $\lambda_0$ of $\mathbf{C} \equiv \mathbf{B}^\dagger \mathbf{B}$, where $\lambda_0 = \sigma_0^2$.
The corresponding eigenoperator is the sought-for {\it slowest operator}; we will denote this operator as $\Q_0^{(M)}$ and the corresponding eigenvalue as $\lambda_0(M)$, which will be the squared residual norm of the slowest operator.
To avoid the trivial zero-eigenvalue solution given by the Hamiltonian itself, we add $\lambda_h |H \rangle \langle H|$ to $\mathbf{C}$, with large enough $\lambda_h$ such that the slowest operator is nontrivial.
Thus found operator $\Q_0^{(M)}$ is orthogonal to $H$ in the Frobenius inner product.

Note that in the $I$-$X$-$Y$-$Z$ Pauli-string basis, $\mathbf{C}$ is always a symmetric matrix with real coefficients.
This guarantees the eigenvalues to be real, and the eigenvectors can be chosen with real amplitudes in the $I$-$X$-$Y$-$Z$ Pauli-string basis.
This means that the slowest operator $\Q_0^{(M)}$ can always be chosen to be Hermitian.
In other words, we fix the overall phase of the eigenoperator by requiring the Hermicity of the operator, up to a minus sign.

We can argue that this defines a procedure to find a translationally invariant (quasi)-local conserved quantity in the thermodynamic limit.
Indeed, consider the limit $\lambda_0(\infty) = \lim_{M \to \infty} \lambda_0(M)$.
Since $\lambda_0(M)$ is a decreasing function of $M$ bounded from below by $0$, $\lambda_0(\infty)$ exists.
If $\lambda_0(\infty) = 0$ and $\lim_{M \to \infty} \frac{\Q_0^{(M)}}{\|\Q_0^{(M)}\|_\text{F}}$ exists, then we have a normalizable operator [hence quasilocal or local if $\lambda_0(M) = 0$ for some finite $M$ already] which commutes with the Hamiltonian.
If such (quasi)local conserved quantity does exist, a suitable thermal equilibrium description should include this quantity in the GGE.
On the other hand, even though an arbitrary linear combination of eigenstate projectors $\hat{A} = \sum_E a_E |E \rangle \langle E|$ commutes with the Hamiltonian, $\hat{A}$ can be non-normalizable under our definition of the Frobenius norm.
It is therefore not guaranteed that $\lambda_0(\infty) = 0$.
Furthermore, even if $\lambda_0(\infty) = 0$, we cannot guarantee that the limit $\lim_{M \to \infty} \frac{\Q_0^{(M)}}{\|\Q_0^{(M)}\|_\text{F}}$ exists.
In practice, one can only find $\Q_0^{M}$ with $M$ finite, but we can try to explore these questions by studying behaviors for increasing $M$.

\subsection{Simplifications due to symmetries}
The size of the matrix $\mathbf{C}$ can be further reduced by using time-reversal and parity symmetries.
The time-reversal operation $U_T$ corresponds to the complex conjugation in the $Z$ basis; this maps $Y_j \to U_T^{-1} Y_j U_T = -Y_j$, while leaving the other Pauli operators unchanged.
Therefore, the time-reversal-even (-odd) sector corresponds to even (odd) number of Pauli $Y$ operators in the Pauli string basis respectively.

The matrix $\mathbf{C}$ can be further simplified by utilizing the parity (i.e., mirror) symmetry with respect to the origin.
To illustrate how the parity operation $U_P$ acts on the $I$-$X$-$Y$-$Z$ Pauli-string basis, we consider an example of $S = \sum_j X_j Y_{j\!+\!1} Z_{j\!+\!2} I_{j\!+\!3} \in \mathcal{T}_4$.
Upon parity operation, $S' = U_P^{-1} S U_P = \sum_j X_{-\!j} Y_{-\!j\!-\!1} Z_{-\!j\!-\!2} = \sum_j Z_j Y_{j\!+\!1} X_{j\!+\!2}$, where in the last equality we gauge-fixed the writing of $S'$.
We see that the parity operation $U_P$ acts on the operators in $\mathcal{T}_M$ by reversing the order of operators in each of the Pauli-string basis vector and gauge-fixing the expression.
More specifically, if $S = \sum_j \sigma^{\mu_1}_j \cdots \sigma^{\mu_{r_0}}_{j+r_0-1} I_{j+r_0} \cdots I_{j+r-1} \in \mathcal{T}_r$, where $\sigma^{\mu_1}_j$ and $\sigma^{\mu_{r_0}}_{j+r_0-1}$ can only be $X$, $Y$, or $Z$, then $U_P^{-1} S U_P = \sum_j \sigma^{\mu_{r_0}}_j \sigma^{\mu_{r_0-1}}_{j+1} \cdots \sigma^{\mu_{1}}_{j+r_0-1}I_{j+r_0} \cdots I_{r-1} \in \mathcal{T}_r$.
We can therefore easily form the parity-even and -odd subspaces by forming $O \pm  U_P^{-1} O U_P$ basis vectors.

\subsection{Algorithm}
For small maximum range $M \leq 8$, we exactly diagonalize the matrix $\mathbf{C}$ to find the lowest eigenvalue and the slowest operator.
For larger maximum range $M \geq 9$, iterative methods are preferred since one can construct $\mathbf{C}$ as a sparse matrix.
While Lanczos method is one of the standard iterative algorithms to find the lowest eigenpair, the smallness of the relevant eigenvalues in the large $g$ regime makes the convergence extremely slow.
Fortunately, the positive-definite character of the matrix $\mathbf{C}$ enables us to adapt a conjugate-gradient-based algorithm.
Here, we use the ``locally optimal block preconditioned conjugate gradient method'' from Ref.~\onlinecite{knyazev_toward_2001} to find the lowest eigenpair.

\section{\label{sec:scaling} Scaling of the squared residual norm}

\begin{figure}
\includegraphics[width=1.0\columnwidth]{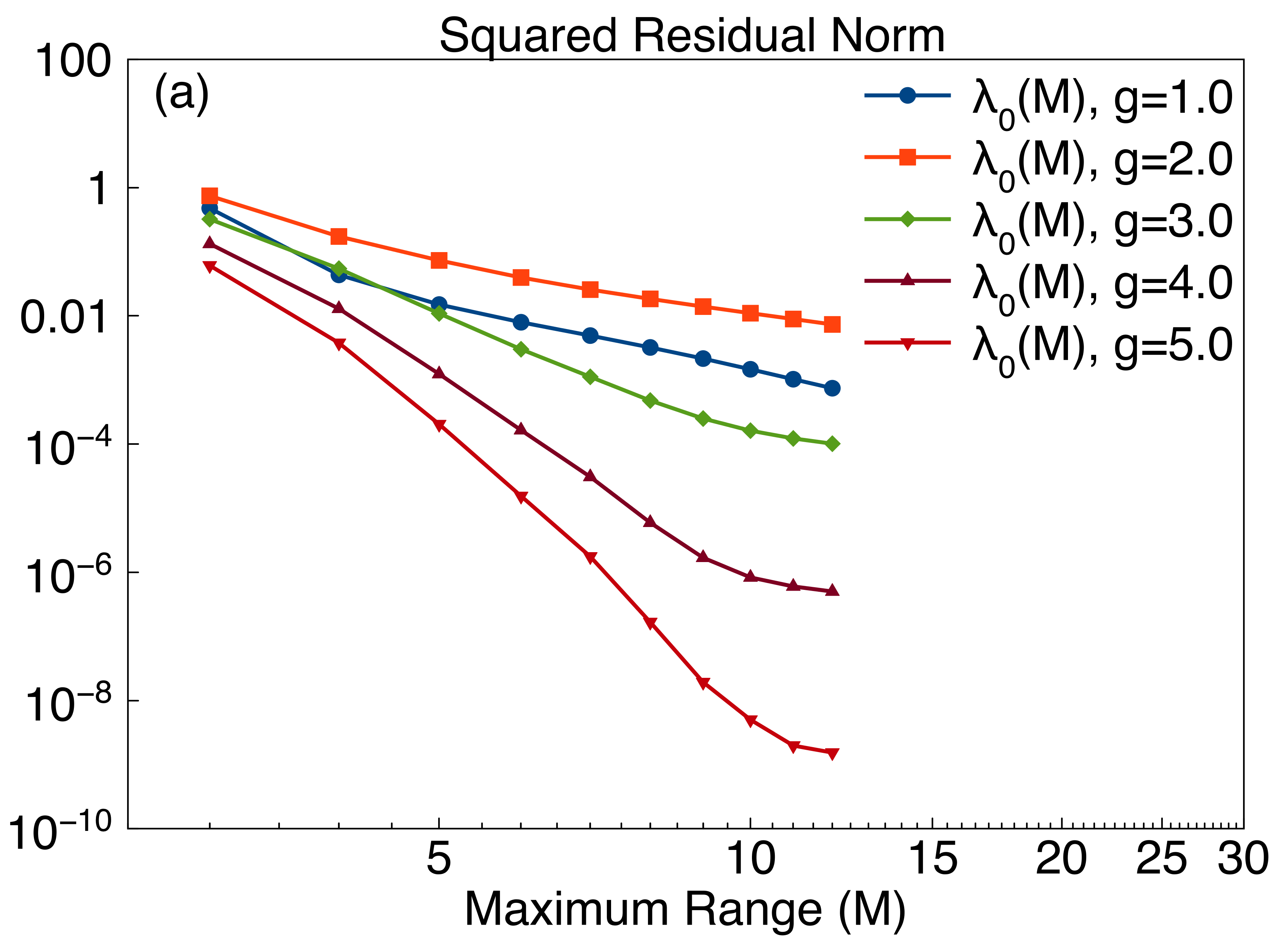}
\includegraphics[width=1.0\columnwidth]{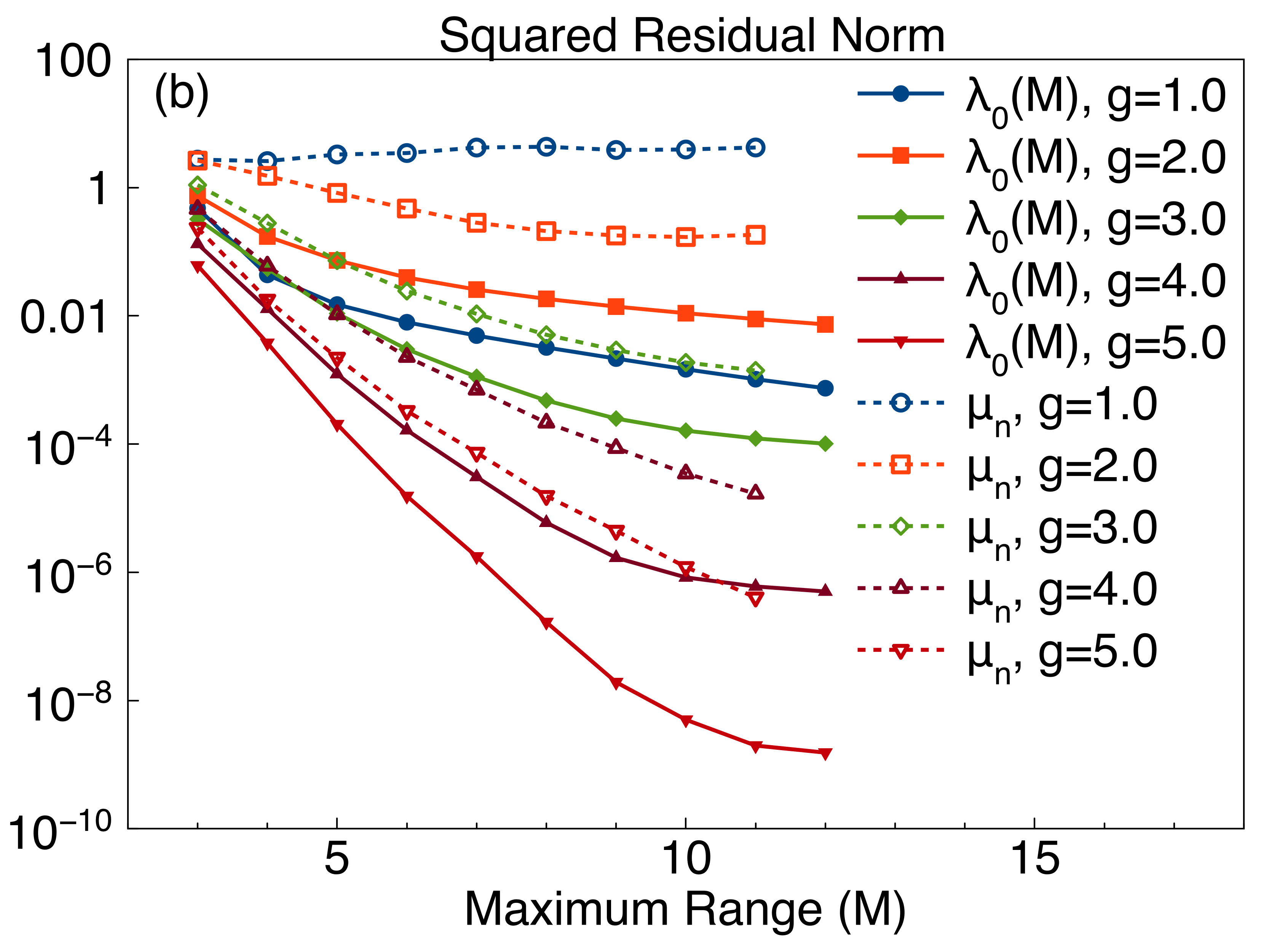}
\caption{\label{fig:scaling}
(color online)
Behavior of the squared residual norm $\lambda_0(M)$ (in units of $J^2$) vs maximum range $M$ on (a) log-log plot and (b) semi-log plot, for model parameters $J = 1.0$, $h = 1.5$, and varying $g$.
For small $g$, $\lambda_0(M)$ decays as power-law in $M$.
For large $g$, it first decays exponentially as one increases $M$, and then turns into a slower trend at larger $M$.
Panel (b) shows additional data from the Schrieffer-Wolff construction of quasi-conserved quantity (see Sec.~\ref{sec:iterative} for details), which can be viewed as a variational bound.
The residual norm $\mu_n$ from the SW construction of order $n$, which corresponds to $M \!=\! n \!+\! 1$ maximum range, shows a classic asymptotic expansion behavior for the smaller $g$ values, where it starts to increase at large order.
While this behavior is not manifest yet for the larger $g$ values, from the observed trends we suspect that $\mu_n$ will also start to increase eventually beyond some order.
}
\end{figure}

Figure~\ref{fig:scaling} shows the $M$-dependence of the squared residual norm $\lambda_0(M)$ on a log-log plot and a semilogrithmic plot.
For small $g$, the dependence is roughly a power law, which is consistent with the result in Ref.~\onlinecite{kim_slowest_2015} in the regime where the system has good ergodic behavior.
On the other hand, for large $g$, $\lambda_0(M)$ first decays exponentially with $M$ but then turns into a slower decay at larger $M$.
The exponential decay was also observed in the case of such ``slowest operator'' construction in the MBL phase.\cite{obrien_explicit_2016}
This exponential behavior differentiates the speed of the dynamics of this operator compared to other quantities.
As one increases the maximum range, one can optimize the residual norm exponentially better, which also indicates longer thermalization time scale, since the residual norm is related to the speed of the dynamics of the operator (see Sec.~\ref{subsec:Opnorm} below).
We therefore expect this quantity to be quasi-conserved, which can affect the thermalization of the system.

Interestingly, the exponential decay of $\lambda_0(M)$ for the slowest operator does not continue to larger $M$.
Instead, the decay trend seems to turn into a power law at larger $M$.
As discussed in the previous section, even though the scaling trend turns into a slower decay at large $M$, one always gets an equal or smaller residual norm as one increases $M$.
If the residual norm goes to zero as $M \to \infty$ and $\lim_{M \to \infty} \frac{\Q_0^{(M)}}{\|\Q_0^{(M)}\|_\text{F}}$ exists, then we would indeed obtain a conserved quasilocal operator.
However, due to limits on our numerical calculations, we cannot reach larger maximum range and cannot be conclusive about the behavior of $\lambda_0(M)$ at large $M$.
The eventual turn to a slower decay (similar to behavior in the good ergodic regime $g \leq 2$) may be signaling that beyond some time the operator will thermalize.
Hence, it may well be that the observed behavior corresponds to a prethermalization phenomenon on some intermediate time scales, where the time scale can be parametrically large.

\subsection{Next-slowest operators}

\begin{figure}
\includegraphics[width=1.0\columnwidth]{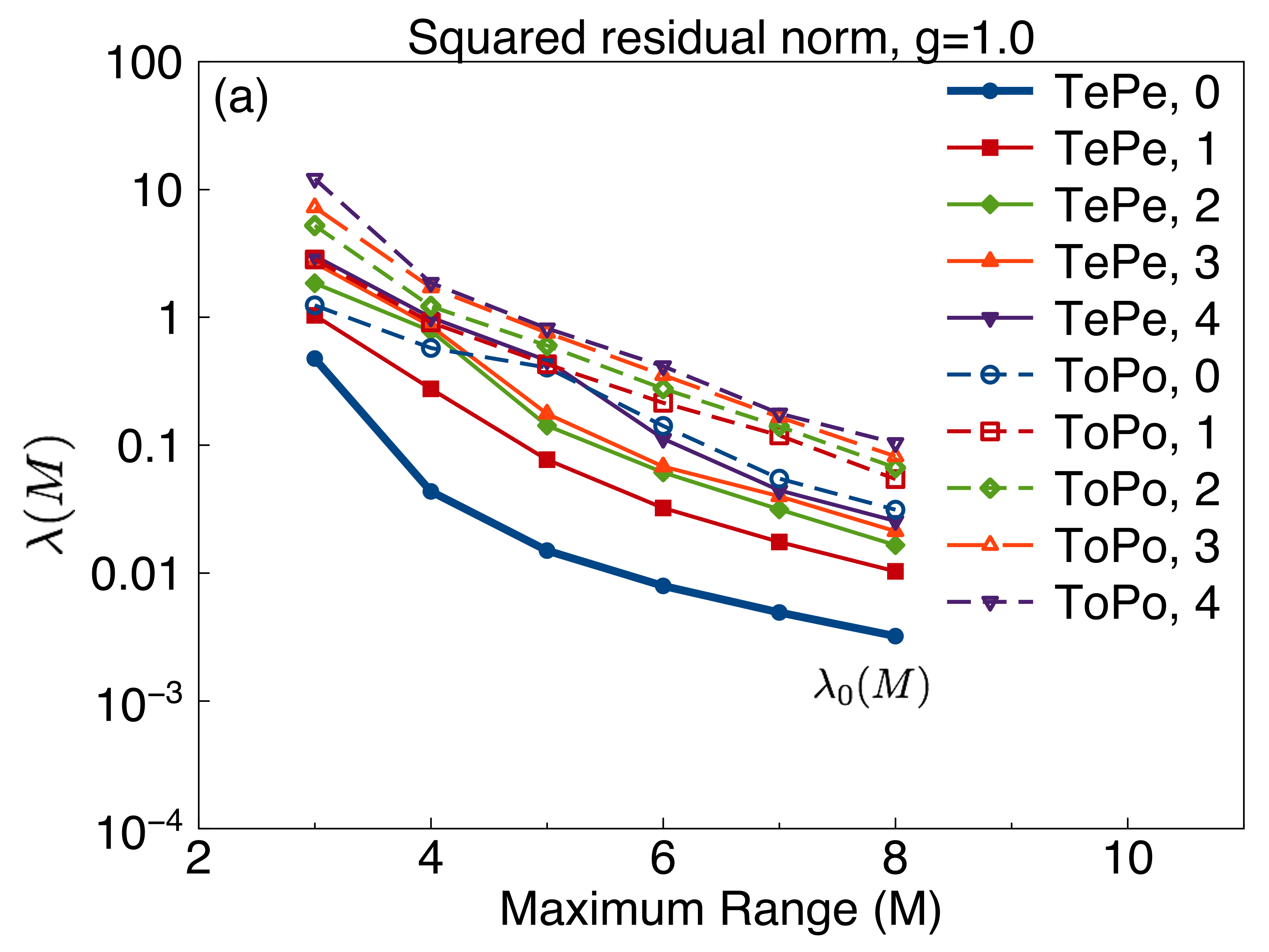}
\includegraphics[width=1.0\columnwidth]{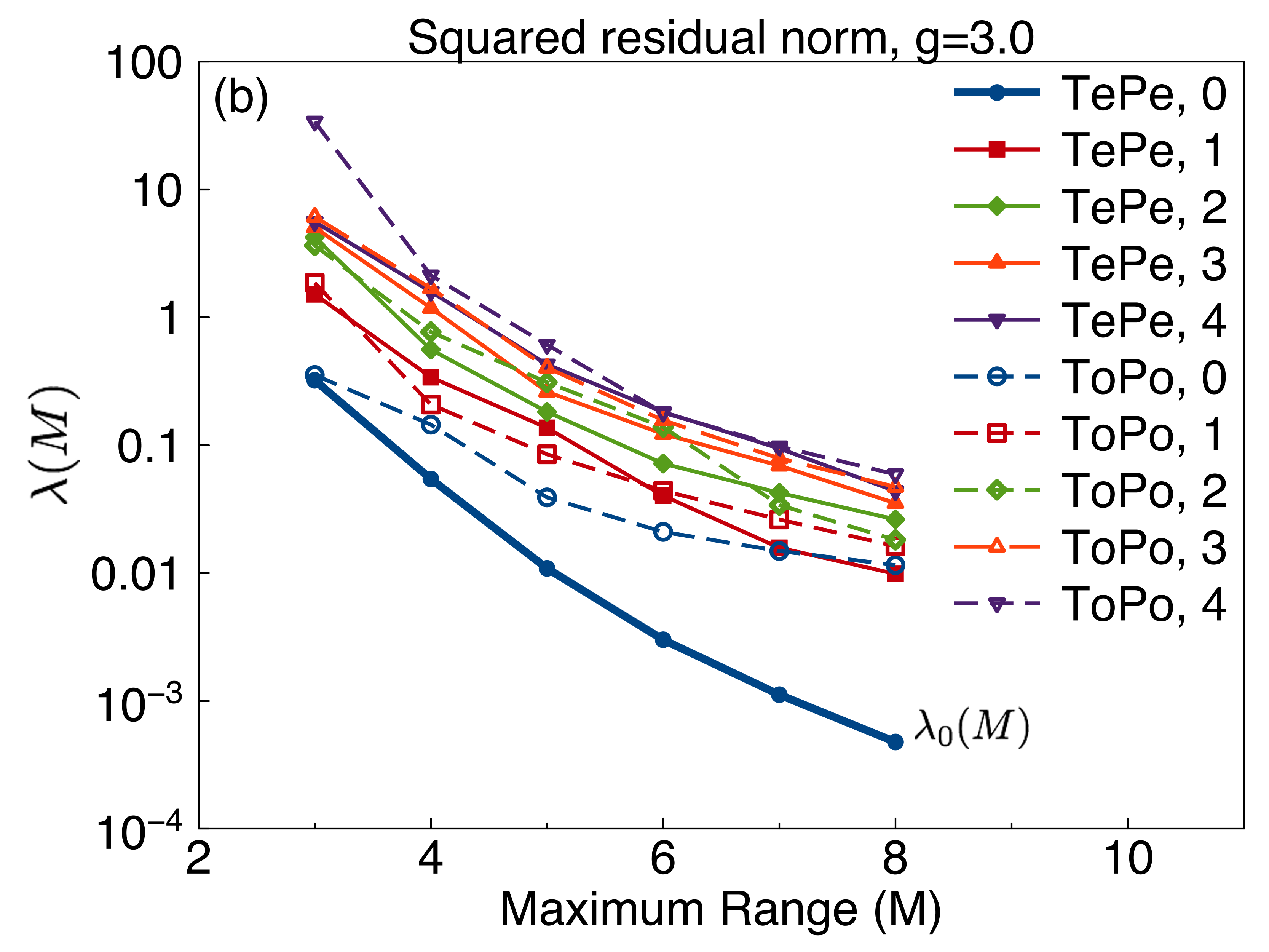}
\includegraphics[width=1.0\columnwidth]{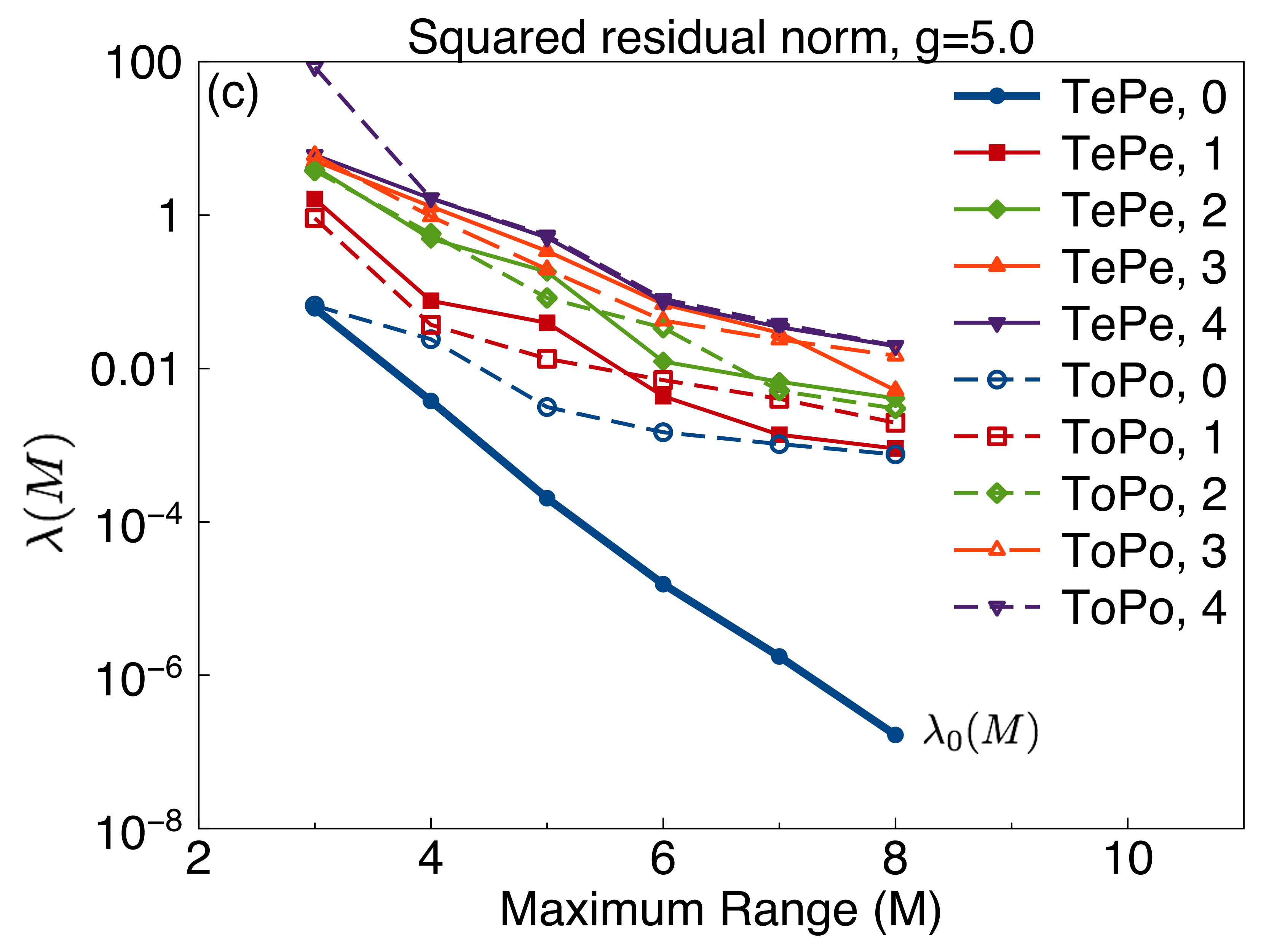}
\caption{\label{fig:higher} (color online)
Behavior of the squared residual norm $\lambda(M)$ for the first five slowest operators in the ``TePe'' and ``ToPo'' sectors.
(a) For $g = 1.0$, the slowest operator in the ``TePe'' sector shows similar dependence on $M$ as the other nearby slow operators; no particularly slow degrees of freedom exist in this case.
On the other hand, in panels (b) for $g = 3.0$ and (c) for $g = 5.0$, the slowest operator has exponential dependence on $M$ up to some range, while the other operators decrease more slowly throughout, which suggests that the slowest operator has parametrically more slow dynamics compared to other degrees of freedom.}
\end{figure}

While the exponential scaling of the slowest operator for large $g$ suggests that it is quasi-conserved, one may wonder how many quasi-conserved quantities exist.
To answer this question, we further study the scaling of the squared residual norm $\lambda(M)$ of the first five slowest operators in the time-reversal and parity even (odd) sector, denoted as ``TePe'' (``ToPo'') in Fig.~\ref{fig:higher}.
The operators in the ``TePo'' and ``ToPe'' sectors have higher squared residual norms than the ones shown in the figure and are hence less interesting and not included.
Here we only show results that are accessible using the exact diagonalization of the matrix $\mathbf{C}$, or $M \leq 8$.

Figure~\ref{fig:higher}(a) shows the scaling of $\lambda(M)$ for $g = 1.0$.
Note that the slowest operator in this case has a similar scaling trend compared to other operators.
Therefore the speed of the dynamics is not hierarchically slower than for other degrees of freedom.

On the other hand, in panels Figs.~\ref{fig:higher}(b) and ~\ref{fig:higher}(c), the slowest operator clearly has faster scaling than the next-slowest operators.
This is another feature suggesting that for large $g$, the speed of the dynamics of the slowest operator is hierarchically slower than other operators, resulting in apparent freezing of its dynamics and hence the prethermalization phenomenon.
We conclude that in these particular cases, there is only one quasi-conserved quantity.
This differs from the proposal in Ref.~\onlinecite{abanin_rigorous_2017} that there may be two independent quasi-conserved quantities (excluding the energy itself) in the strong coupling regime.
We suspect that this difference comes from our separation of operators into independent ones using the orthogonality in the Frobenius inner product.

\subsection{Relation to operator norm and thermalization time scale}
\label{subsec:Opnorm}

\begin{figure}
\includegraphics[width=1.0\columnwidth]{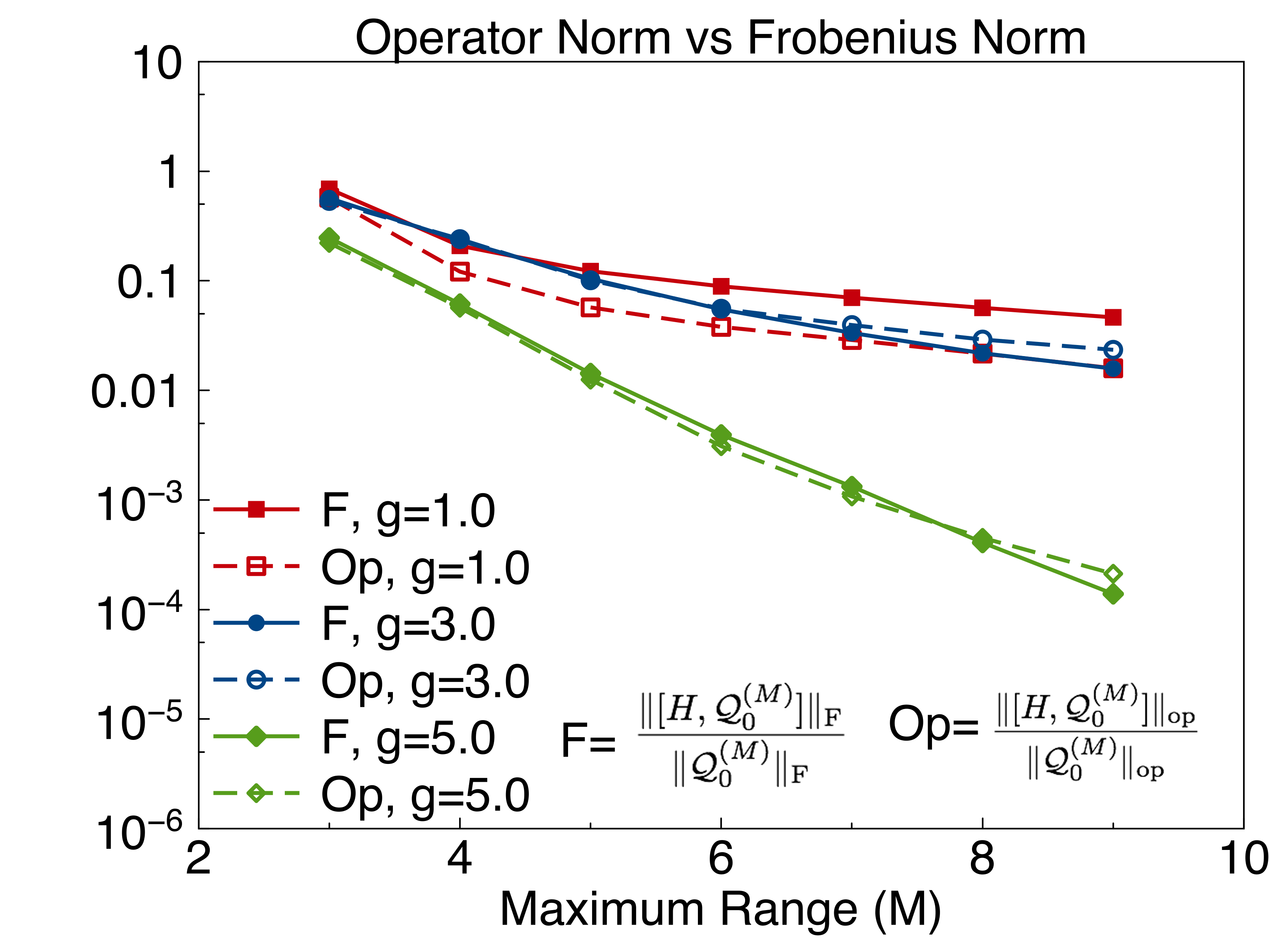}
\caption{\label{fig:opnorm}
(color online) Comparison between the residual Frobenius norm and operator norm measures of the slowest operator $\Q_0^{(M)}$; the operator is obtained from the minimization of the residual Frobenius norm as described in Sec.~\ref{sec:method}.
The inverse of $\frac{\|[H, \Q_0^{(M)}]\|_\text{op}}{\|\Q_0^{(M)}\|_\text{op}}$ gives the thermalization time scale of $\Q_0^{(M)}$.
For large coupling, cases $g = 3.0$ and $g = 5.0$, we find that the numerical values of the residual Frobenius and operator norm measures are close to each other up to some $M$ and then start deviating  (see text for some discussion).
}
\end{figure}

Minimizing the commutator $[H, \Q]$ with respect to the Frobenius norm is advantageous because it can be relatively easily calculated numerically and is independent of the system size.
On the other hand, to relate the smallness of the commutator to the dynamics, it is more appropriate to use the conventional operator norm.
Indeed, following Ref.~\onlinecite{kim_slowest_2015}, let us consider  a quench setting where we start from some initial state $|\psi_{\text{ini}} \rangle$.
Using the Heisenberg representation of observables, $\Q_H(t) \equiv e^{i H t} \Q e^{-i H t}$, and denoting the expectation value of the operator $\langle \Q_H(t) \rangle \equiv \langle \psi_{\text{ini}}| \Q_H(t) |\psi_{\text{ini}} \rangle$, the deviation of the expectation value from its initial value can be estimated as
\begin{eqnarray}
|\langle \Q_H(t) \rangle - \langle \Q \rangle| &=& \left| \left\langle \int_0^t d\tau \frac{d\Q_H}{d\tau}(\tau) \right\rangle \right| \nonumber \\
	&\leq& \int_0^t d\tau \left| \left\langle \frac{d\Q_H}{d\tau}(\tau) \right\rangle \right| \nonumber \\
&\leq& \int_0^t d\tau \| [H, \Q_H(\tau)] \|_\text{op} = t \|[H, \Q]\|_\text{op} ~, ~~~
\end{eqnarray} 
where we have used $\|[H, \Q_H(\tau)]\|_\text{op} = \|[H, \Q]\|_\text{op}$ for arbitrary $\tau$, and the above inequality holds for any initial state.
If we assume that $\Q$ has unit operator norm, we see that for $\langle \Q_H(t) \rangle$ to deviate from its initial value by an order-one number, the time scale is $t_* \sim (\|[H, \Q]\|_\text{op})^{-1}$.
For a general not normalized $\Q$, including the suitable normalization gives the time scale $t_* \sim \left(\frac{\|[H, \Q]\|_\text{op}}{\|\Q\|_\text{op}} \right)^{-1}$.

Figure~\ref{fig:opnorm} demonstrates the comparison between the Frobenius norm measure and the operator norm measure of the smallness of the commutator $[H, \Q_0^{(M)}]$, where the slowest operator $\Q_0^{(M)}$ is as before obtained by minimizing the residual Frobenius norm for given $M$.
Note that the operator norm per site of a translationally invariant operator like $\sum_{j=1}^L q_j^{(M)}$, unlike the intensive Frobenius norm defined earlier, depends on the system size $L$ and should be obtained in the thermodynamic limit (a familiar example is the ground-state energy per site of a translationally invariant Hamiltonian).
However, we expect the size dependence to diminish for increasing $L$.
We confirmed this by calculating the operator norms by diagonalizing the corresponding operators on finite systems up to size $L = 16$, and Fig.~\ref{fig:opnorm} shows our results for the largest $L$; we were able to go only up to $M=9$ because the calculations became prohibitively expensive for larger $M$.
Unlike the residual Frobenius norm, the residual operator norm $\frac{\|[H, \Q_0^{(M)}]\|_\text{op}}{\|\Q_0^{(M)}\|_\text{op}}$ can increase with $M$ since the minimization procedure is not with respect to the operator norm.
This can also potentially serve as a criterion for picking an ``optimal'' quasi-conserved operator $\Q_0^{(M_*)}$ for some $M = M_*$ that gives the minimum residual operator norm measure.
However, we do not observe a clear minimum of the residual operator norm measure for the accessible $M$.
Nevertheless, we can already bound $t_*$ from below from the $M \!=\! 9$ data.
Thus, for $g = 5.0$, we can bound $t_* > 5 \cdot 10^3$ which is already very long; while for $g = 3.0$, we can bound $t_*$ from below by approximately $t_* > 30$.

While here we were able to calculate the operator norm explicitly numerically, it is instructive to consider the following crude bound for the prethermalization condition obtained from the scaling of the residual Frobenius norm.
First, we note that we can write $[H, \Q_0^{(M)}] = \sum_j \eta_j$, where $\eta_j$ has maximum range $M + 1$.
We then have $\|[H, \Q_0^{(M)}]\|_\text{op} \leq \sum_j \|\eta_j\|_\text{op} = L \|\eta_j\|_\text{op} \leq L \, 2^{(M+1)/2} \|[H, \Q_0^{(M)}]\|_\text{F}$ (recall that here and below we use the intensive Frobenius norm).
On the other hand, for $\Q_0^{(M)} = \sum_j q_j$, heuristically we can estimate $\|\Q_0^{(M)}\|_\text{op}  \approx \sum_j \|q_j^{(M)}\|_\text{op} = L \|q_j^{(M)}\|_\text{op}$, and we also have exact bound $\|q_j^{(M)}\|_\text{op} \geq \|\Q_0^{(M)}\|_\text{F}$.
 We therefore obtain
\begin{equation}
\label{heuristicOp_Fbound}
\frac{\|[H, \Q_0^{(M)}]\|_\text{op}}{\|\Q_0^{(M)}\|_\text{op}} \leq 2^{\frac{M+1}{2}} \frac{\|[H, \Q_0^{(M)}]\|_\text{F}}{\|\Q_0^{(M)}\|_\text{F}} = 2^{\frac{M+1}{2}} \sqrt{\lambda_0(M)} ~,
\end{equation}
(which is nonrigorous bound).
To maximize the thermalization time scale, we find $\bar{M}_*$ by minimizing the right-hand side and obtain a crude criterion
\begin{equation}
\frac{d \log_{10} \lambda_0(M)}{dM}|_{M = \bar{M}_*} = -\log_{10} 2 ~.
\end{equation}
Thus, the optimal $\bar{M}_*$ from this heuristic bound is determined as the point where the magnitude of the slope of $\log_{10} \lambda_0(M)$ vs $M$ drops below value $\log_{10} 2$ (assuming that the magnitude of the slope is decreasing with $M$, as observed in Fig.~\ref{fig:scaling}).
We expect $\bar{M}_* \leq M_*$ (the latter defined from the true operator-norm minimization).

The above arguments also show how one may reconcile the fact that while the Frobenius norm measure $\lambda_0(M)$ is always decreasing with $M$, the thermalization time scale could still be finite.
The actual data for the operator norm vs Frobenius norm in Fig.~\ref{fig:opnorm} shows that the operator norm measure is numerically close to the Frobenius norm over the available maximum range $M$, particularly for large $g$.
That is, the factor of $2^{\frac{M+1}{2}}$ in the heuristic bound Eq.~(\ref{heuristicOp_Fbound}) between the two measures is an overestimate, and at least over this range of $M$ the Frobenius norm measure can be used to bound the speed of the dynamics.

We can understand the rough agreement between the Frobenius and operator norm measures if the operators $\Q_0^{(M)}$ and $[H, \Q_0^{(M)}]$ have roughly similar ``profiles'' in the operator space.
Indeed, in this case, the numerators on both sides of the  inequality in Eq.~(\ref{heuristicOp_Fbound}) and the denominators should have similar relations, which would cancel out in the ratio (while the overestimating factor $2^{\frac{M+1}{2}}$ arose from using different limits of the relations between the Frobenius and operator norms for the denominator and numerator).
We expect this to be particularly true when $\Q_0^{(M)}$ is ``localized'' in real space, which we indeed find in the strong coupling regime at least for the available $M$---see our understanding of the slowest operator from the perturbative SW picture in Sec.~\ref{sec:iterative} and direct measurements of its profile in Sec.~\ref{sec:profile}.
We do start observing some deviations between the Frobenius and operator norm measures for larger $M$, which could be indicating changing localization properties; however, the differences are still small to reach definite conclusions.

Examining carefully all data in Fig.~\ref{fig:opnorm}, we would like to point out that even though for $g = 1.0$ the operator norm measure is smaller than the one for $g = 3.0$, it does not imply that the system with $g = 1.0$ will exhibit prethermalization.
For a fair comparison of the dynamics, one also needs to compare the thermalization time scale of $\Q_0^{(M)}$ to other degrees of freedom in the same system.
We indeed know from the previous section, cf.\ Fig.~\ref{fig:higher}, that for $g = 1.0$, the next-slowest operators have comparable relaxation times to $\Q_0^{(M)}$ and the prethermalization phenomenon is less likely than for $g = 3.0$, where the slowest operator is more separated from the rest.
This could explain our findings in Sec.~\ref{sec:dynamics} of clear prethermalization at $g = 3.0$ and no prethermalization at $g = 1.0$.

While the residual norm provides us some bound on the thermalization time scale, it is also important to obtain the physical meaning of the slowest operator.
In the system in the good ergodic regime studied in Ref.~\onlinecite{kim_slowest_2015}, in the nontranslationally invariant setting, the slowest operator can be understood as dressed energy density modulation operator.
On the other hand, in the translationally invariant setting, the slowest operator does not have simple connection to the energy density modulation and its physical meaning remains an open question.
In the MBL system, Ref.~\onlinecite{obrien_explicit_2016} used this approach to explicitly construct the approximately conserved operators as local integrals of motion.
As we will show in the next Sec.~\ref{sec:iterative}, the slowest operator we found in the large $g$ regime can be understood as a dressed total spin-$z$ operator, coming from the solvable limit $H_0 = \sum_j (g X_j + h Z_j)$, which can be viewed as a quasi-local integral of motion.

\section{\label{sec:iterative} Schrieffer-Wolff Construction of Quasi-Conserved Quantity}
Reference~\onlinecite{abanin_rigorous_2017} used a renormalization scheme to construct an effective Hamiltonian which commutes with $H_0$ up to some order in small parameter, which can then be used to describe the prethermalization dynamics.
Here, we use an approach with similar spirit but based on the {\emph local} Schrieffer-Wolff (SW) transformation\cite{datta_low-temperature_1996, bravyi_schriefferwolff_2011} to construct a quasi-conserved operator perturbatively.
The term ``local'' is stressed since the generators are solved in the form of sum of local terms, in contrast with the ``global'' SW transformation, where the generators are solved using projectors of the $H_0$ eigenspaces.\cite{bravyi_schriefferwolff_2011}
The locality in particular allows us to construct the quasi-conserved quantity numerically to high order and measure its properties exactly, in contrast to the more abstract construction in Ref.~\onlinecite{abanin_rigorous_2017}.
A popular variant of a local SW transformation was in fact proposed in Ref.~\onlinecite{macdonald_$fractu$_1988} as a perturbative treatment of the Hubbard model in the large $U$ limit; this reference used generalized ``ladder'' operators connecting different Hubbard sectors, and we discuss the relation to our approach in App.~\ref{app:ladderformalism}.
Before proceeding, we briefly point some differences with Ref.~\onlinecite{bravyi_schriefferwolff_2011}.
First, our setup works in the thermodynamic limit $L \to \infty$ from the start.
More importantly, we choose the solution of Eq.~(\ref{eqn:solSm}) for the generator that eliminates the off-diagonal part of $V_m$ among all the sectors, while in Ref.~\onlinecite{bravyi_schriefferwolff_2011} one is only focusing on the off-diagonal part between the ground-state sector and other sectors.

We first describe the specific SW transformation used here and how we numerically construct a perturbation series for a quasi-conserved operator $\tilde{I}^{(n)}$ to $n$-th order.
We then calculate the squared residual norm of $\tilde{I}^{(n)}$ and the overlap between $\Q_0^{(M)}$ and $\tilde{I}^{(n)}$ to demonstrate the similarity between the two operators.
We will see that the slowest operator $\Q_0^{(M)}$ in the large $g$ regime can be understood---at least up to the maximum range accessible in our work---as $\tilde{I}^{(n)}$, which is essentially dressed ``total spin-$z$ operator.''

\subsection{Procedure of SW transformation}
\label{subsec:SWproc}
In the large-$g$ limit, we can decompose $H = H_0 + \epsilon T$, with $H_0 = \sum_j (g X_j + h Z_j)$ being our solvable limit and $\epsilon T = J \sum_j Z_j Z_{j+1}$ treated as perturbation with small parameter $\epsilon$.
[For example, we can define $\epsilon \equiv J/\sqrt{g^2 + h^2}$ so that for convenience $\|T\|_F = \|H_0\|_F$ in the intensive Frobenius norm, but the specific choice is not important.]
We construct a unitary transformation $U = e^{-i \epsilon S_1} e^{-i \epsilon^2 S_2} \dots e^{-i \epsilon^n S_n}$, with $S_m$ being Hermitian and $\epsilon$-independent, such that the rotated Hamiltonian $H' \equiv U^\dagger (H_0 + \epsilon T) U$ commutes with $H_0$ up to order $n$ in the formal expansion in $\epsilon$.
Stated another way, the eigenvalues of $H_0$ define the corresponding unperturbed sectors, and we want $H'$ to have only sector-diagonal terms up to order $n$ in $\epsilon$, while sector-off-diagonal terms are present only in higher order.
If we then undo the rotation on $H_0$ back to the original picture, i.e., perform the inverse rotation to define $I \equiv U H_0 U^\dagger$, we obtain an operator that commutes with $H$ up to order $n$ by construction.

To be more specific, we follow Ref.~\onlinecite{datta_low-temperature_1996} and consider an expansion of $H'$ in powers of $\epsilon$:
\begin{equation}
\label{Hprime_expansion}
H' = H_0 + \sum_{m=1}^n \epsilon^m [i\ad_{S_m}(H_0) + V_m] + H_{>n} ~,
\end{equation}
where $V_1 \equiv T$ and
\begin{eqnarray}
\label{eqn:Vm}
V_m &=& \sum_{p=2}^m \sum_{[k_1, \dots, k_p] = m} \!\!\! \mathfrak{f}(k_1, \dots, k_p) \, i\ad_{S_{k_p}} \dots i\ad_{S_{k_1}}(H_0) \nonumber \\
&+& \sum_{p=1}^{m-1} \sum_{[k_1, \dots, k_p] = m-1} \!\!\!\! \mathfrak{f}(k_1, \dots, k_p) \, i\ad_{S_{k_p}} \dots i\ad_{S_{k_1}}(T) ~~~~~
\end{eqnarray}
for $m \geq 2$.
Here we have used the notation ``$[k_1, \dots, k_p] = m$'' to mean the summation conditions $1 \leq k_i \leq n$ for $i = 1, \dots, p$ and $k_1 + \dots + k_p = m$, while the function $\mathfrak{f}(k_1, \dots, k_p) = \Theta(1 \!\leq\! k_1 \!\leq\! \dots \!\leq\! k_p \!\leq\! n)/[\prod_{l=1}^n \text{card}(l)!]$, where $\Theta(\bullet) = 1$ if the condition in the argument is true and $\Theta(\bullet) = 0$ otherwise, and $\text{card}(l)$ counts the number of elements in $\{k_1, \dots, k_p\}$ that are equal to $l$.
By construction, each $V_m$ is $\epsilon$-independent; it enters with a coefficient $\epsilon^m$ and is part of the $m$-th term in Eq.~(\ref{Hprime_expansion}) for $m = 1, \dots, n$.
Furthermore, $H_{>n} = \sum_{m=n+1}^\infty \epsilon^m V_m$ collects all the terms with $\epsilon$ powers higher than $n$.

The generators of the SW transformation are solved order by order by finding $i S_m$ such that 
\begin{equation}
\label{eqn:eqnforSm}
i\ad_{S_m}(H_0) + V_m = V_m^\text{diag} ~,
\end{equation}
where we have defined $O^\text{diag}$ as a part of an operator $O$ that is ``diagonal'' in the $H_0$ sector label; i.e., $O^\text{diag}$ is the component of the operator that commutes with $H_0$.
Equivalently, $O^\text{diag}$ is the component of $O$ in the kernel (nullspace) of $\ad_{H_0}$.
The remainder $O^\text{off-diag} \equiv O - O^\text{diag}$ is the ``off-diagonal'' part of the operator, and can be also viewed as a component of $O$ orthogonal to the kernel of $\ad_{H_0}$ in the Frobenius inner product\cite{datta_low-temperature_1996}.
We can solve for the generator
\begin{equation}
\label{eqn:solSm}
i S_m = [\ad_{H_0}]^{-1} V_m^\text{off-diag} ~,
\end{equation}
where $[\ad_{H_0}]^{-1}$ is the {\it pseudoinverse} of $\ad_{H_0}$.
Note that $i S_m$ solving Eq.~(\ref{eqn:eqnforSm}) is determined only up to a component in the kernel of $\ad_{H_0}$, and we make a choice here where such component is zero, i.e., $i S_m$ is composed of only sector-off-diagonal operators; this is common choice in the SW approach, cf.~Refs.~\onlinecite{macdonald_$fractu$_1988, datta_low-temperature_1996, bravyi_schriefferwolff_2011}.
The described procedure generates an effective Hamiltonian which commutes with $H_0$ up to order $n$ by truncating out $H_{>n}$, obtaining $H^{(n)}_\text{eff} = H_0 + \sum_{m=1}^n \epsilon^m V_m^\text{diag}$.

An important property of the above SW transformation is its locality, which ensures the representability of $S_m$ and $V_m$ in finite-dimensional operator spaces, making the SW procedure programmable as operations of matrices and vectors.
In fact, one can show that for $H_0 \in \mathcal{T}_1$ and $T \in \mathcal{T}_2$ we have $V_m \in \mathcal{T}_{m+1}$ and $S_m \in \mathcal{T}_{m+1}$, see Ref.~\onlinecite{bravyi_schriefferwolff_2011} and Proposition~\ref{prop:locality} in App.~\ref{app:boundHn}.

We remark that the SW transformation generally does not converge when one takes the $n \to \infty$ limit.
There are rigorous results for the convergence of the ground state energy estimates for gapped Hamiltonians\cite{datta_low-temperature_1996, bravyi_schriefferwolff_2011} but no known results for the ability of the SW procedure to capture the entire spectrum of interest here.
Nevertheless, the SW transformation is well-defined for any finite $n$ and can be used to obtain rigorous bounds on the dynamics in the spirit of Refs.~\onlinecite{mori_rigorous_2016, kuwahara_floquet-magnus_2016, abanin_rigorous_2017}.
Thus, one can show that, for small enough $\epsilon$, $\|H_{>n}\|_\text{F} < O(n^{2n+2} \epsilon^{n+1})$, see Ref.~\onlinecite{bravyi_schriefferwolff_2011} and Theorem~\ref{thm:boundHn} in App.~\ref{app:boundHn}.
The dynamics described by $H' = H_\text{eff}^{(n)} + H_{>n}$ in the rotated picture does not truly conserve $H_0$ but only approximately.
In other words, while $H_\text{eff}^{(n)}$ conserves $H_0$, the ``remainder'' $H_{>n}$ does not and is responsible for the eventual thermalization of the dynamics, which can be very slow if $\epsilon$ is small.

We can thus intuitively understand the prethermalization via this perturbative SW construction.\cite{mori_rigorous_2016, kuwahara_floquet-magnus_2016, abanin_rigorous_2017, else_prethermal_2017}
The solvable limit $H_0$ defines different sectors labeled by different integers, which can be viewed as counting the number (up to some off-set) of some emergent ``particles.'' (see also App.~\ref{app:ladderformalism}).
The perturbation term $\epsilon T$ introduces interactions within the sectors and transitions between the sectors.
The interactions within the sectors are indeed the ``diagonal'' part of $T$.
At $m$-th order, the coefficient $\epsilon^m$ in the SW perturbation theory basically describes the transition amplitude of any process with $m$ inter-sector transitions.
The generator $i S_m$ is set to rotate the picture such that these processes are eliminated.
The remaining part $V_m^\text{diag}$ basically describes the processes which start and end in the same sector connected by $m$ times of the inter-sector transitions.
The perturbation series would be convergent for small enough $\epsilon$ if there were at most $\mathcal{O}(e^{c m})$ of such processes.
However, generically, in a translationally invariant system, there are order $\mathcal{O}(m^{\gamma m})$ such processes coming from combinatorial factorials in $m$.
The exponential suppression of the transition amplitude is then not enough to suppress the factorial factor.
Therefore, even though at high order of $n$, the transition amplitude is perturbatively small $O(\epsilon^n)$, manifesting slowness of individual processes, there are, however, too many ways of the transitions $O(n^{\gamma n})$ such that the system will eventually thermalize. 

\subsection{Quasi-conserved quantity by SW transformation}
\label{subsec:quasiIOviaSW}
Once we have obtained the generators for the SW transformation, we can rotate $H_0$ back to the original picture and obtain the quasi-conserved operator.
Consider
\begin{equation}
I \equiv U H_0 U^\dagger = H_0 + \sum_{m=1}^n \epsilon^m I_m + I_{>n} ~,
\end{equation}
where 
\begin{equation}
\label{eqn:solIm}
I_m = \sum_{p=1}^m \, \sum_{[k_1, \dots, k_p] = m} \!\!\! (-1)^p \, \mathfrak{f}(k_1, \dots, k_p) \, i\ad_{S_{k_1}} \dots i\ad_{S_{k_p}}(H_0)
\end{equation}
and $I_{>n} = \sum_{m = n + 1}^\infty \epsilon^m I_m$ collects all the higher-power in $\epsilon$ terms.
We then obtain the quasi-conserved operator $I^{(n)} = H_0 + \sum_{m=1}^n \epsilon^m I_m$.
In Appendix~\ref{app:boundHn}, we show that $I^{(n)} \in \mathcal{T}_{n+1}$.
To compare with the slowest operator, we remove the part of $I^{(n)}$ that is parallel to $H$ and normalize the resulting operator:
\begin{eqnarray}
\label{eqn:I(n)perp}
I^{(n)\perp} &=& I^{(n)} - H \frac{\langle H, I^{(n)} \rangle}{\|H\|_\text{F}^2} ~, \\
\label{eqn:I(n)tilde}
\tilde{I}^{(n)} &=& \frac{I^{(n)\perp}}{\|I^{(n)\perp}\|_\text{F}} ~.
\end{eqnarray}
For small enough $\epsilon$, we can bound the squared residual norm as
\begin{equation}
\mu_n \equiv \|\ad_H(\tilde{I}^{(n)}) \|_\text{F}^2 \leq \mathcal{O}(n^{4n} \epsilon^{2n}) ~.
\end{equation}
The proof of this bound and a more precise statement is in Appendix~\ref{app:boundadHI}.

Applying the previous heuristic argument for the thermalization time scale, Eq.~(\ref{heuristicOp_Fbound}), we get $t_*^{-1} \sim \mathcal{O}((2\epsilon)^n n^{2n})$.
If we treat the perturbation strength $\epsilon$ as given, and the SW order $n$ as an optimization parameter, then we can find that the residual operator norm is minimized at $n = n_* = 1/(e\sqrt{2\epsilon})$. 
The thermalization time scale is therefore maximized as $t_* = \mathcal{O}(\exp(\frac{\sqrt{2}}{e\sqrt{\epsilon}}))$.
Note that unlike Refs.~\onlinecite{ mori_rigorous_2016, kuwahara_floquet-magnus_2016,abanin_effective_2017,abanin_rigorous_2017}, where the heating rate is proven to be $\mathcal{O}(\exp(\frac{A}{\epsilon}))$, we only obtain $\mathcal{O}(\exp(\frac{A'}{\sqrt{\epsilon}}))$.
This can be traced back to the estimation of the convergence radius in Appendices.~\ref{app:boundHn} and \ref{app:boundadHI} to be $\rho_n \sim 1/n^2$, hence the squared residual norm $\mu_n \sim \mathcal{O}(n^{4n} \epsilon^{2n})$.
We suspect that a tighter convergence radius $\rho_n \sim 1/n$ is possible (see App.~\ref{app:conv_radius}); hence the bound on the thermalization time-scale could be improved to $\mathcal{O}(\exp(\frac{A}{\epsilon}))$.\cite{ftnote}
Without pursuing this tighter bound further, we leave this for future studies.

As mentioned earlier, the locality of $i S_m$ and $V_m$ allows us to formulate this procedure in finite-dimensional operator Hilbert spaces amenable to numerical calculations.
Figure~\ref{fig:scaling}(b) shows the squared residual norm calculated from such SW construction of the quasi-conserved operator for several values of parameter $g$.
Note that at order $n$, the constructed operator has maximum range $M = n + 1$.
The trend of $\mu_n$ at large $g$ more or less follows the trend of $\lambda_0(M)$, where the residual norm drops almost exponentially in low order, and turns into a slower trend, which is possibly a manifestation of the combinatorial factor $\mathcal{O}(n^{\gamma n})$.
While not appearing in the figure yet for large $g$, we expect $\mu_n$ will eventually start increasing at high enough order $n$; this is because in generic systems the combinatorial factors (like the ones appearing in the previous paragraph) will win over the exponential suppression at large enough $n$; such behavior of $\mu_n$ is observed in the $g = 1$ and $g = 2$ cases. 
Nevertheless, noting that the above arguments are based on the ``worst-case-scenario'' analytical bounds on the perturbatively-constructed operators, our numerical results for $\mu_n$ in the larger $g$ cases do not rule out the possibility that $\mu_n \rightarrow 0$.
On the other hand, unlike the perturbative construction, the numerical minimization for the slowest operator is guaranteed to get an equal or smaller residual norm when increasing $M$.

\begin{figure}
\includegraphics[width=1.0\columnwidth]{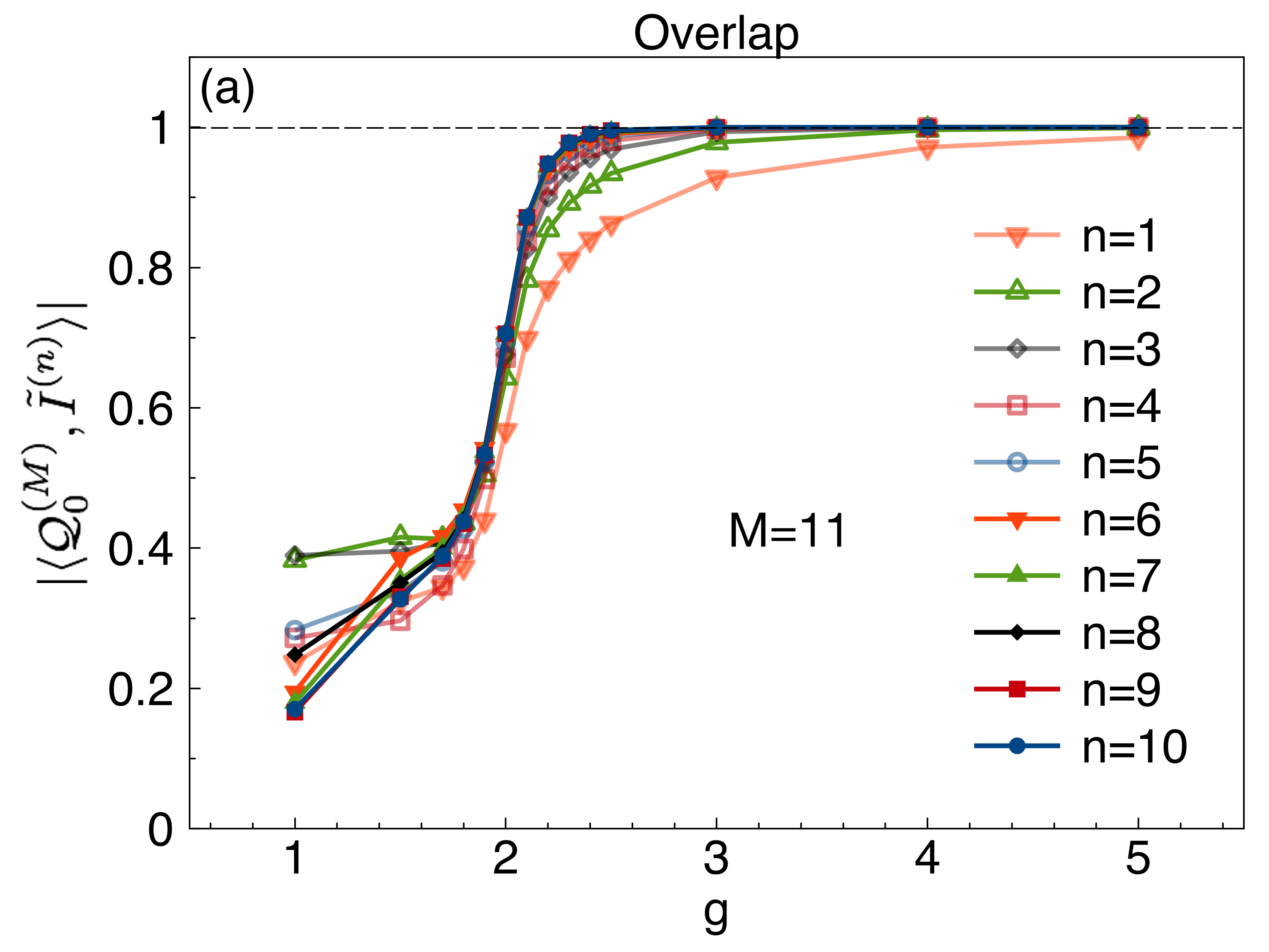}
\includegraphics[width=1.0\columnwidth]{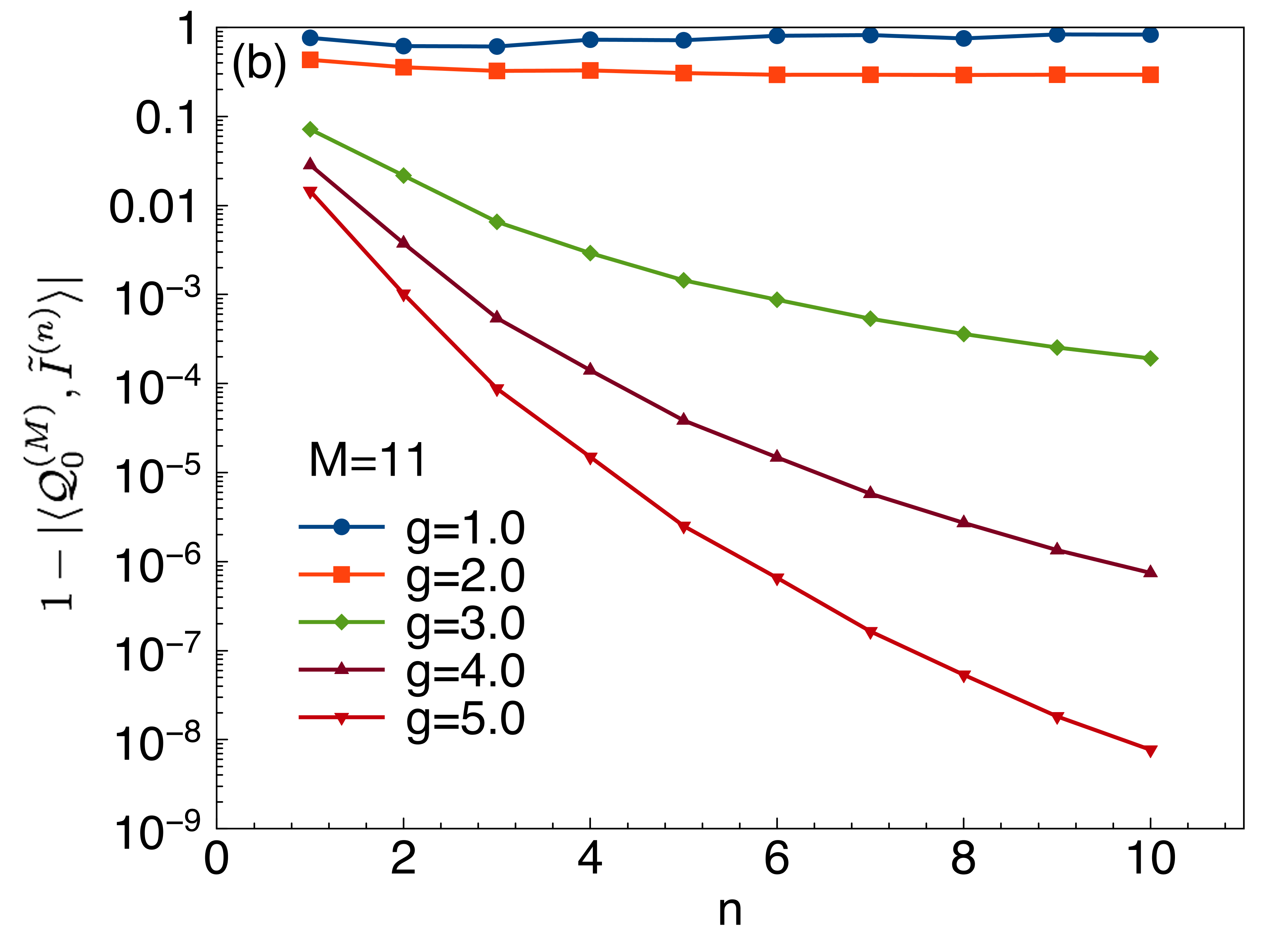}
\caption{\label{fig:overlap}
(color online) (a) The overlap between the full numerical optimization $\Q_0^{(M)}$ with $M = 11$ and the perturbative SW construction $\tilde{I}^{(n)}$ with order $n = 1$ to $10$.
(b) One minus the overlap on the log-linear plot.
At large $g$, the overlap between the two operators is almost $100\%$, which means that the slowest operator we found is essentially the dressed spin operator coming from the solvable limit $H_0$.
On the other hand, for small $g$, the slowest operator does not look like the perturbative SW construction operator anymore.
Interestingly, there is apparently a strong change in behavior around $g_c \approx 2$; however, we do not know if there is a true transition.}
\end{figure}

Figure~\ref{fig:overlap}(a) shows the overlap between the slowest operator $\Q_0^{(M)}$ with maximum range $M = 11$ and the SW construction $\tilde{I}^{(n)}$ with order $n$ up to $10$.
The overlap at large $g$ is almost $100\%$!
Accordingly, we can understand the slowest operator we found in the large $g$ limit as the translationally invariant sum of the dressed spin-$z$ operator, or the dressed $H_0$.
Interestingly, there appears to be a strong change in behavior at $g_c \approx 2$. 
For $g > g_c$, the slowest operator looks like the dressed spin-$z$ operator, with an exponential scaling of the residual norm for small $M$; on the other hand, for $g < g_c$, the slowest operator does not look like the dressed spin-$z$ operator, and its residual norm has a power-law scaling.

Note that despite the fact that the SW construction $\tilde{I}^{(n)}$ and the slowest operator $\Q_0^{(M)}$ have very high overlap $1 - \alpha$, where $\alpha$ can be a very small number as shown in Fig.~\ref{fig:overlap}(b), the difference between their squared residual norms can still be sizable.
Indeed, consider $\tilde{I}^{(n)} = (1 - \alpha) \Q_0^{(M)} + \beta \eta$, where $\|\Q_0^{(M)}\|_\text{F} = \|\eta\|_\text{F} = 1$ and $\eta$ is some operator perpendicular to $\Q_0^{(M)}$ in the Frobenius inner product.
The normalization condition of $\tilde{I}^{(n)}$ gives $\beta^2 = 2\alpha - \alpha^2$, hence $\beta = \mathcal{O}(\sqrt{\alpha})$.
The squared residual norm of $\tilde{I}^{(n)}$ is $\|\ad_H(\tilde{I}^{(n)})\|_\text{F}^2 = (1 - \alpha)^2 \|\ad_H(\Q_0^{(M)})\|_\text{F}^2 + \beta^2 \|\ad_H(\eta)\|_\text{F}^2 + 2 \beta (1 - \alpha) {\rm Re}[\langle \ad_H(\Q_0^{(M)}), \ad_H(\eta) \rangle]$.
We can thus see that
\begin{eqnarray*}
&& \|\ad_H(\tilde{I}^{(n)})\|_\text{F}^2 - \|\ad_H(\Q_0^{(M)})\|_\text{F}^2 \approx \\
&& \approx 2\alpha \|\ad_H(\eta)\|_\text{F}^2 + 2\sqrt{2\alpha} {\rm Re}[\langle \ad_H(\Q_0^{(M)}), \ad_H(\eta) \rangle] ~,
\end{eqnarray*}
where we expressed everything in terms of the small number $\alpha$ and kept only terms that are expected to dominate.
Note that while $\|\ad_H(\Q_0^{(M)})\|_\text{F}$ is a small number, no such smallness is expected for $\|\ad_H(\eta)\|_\text{F}$ since the deviation direction $\eta$ is not special in any way.
Since $\|\eta\|_\text{F} = 1$, we expect that $\|\ad_H(\eta)\|_\text{F}$ is a number of order $1$ in the energy units of $H$ (and could be larger depending on the range of typical terms in $\eta$), which could be sufficient to explain the visible difference between the two residual norms in Fig.~\ref{fig:scaling}(b) despite the high overlap between $\tilde{I}^{(n)}$ and $\Q_0^{(M)}$.

\section{\label{sec:characteristic} Characterizing the Slowest Operators}
In this section, we analyze some properties of the quasi-conserved operators that we found in Sec.~\ref{sec:method}.
We measure their ``locality'' in the operator space and in the real space, to contrast different behaviors of the slowest operators between small $g$ and large $g$ regimes.

\subsection{\label{sec:OIPR} Operator inverse participation ratio}
From the previous section, we expect that for large $g$ the quasi-conserved operator looks like a dressed spin operator.
It is therefore reasonable to expect that $\Q_0^{(M)}$ should be a sum of a small number of Pauli string operators, analogous to the local integrals of motion in MBL studies.~\cite{obrien_explicit_2016}
Using the Pauli string basis $I$, $X$, $Y$, $Z$ (without forming the parity-invariant basis), we measure the operator inverse participation ratio (OIPR)
\footnote{Here we call the quantity in Eq.~(\ref{eqn:OIPR}) ``operator inverse participation ratio'' so that it is consistent with usual definition, e.g., as used in single-particle localization problems where for a normalized wavefunction $\psi(x)$ the inverse participation ratio is $1/\sum_x |\psi(x)|^4$; this convention is different from that in Ref.~\onlinecite{obrien_explicit_2016}.}
defined as
\begin{equation}
\label{eqn:OIPR}
{\rm OIPR}(\Q_0^{(M)}) = \left(\sum_{i=1}^{3 \cdot 4^{M-1}} |a_i|^4 \right)^{-1} ~,
\end{equation}
where $a_i$'s are the amplitude of the $I$-$X$-$Y$-$Z$ Pauli-string basis and we assumed normalization $\sum_{i=1}^{3 \cdot 4^{M-1}} |a_i|^2 = 1$.
The OIPR is bounded from below by $1$.

\begin{figure}
\includegraphics[width=1.0\columnwidth]{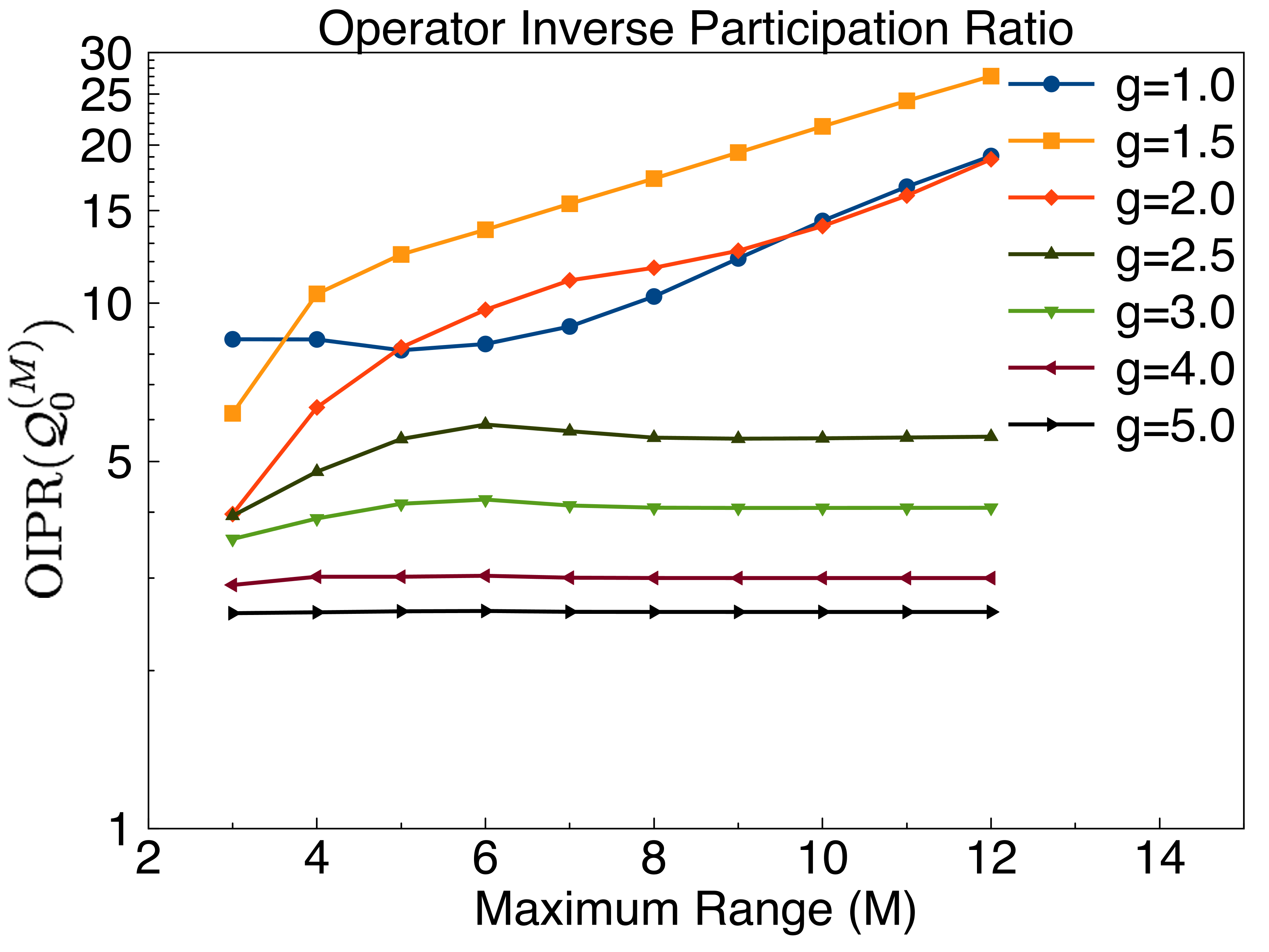}
\caption{\label{fig:OIPR}
(color online) Operator inverse participation ratio of the slowest operator vs maximum range $M$ for different $g$.
For large $g \gtrsim 2$, the OIPR appears to converge to a finite value, which suggests its locality in the operator space.
On the other hand, in the ergodic regime, $g \lesssim 2$, the OIPR does not converge and instead grows strongly with $M$ (the behavior on the linear-log plot suggests exponential growth).}
\end{figure}

Figure~\ref{fig:OIPR} shows the OIPR of the slowest operator $\Q_0^{(M)}$ for different $g$. 
Interestingly, for larger $g \gtrsim 2$, the OIPR seems to converge to a finite value at large enough $M$.
This behavior is consistent with our expectation that the quasi-conserved operator is a dressed total spin operator.
The convergence of the OIPR indicates locality in the operator space.
On the other hand, for small $g \lesssim 2$, the OIPR does not saturate but instead grows strongly with $M$.
This suggests that the slowest operators we found in the ergodic regime are composed of an extensive number of the Pauli string basis states; hence they are ``delocalized'' in the operator space.

\subsection{\label{sec:profile} Real-space profile of the slowest operator}

\begin{figure*}
\includegraphics[width=0.65\columnwidth]{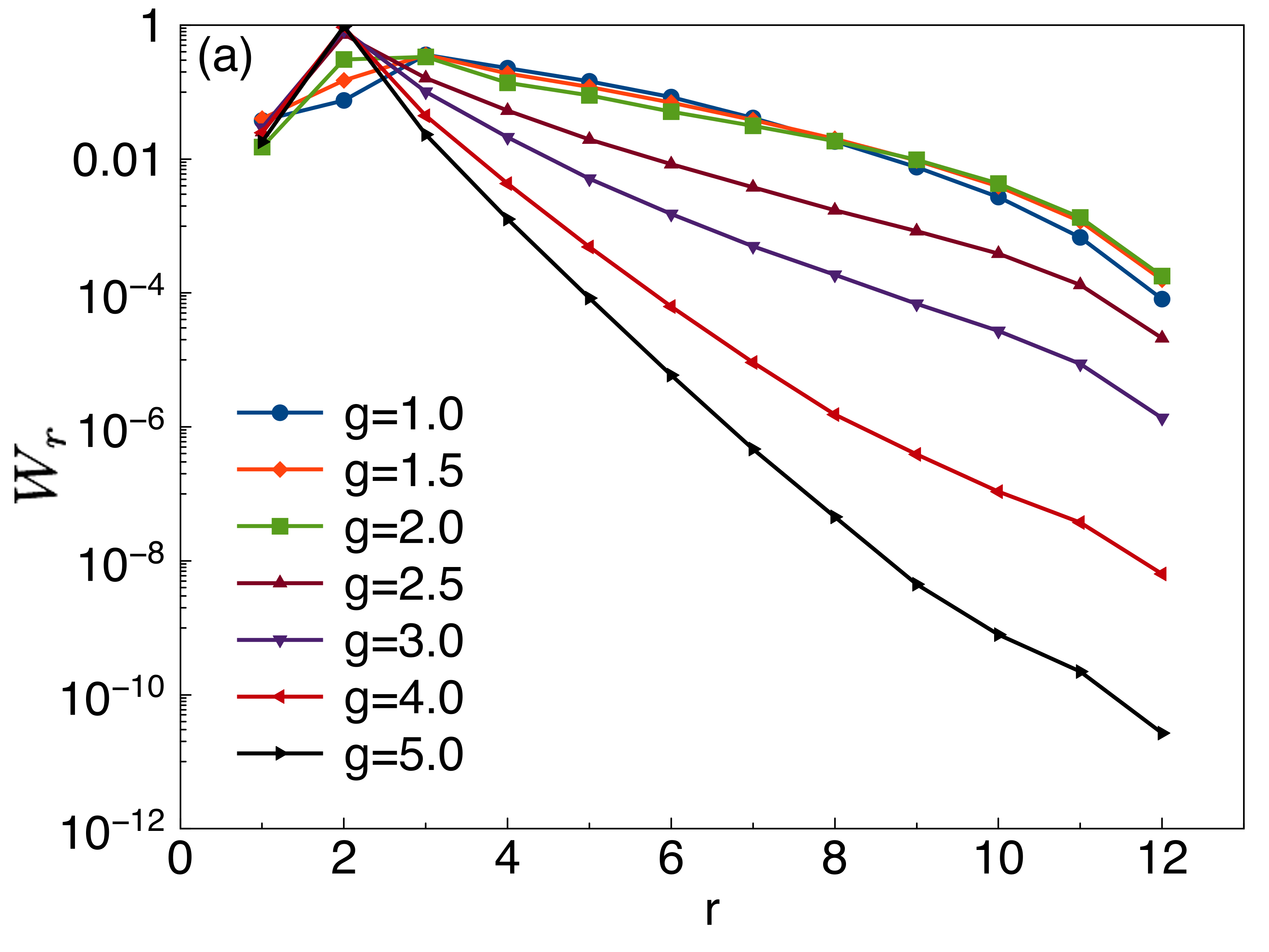}
\includegraphics[width=0.65\columnwidth]{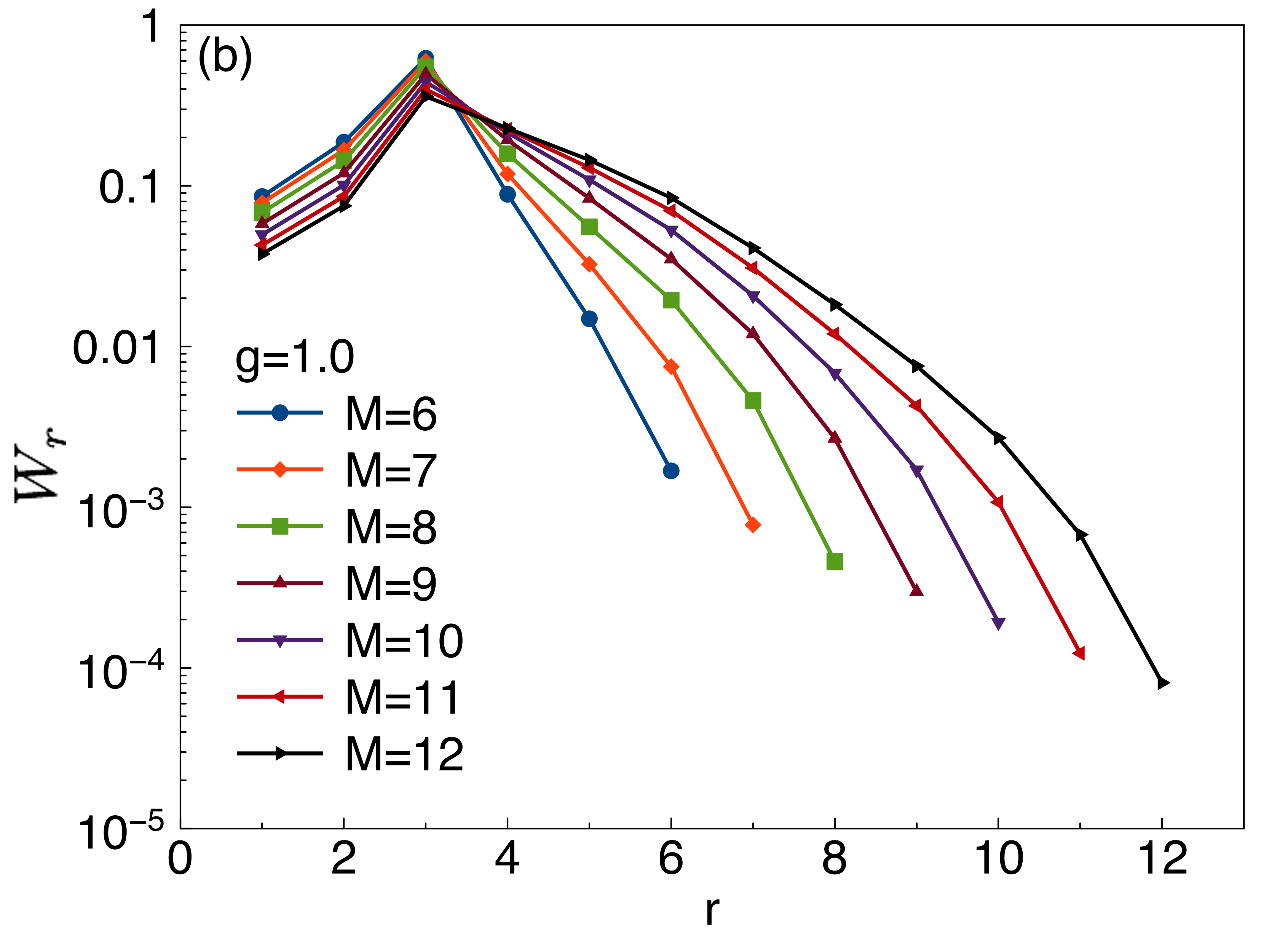}
\includegraphics[width=0.65\columnwidth]{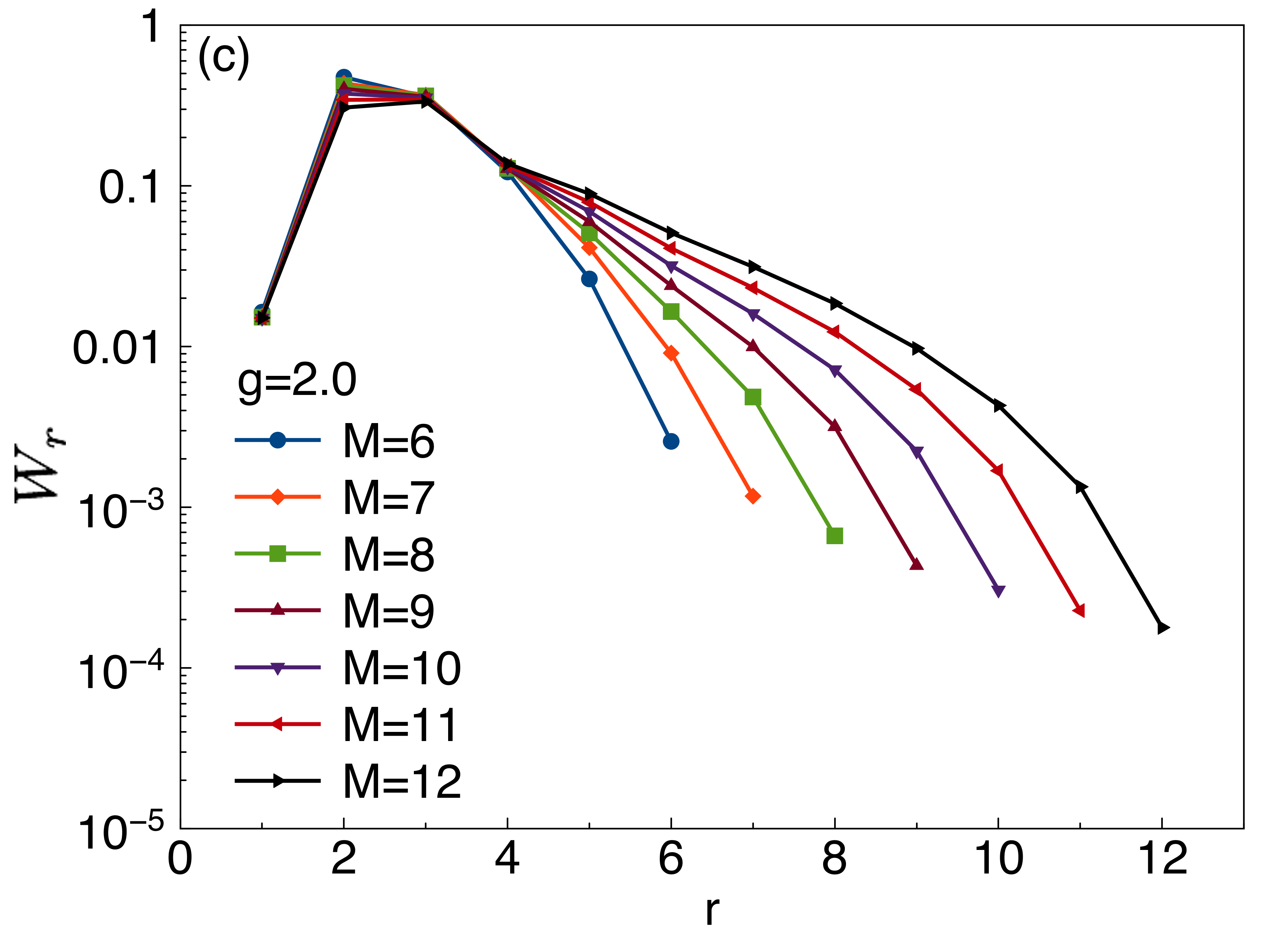}
\includegraphics[width=0.65\columnwidth]{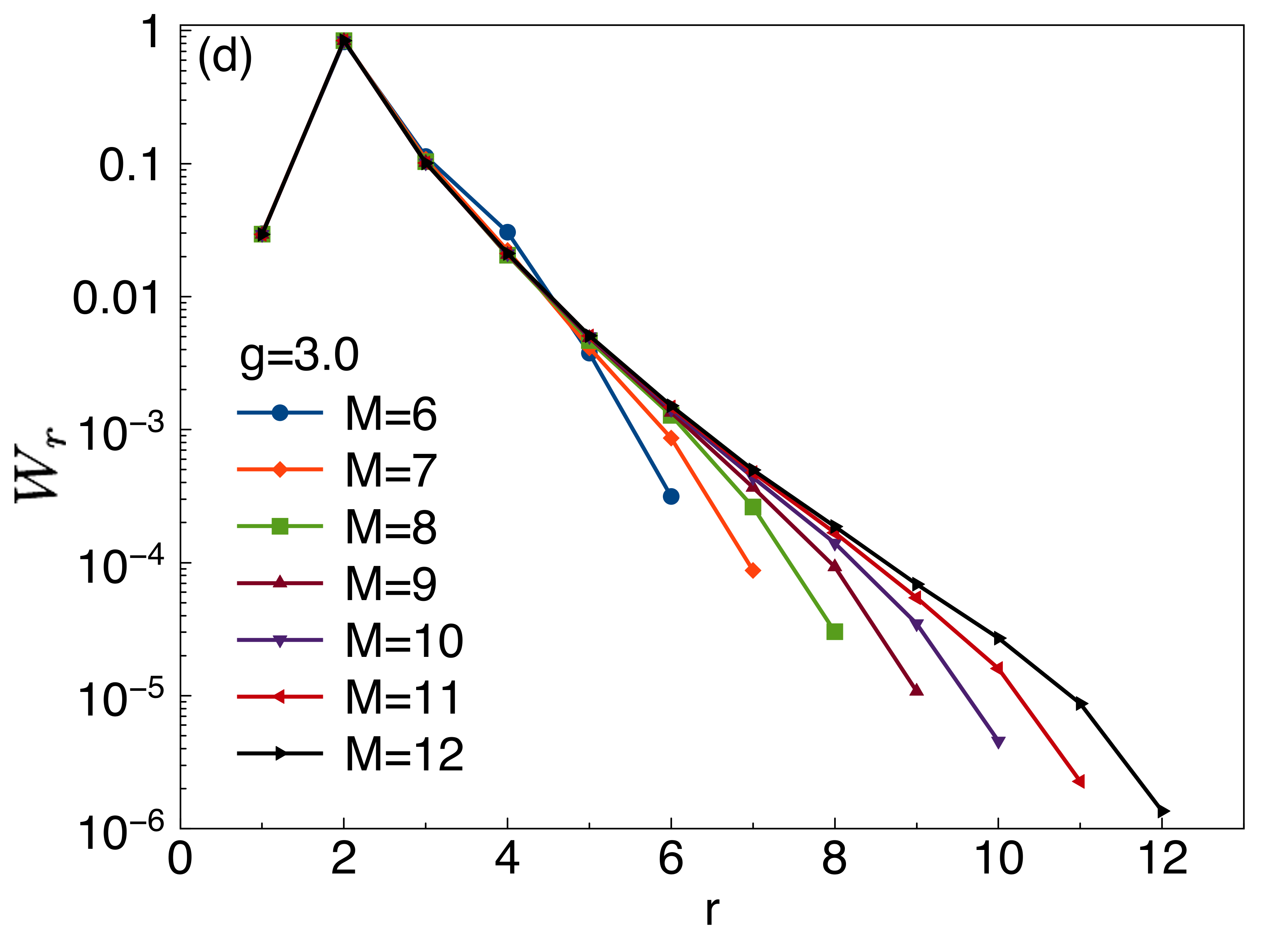}
\includegraphics[width=0.65\columnwidth]{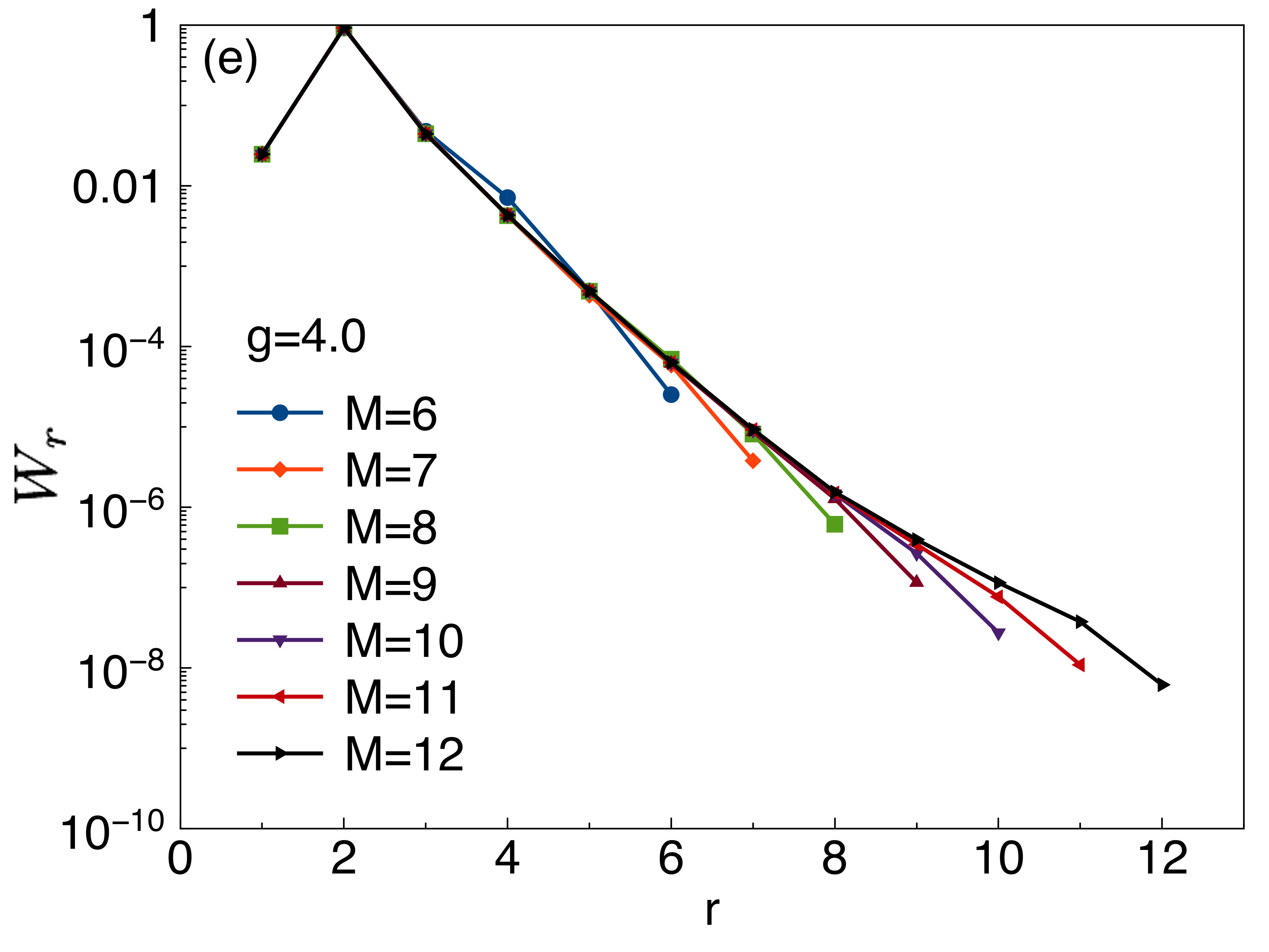}
\includegraphics[width=0.65\columnwidth]{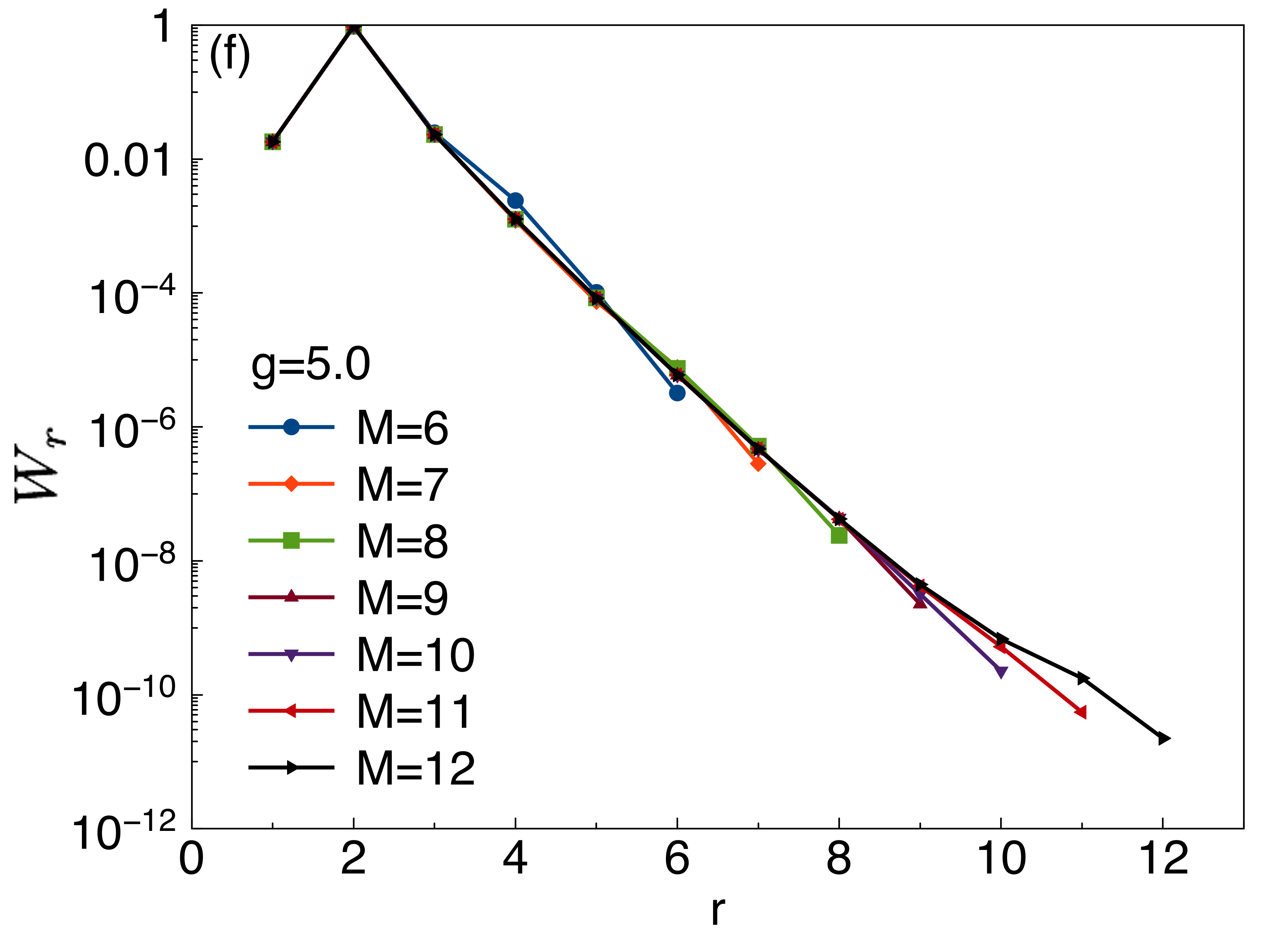}
\caption{\label{fig:Opprofile}
(color online) (a) The weight $W_r$ of range-$r$ operators contained in $\Q_0^{(M)}$ with maximum range $M = 12$ for various $g$.
For large $g \gtrsim 2$, the weight $W_r$ decays exponentially at short distance $r$.
The decay length grows as $g$ decreases.
For small $g \lesssim 2$, the decay of $W_r$ is naively better described by a Gaussian, with the curves almost independent of $g$. 
(b)-(f) The weight $W_r$ of range-$r$ operators in $\Q_0^{(M)}$ when varying $M$ from $M = 6$ to $M = 12$ for fixed $g$ indicated in each panel.
For large $g$, the exponentially decaying part at short distances is essentially converged in $M$; however, the long-distance behavior is not clear.
For small $g$, the weight distribution is pushed to larger $r$ and shows significantly slower decay as a function of $r$ when one increases $M$; this suggests that these operators are not normalizable in the large $M$ limit.
}
\end{figure*}

In this subsection, we examine the real-space shape of the slowest operator more closely.
We define $W_r$ as the weight of $\Q_0^{(M)}$ on range-$r$ operators.
In other words, we can decompose $\Q_0^{(M)} = \sum_{r=1}^M O_r$, with $O_r$ being an operator with range exactly equal to $r$, and define $W_r = \|O_r\|_\text{F}^2$.
The normalization condition ensures that $\sum_r W_r = \sum_r \|O_r\|_\text{F}^2 = \|\Q_0^{(M)}\|_\text{F}^2 = 1$.
Figure~\ref{fig:Opprofile}(a) shows the weights $W_r$ measured for the slowest operator $\Q_0^{(M)}$ with $M = 12$.

For large $g \gtrsim 2$, the weight has an almost-exponential decay at small $r$.
Figures~\ref{fig:Opprofile}(b)-(e) show the weights $W_r$ for $\Q_0^{(M)}$ at fixed $g$ when increasing $M$ from $M = 6$ to $M = 12$.
From the plots, we can see that for large $g$, the weight of the profile is peaked on $2$-local operators, which we can understand already from the leading order SW construction, see Eq.~(\ref{eqn:I1perp}) in App.~\ref{app:ladderformalism}.
We also see that the exponentially decaying part of $W_r$ at short distances is essentially converged, or independent of $M$.
However, the ``shape'' of the operator at long distances is not yet converged and is hence undetermined. 
Despite the fact that we can not determine the long-distance behavior for the slowest operators due to computational limitations, it is clear that the short-distance decay becomes slower when one decreases $g$.

On the other hand, for small $g \lesssim 2$, there is no clear exponential decay even at short distance.
In fact, for fixed $g$ and $M$, the weights appear to decay faster than exponentially (with a Gaussian-like profile).
However, the overall curve shifts to larger $r$ as one increases $M$, with no apparent convergence to some fixed curve independent of $M$.
This suggests the non-normalizability for the $\lim_{M \to \infty} \Q_0^{(M)}$ operators in the small $g$ regime and is also consistent with the result of increasing OIPR as one increases $M$, since there are more Pauli string operators involved in $\Q_0^{(M)}$.

\section{\label{sec:dynamics} Dynamical Simulation}

\begin{figure}
\includegraphics[width=1.0\columnwidth]{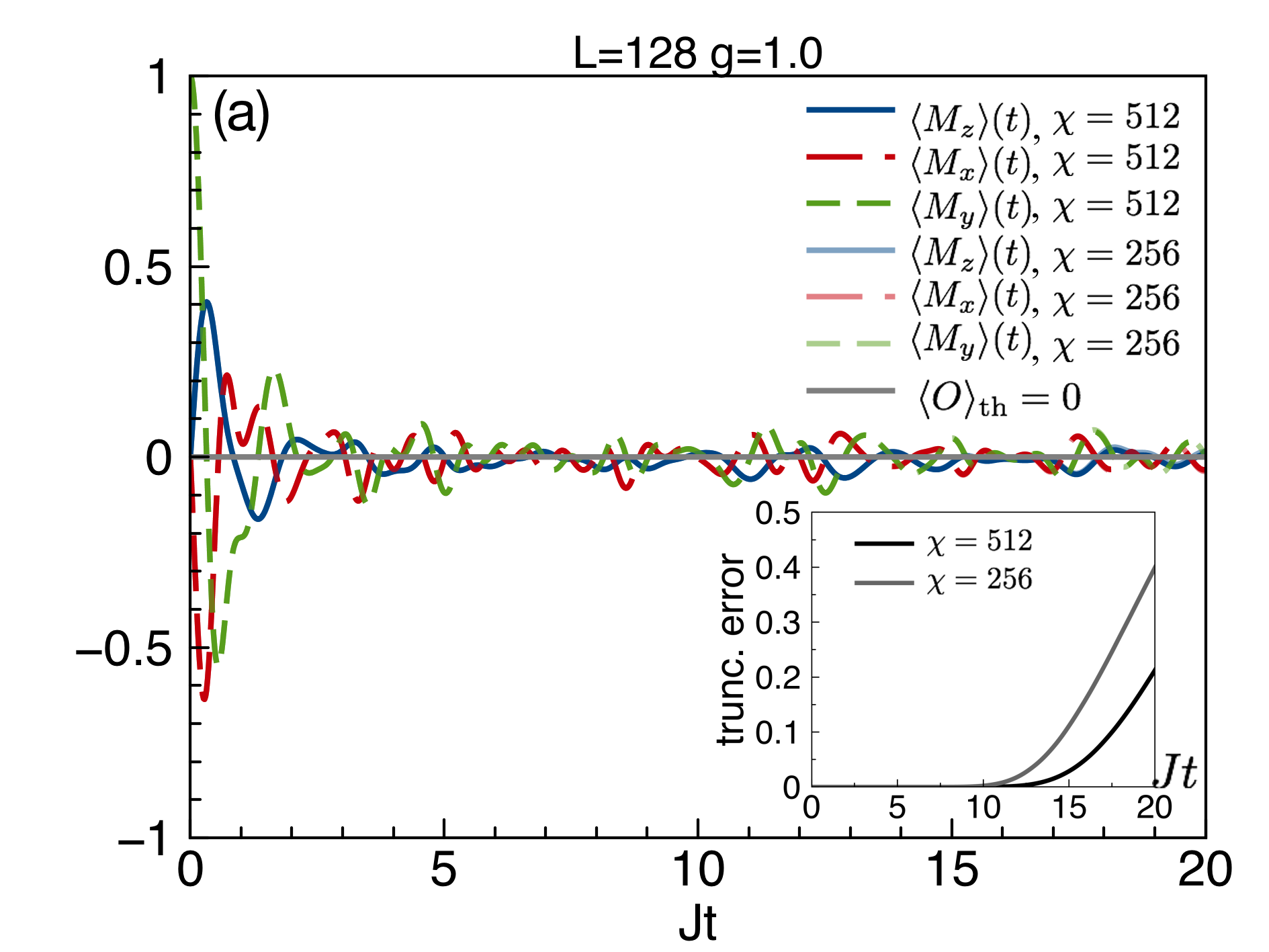}
\includegraphics[width=1.0\columnwidth]{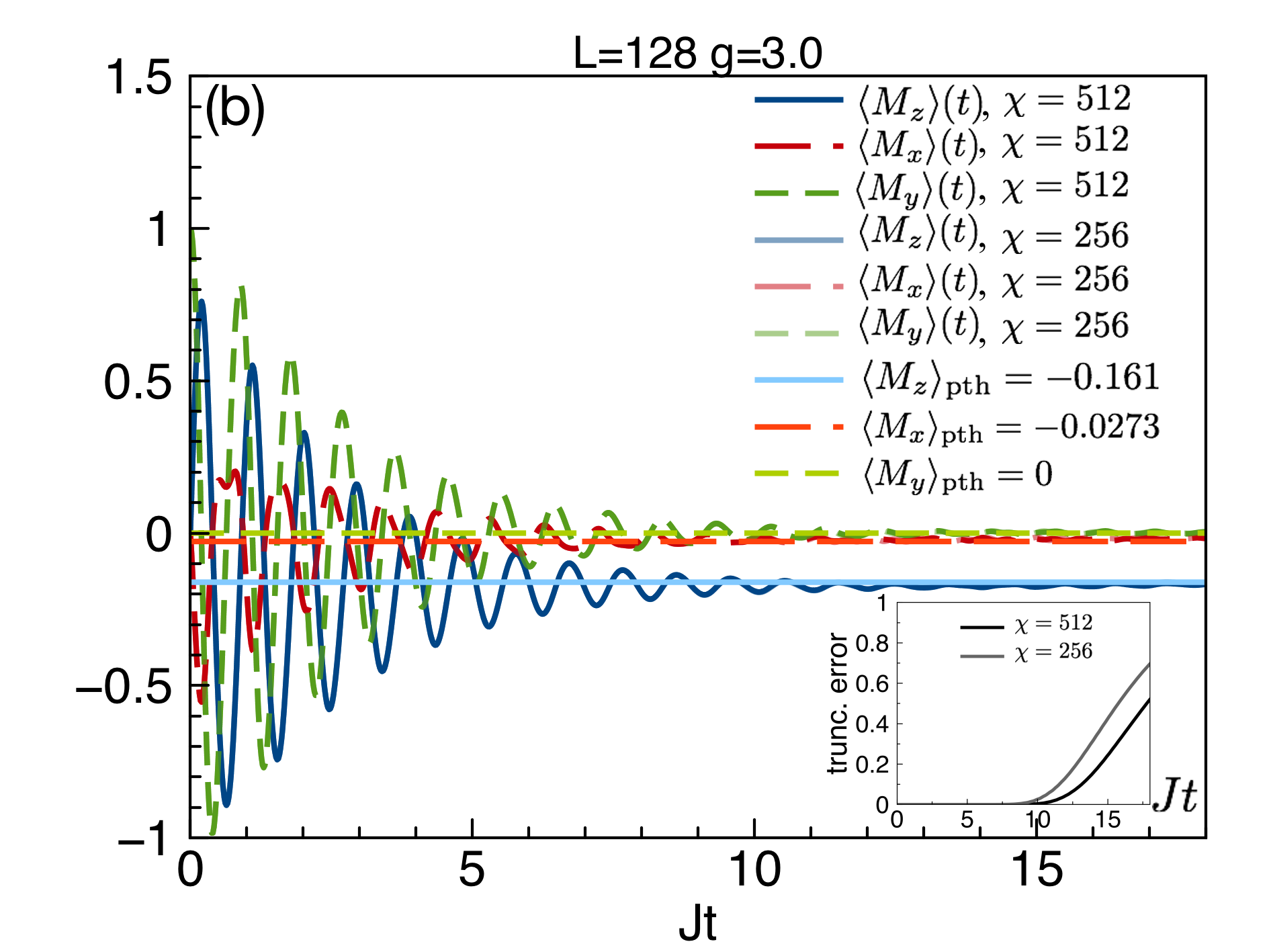}
\caption{\label{fig:TEBD}
(color online) TEBD simulations with bond dimensions $\chi = 256$ and $\chi = 512$ of the evolution of various ``magnetizations'' $\langle M_{x,y,z} \rangle$ upon quench from the initial state $|Y\!+ \rangle$.
The Hamiltonian is given by Eq.~(\ref{eqn:hamiltonian}) with parameters $J = 1.0$, $h = 1.5$, and different $g$ indicated in each panel.
(a) Evolution of the magnetizations for $g = 1$.
The magnetizations appear to approach the thermal value $\langle O \rangle_\text{th} = 0$ expected for any traceless observable $O$.
(b) Evolution of the magnetizations for $g = 3$.
The magnetizations are approaching values described by the generalized Gibbs ensemble that includes also the quasi-conserved operator (see text for details); the expected prethermalized values are marked with subscript ``pth.''
Insets in both panels show truncation error $1 - \langle \psi(t) | \psi(t) \rangle$ of the matrix-product states.
We set the cut-off for the $\chi = 256$ simulation as $\mathfrak{s}_0=10^{-6}$, while for the $\chi = 512$ simulation the cut-off is $\mathfrak{s}_0=10^{-8}$.}
\end{figure}

In order to demonstrate the effect of the quasi-conserved operator that we found in the large $g$ limit, we perform a quench dynamics calculation and observe an intermediate prethermalization state.
We explicitly show that to describe the prethermalization state, one needs to include the slowest operator in the generalized Gibbs ensemble (GGE).
We prepare the initial state as a product state with all spins pointing in the positive-$y$ direction, $|\psi \rangle = |Y\!+ \rangle$ at time $t = 0$.
We evolve the state under the Hamiltonian Eq.~(\ref{eqn:hamiltonian}) as $|\psi(t) \rangle = e^{-iHt} |\psi \rangle$ and measure the evolution of the magnetizations $\langle M_\mu \rangle (t) \equiv \frac{1}{L} \sum_{j=1}^L \langle \psi(t) |\sigma_j^\mu |\psi(t) \rangle / \langle \psi(t) |\psi(t) \rangle$, where $\mu = x, y, z$.
We use time-evolved block-decimation (TEBD) method\cite{vidal_efficient_2003} to simulate the quench dynamics in a system of length $L = 128$ with open boundary conditions.
We use second-order Trotter-Suzuki decomposition with Trotter step $\delta t = 0.02$, which is sufficiently small to achieve the desired accuracy.
We control truncations of the MPS using ``cut-off'' $\mathfrak{s}_0$, which means that we discard singular values smaller than $\mathfrak{s}_0$.
We also use ``bond dimension'' $\chi$, which means that we keep at most $\chi$ singular values.
Two different sets of truncation parameters are used and compared against each other in order to estimate the effect of truncations on the MPS: $\mathfrak{s}_0 = 10^{-6}, \chi = 256$ and $\mathfrak{s}_0 = 10^{-8}, \chi = 512$.
Figure~\ref{fig:TEBD} shows the results of the TEBD calculations.
The loss of norm (truncation error) seen in the insets is due to various truncations and provides some measure of the accuracy of the time evolution (note that it is roughly compensated in the magnetization measurements by normalizing at each $t$, so the exhibited magnetizations are still reasonably accurate over the time range shown).

The effective inverse temperature $\beta$ for any initial state $|\psi \rangle$ is determined by finding the parameter $\beta$ such that equation $\langle \psi| H |\psi \rangle = \frac{1}{Z} \tr[e^{-\beta H} H]$ is satisfied, where $Z = \tr[e^{-\beta H}]$.
The thermal value is defined as $\langle \dots \rangle_\text{th} = \frac{1}{Z} \tr[\rho_\text{th} \dots]$, where $\rho_\text{th} = e^{-\beta H}$ is the associated Gibbs ensemble.
Since $\langle Y\!+| H | Y\!+ \rangle = 0$, it is easy to verify that the effective inverse temperature $\beta = 0$ for this initial state.
As a result, for any traceless observable $O$, the thermal value $\langle O \rangle_\text{th} = 0$.
Hence, if the system thermalizes, the magnetizations $\langle M_\mu \rangle (t)$ should approach zero.

Figure~\ref{fig:TEBD} shows the dynamical evolution of the magnetizations for parameters $g = 1$ and $g = 3$ for system size $L = 128$.
For $g = 1$, even though the magnetizations have not fully equilibrated yet on our simulation times, we can see that they are fluctuating around zero, which is the expected thermal value.
It is therefore reasonable to assume that the magnetizations are equilibrating toward zero, and the system thermalizes, without any prethermalization stage.
On the other hand, for $g = 3$, it is visually clear that $\langle M_z \rangle (t)$ is approaching a sizable nonzero value.
$\langle M_x \rangle (t)$ is also approaching a small nonzero value, even though it is less clear visually.
The prethermalization stage persists over our simulation time, which is consistent with our bound on $t_*$ in Sec.~\ref{subsec:Opnorm}.

Crude features in the dynamics for $g = 3$ can in fact be understood easily as the precession of the spins.
If $J = 0$, the spins, which are pointing along $y\!+$ direction initially, will precess under $H_0$ persistently.
The $T = J \sum_j Z_j Z_{j+1}$ term introduces interactions among the spins, resulting in the decay of the precession, therefore the damping of the magnetization oscillation.
There is a simple quasiparticle description to understand the oscillation and the decay.\cite{lin_quasiparticle_2017} 
Viewing $H_0$ as the ``total particle number,'' part of the $T$ term introduces ``hopping'' of the ``particles.''
The oscillation frequency can essentially be understood as the quasiparticle excitation energy.
Even if we modeled the quasiparticles using an integrable hard-core boson Hamiltonian, the oscillations will damp eventually.
However, the equilibrium value (at least at this intermediate stage) is not described by the Gibbs ensemble.

Here we verify the conjecture that, to describe these intermediate equilibrium values, one needs to include the quasi-conserved quantity into a generalized Gibbs ensemble (GGE).
The GGE in this case is $\rho_\text{pth} \equiv e^{-\alpha H} e^{-\mu \Q_0^{(M)}} / Z_\text{pth}$, and $Z_\text{pth} \equiv \tr[e^{-\alpha H} e^{-\mu \Q_0^{(M)}}]$.
[Here we used the above form for the GGE rather than $e^{-\alpha H - \mu \Q_0^{(M)}}$, since the former is easier to evaluate numerically where one only needs to diagonalize $\Q_0^{(M)}$ once, instead of diagonalizing $\alpha H + \mu \Q_0^{(M)}$ for each pair of $(\alpha, \mu)$.
Furthermore, since $\Q_0^{(M)}$ and $H$ almost commute, we expect the two expressions are approximately the same.]
The parameters $(\alpha, \mu)$ are determined by finding the values satisfying the following equations
\begin{eqnarray}
\langle \psi| H |\psi \rangle &=& \frac{1}{Z_\text{pth}} \tr[H \rho_\text{pth}] ~, \\
\langle \psi| \Q_0^{(M)} |\psi \rangle &=& \frac{1}{Z_\text{pth}} \tr[\Q_0^{(M)} \rho_\text{pth}] ~.
\end{eqnarray}
For the initial state $|Y\!+ \rangle$, $\langle Y\!+ | H |Y\!+ \rangle = 0$; while $\frac{1}{L} \langle Y\!+ | \Q_0^{(M)} | Y\!+ \rangle = 0.63889$ using $\Q_0^{(M=12)}$.
In fact, the ``particle densities'' in the initial state, $\frac{1}{L} \langle Y\!+ | \Q_0^{(M)} | Y\!+ \rangle$, measured from $M = 8$ to $M = 11$ are within approximately $1\%$ from the $M = 12$ result.
Note also that since the initial state is a product state, the particle density in a finite system of size $L$ will be independent of $L$ as long as $L \geq M$.
We then solve for $(\alpha, \mu)$ on the right-hand side using Newton's method, while $\rho_\text{pth}$ is evaluated by the exact diagonalization of $H$ and $\Q_0^{(M)}$ for system size $L = 16$ and $M = 8$ (the largest $L$ and $M$ accessible with our computation resources), under periodic boundary condition; we find $(\alpha, \mu) = (-0.05155, -1.4417)$.
We then calculate the prethermalized GGE values as $\langle M_z \rangle_\text{pth} = -0.161045$, $\langle M_x \rangle_\text{pth} = -0.0273397$, and $\langle M_y \rangle_\text{pth} = 0$ (by time-reversal symmetry in the effective Hamiltonian for the prethermalized state), where $\langle \dots \rangle_\text{pth} = \frac{1}{Z_\text{pth}} \tr[\rho_\text{pth} \dots]$.
Figure~\ref{fig:TEBD}(b) shows a fair agreement between the observed prethermal equilibrium values $\langle M_\mu \rangle(t)$ and the GGE estimates $\langle M_\mu \rangle_\text{pth}$.

We have thus explicitly verified that the quasi-conserved operator in the large $g$ regime has nontrivial effects on the relaxation of the system.
Furthermore, to describe the equilibrium values at the intermediate prethermalization stage, one needs to include this quasi-conserved operator in the generalized Gibbs ensemble.

\section{\label{sec:conclusion} Discussion}
We numerically construct the slowest operator that is translationally invariant with maximum range $M$.
In the small coupling regime, the norm of the commutator of the slowest operator with the Hamiltonian has a power-law dependence on $M$.
On the other hand, in the strong coupling regime, we find exponential decay at least at small $M$, identifying the slowest operator as quasi-conserved operator.
At larger $M$, however, the decay becomes slower, possibly a power law.
This may be related to the eventual thermalization of the system, after a prethermalization stage with a parametrically long time scale.
The true behavior at large $M$ is not certain due to the limitations of our numerical calculations, constrained by the exponentially large operator Hilbert space.
However, from the analysis of the OIPR, it appears that the quasi-conserved operator resides only on a very small fraction of states in the total Hilbert space.
It may therefore be possible to reduce the relevant operator Hilbert space dimension by identifying the property of this space and by restricting studies to only such an ansatz, which could potentially allow reaching larger maximum range; we leave this idea for future studies.

Our TEBD calculation of the dynamics after a quench explicitly confirms the existence of the prethermalization stage for large $g$ and further supports the GGE construction that includes the quasi-conserved operator.
From the residual Frobenius norm of the quasi-conserved operator $\sqrt{\lambda_0(M)}$, we can heuristically provide a lower bound on the thermalization time scale as $t_* \sim 2^{-\frac{M+1}{2}} \lambda_0(M)^{-1/2}$; we can also bound the thermalization time more accurately by measuring the conventional operator norm, $t_* \sim (\|[H, \mathcal{Q}_0^{(M)}]\|_{\text{op}}/\|\mathcal{Q}_0^{(M)}\|_{\text{op}})^{-1}$.
However, we cannot determine the time scale of the prethermalization stage from the TEBD calculations due to the limited accessible simulation time.
Even if we could extend the TEBD calculation to longer times, we may have to consider a different truncation scheme\cite{Chris_DMT} to get more accurate results.
A straightforward truncation of small singular values in the MPS state does not necessary conserve the quasi-conserved quantity, and hence may artificially decrease the prethermalization time.
It would be interesting to extract the prethermalization time scale directly from simulations or even from experiments to compare with our heuristic argument.

Another interesting observation which we still do not fully understand is the apparent ``transition'' between the prethermalization and ergodic behaviors.
While it is not clear what defines the prethermalization ``phase,'' it appears that the different scaling behavior of the residual norm can serve as an indicator.
Furthermore, the OIPR seems to provide a stronger signature:  the OIPR of the slowest operator appears to converge with $M$ in the large coupling regime $g \geq 2.5$, while the OIPR diverges in the ergodic regime.
Also, the operator profile appears to converge with increasing $M$ for $g \geq 3$ while it does not converge for $g \leq 2$.
The persistence of this sharp distinction between the prethermalization and ergodic behaviors to larger $M$ or even $M \to \infty$ deserves more study.

An exciting possibility which may be suggested by our results for $g \geq 3$ is the existence of the truly conserved quasi-local quantity,\cite{fagotti_conservation_2014, grover_quantum_2014, garrison_partial_2017} or the convergence of the SW transformation in the $n \to \infty$ limit.
While the theoretical upper bounds on the norms in the SW series do not prove the convergence, they do not disprove it either.
In fact, from our numerical calculations in App.~\ref{app:Vmnorm}, the convergence of the SW transformation might even be possible.
This would imply that we can find a (quasi-local) unitary transformation $U$ such that $U^\dagger H U$ commutes with $H_0$.
A partial breakdown of ETH would be possible due to the existence of this emergent ``particle conservation'' in the entire spectrum.
In fact, the quantum Ising model $H = \sum_j X_j + \epsilon  \sum_j Z_j Z_{j+1}$ provides an example where the SW procedure converges.\cite{else_prethermal_2017, Lin_QIM_unpub2017}
In this case, instead of one (or few) conserved quantity, there is a macroscopic number of conservation laws due to the model's integrability.
Nevertheless, the SW procedure ``does not know'' the free fermion solution but still converges and finds a conserved quantity, which happens to be the total number of the Bogoliubov quasi-particles.
Our intriguing results in the nonintegrable model thus warrant further detailed studies of the convergence of the SW transformation.

In conclusion, by numerically searching for the slowest operator, we identified the quasi-conserved operator at large coupling, which we believe is responsible for the prethermalization behavior.
The residual norm of the quasi-conserved operator has exponential decay with its maximum range up to some point; the OIPR and real-space profile show that it is localized in the operator Hilbert space and real space.
By comparing with the perturbative SW construction, we concluded that the quasi-conserved quantity is essentially the dressed total spin-$z$ operator.
Finally, by simulating the quench dynamics, we verified the conjecture that the quasi-conserved quantity leads to prethermalization behavior.
Furthermore, the apparent equilibrium values at the prethermalization stage can be described by including the quasi-conserved quantity in the GGE.

\begin{acknowledgments}
The authors would like to thank M.~C.~Ba\~nuls, M.~Serbyn, V.~Khemani, and N.~J.~Robinson for valuable discussions at a poster session at Aspen Center for Physics where bulk of this work was presented in January 2017.
We would also like to thank J.~R.~Garrison, R.~V.~Mishmash and M.~P.~A.~Fisher for inspiring discussions of their work and C.~White for discussions on the TEBD calculation.
This work was supported by NSF through grant DMR-1619696 and also by the Institute for Quantum Information and Matter, an NSF Physics Frontiers Center, with support of the Gordon and Betty Moore Foundation.
\end{acknowledgments}

\appendix
\section{Generalized Ladder Algebra Formalism}
\label{app:ladderformalism}
In Ref.~\onlinecite{macdonald_$fractu$_1988}, MacDonald {\it et.~al.}\ proposed a perturbative expansion for the electronic Hubbard model in the large $U$ limit using generalized ladder algebra formalism.
In fact, their transformation is a variant of a local SW transformation.\cite{datta_low-temperature_1996, bravyi_schriefferwolff_2011}
A small difference from the SW transformation used in the present work is that Ref.~\onlinecite{macdonald_$fractu$_1988} constructs a unitary transformation of the form $\exp(i \epsilon S_1 + i \epsilon^2 S_2 + \dots + i \epsilon^n S_n)$ rather than $\exp(i \epsilon S_1) \exp(i \epsilon^2 S_2) \dots \exp(i \epsilon^n S_n)$.
This modifies Eq.~(\ref{eqn:Vm}) by replacing $\mathfrak{f}(k_1, \dots, k_p)$ to $\frac{1}{p!}$, see Ref.~\onlinecite{datta_low-temperature_1996}.
The variant in the present paper is slightly easier to use in numerical calculations because there are fewer terms in the series.

For our spin Hamiltonian, the spectrum of the solvable limit $H_0$ is composed of different sectors labeled by different ``particle'' numbers.
To be concrete, consider 
\begin{equation}
\label{H0app}
H_0 = \Gamma \sum_j Z_j ~,
\end{equation}
where we have rotated $g X_j + h Z_j$ to the new $z$-direction and $\Gamma = \sqrt{g^2 + h^2}$.
The (rotated) perturbation $T$ can be decomposed into $T = \sum_{\ell=-2}^2 T_{\ell}$, where $T_{\ell}$-s are called generalized ladder operators, with the property that $[H_0, T_{\ell}] = 2\Gamma \ell T_{\ell}$.
More explicitly, defining $P_j, M_j = \frac{1}{2}(X_j \pm i Y_j)$, we have
\begin{eqnarray}
T_{+2} &=& t_2 \sum_j P_j P_{j+1} ~, \label{T+2} \\
T_{-2} &=& t_2 \sum_j M_j M_{j+1} = T_{+2}^\dagger ~, \\
T_{+1} &=& t_1 \sum_j (P_j Z_{j+1} + Z_j P_{j+1}) ~, \\
T_{-1} &=& t_1 \sum_j (M_j Z_{j+1} + Z_j M_{j+1}) = T_{+1}^\dagger ~, \\
T_0 &=& u_0 \sum_j Z_j Z_{j+1} + w_0 \sum_j (P_j M_{j+1} + M_j P_{j+1}) ~, ~~~ \label{T0}
\end{eqnarray}
where $t_1 = -\frac{J g h}{\Gamma^2}$, $t_2 = \frac{J g^2}{\Gamma^2}$, $u_0 = \frac{J h^2}{\Gamma^2}$, and $w_0 = t_2$.

Let us further define
\begin{equation}
T^{(k)}(\ell_1, \dots, \ell_k) \equiv T^{(k)}[\ell] = T_{\ell_1} \dots T_{\ell_k} ~.
\end{equation}
One can easily verify that these operators are also generalized ladder operators:
$[H_0, T^{(k)}[\ell]] = 2\Gamma M^{(k)}[\ell] T^{(k)}[\ell]$, where $M^{(k)}[\ell] \equiv \sum_{i=1}^k \ell_i$.
In particular, if $M^{(k)}[\ell] = 0$, then $T^{(k)}[\ell]$ is in the nullspace of $\ad_{H_0}$.

It is easy to argue that $V_m$ can all be expressed as nested commutators of $T_\ell$-s by mathematical induction from Eq.~(\ref{eqn:Vm}) and Eq.~(\ref{eqn:solSm}), given that $i S_k$ and $V_k$ are all composed of nested commutators of $T_\ell$-s for $k < m$.
Assuming $V_m = (2\Gamma)^{1-m} \sum_{\{\ell\}} C^{(m)}[\ell] \, T^{(m)}[\ell]$, where coefficients $C^{(m)}[\ell]$ have special structure such that $V_m$ is composed of nested commutators of $T_\ell$-s, Eq.~(\ref{eqn:solSm}) gives 
\begin{equation}
i S_m = (2\Gamma)^{-m} \sum_{\{\ell\}, M^{(m)}[\ell] \neq 0} \frac{C^{(m)}[\ell] \, T^{(m)}[\ell]}{M^{(m)}[\ell]} ~.
\end{equation}
One can therefore see that it is a special type of the local SW where everything is expressed by the generalized ladder algebra.

As an example, we work out the effective Hamiltonian and the quasi-conserved operator to second-order.
At first order, $V_1 = T$, so we want to find $i S_1$ such that $i\ad_{S_1}(H_0) + T = T_0$.
The solution is 
\begin{equation}
i S_1 = \frac{1}{4\Gamma}(T_{+2} - T_{-2}) + \frac{1}{2\Gamma}(T_{+1} - T_{-1}) ~.
\end{equation}
We therefore obtain 
\begin{eqnarray}\label{eqn:V2app}
V_2 &=& \frac{1}{2} i\ad_{S_1} i\ad_{S_1}(H_0) + i\ad_{S_1}(T)
= \frac{1}{2} i\ad_{S_1}(T_0 + T) \nonumber \\
&=& \frac{1}{8\Gamma} 
\Big( 2[T_{+2}, T_0] - 2[T_{-2}, T_0] + 4[T_{+1}, T_0] - 4[T_{-1}, T_0] \nonumber \\
&-& [T_{+2}, T_{+1}] + [T_{-2}, T_{-1}] + 3[T_{+2}, T_{-1}] - 3[T_{-2}, T_{+1}] \nonumber \\
&+& 2[T_{+2}, T_{-2}] + 4[T_{+1}, T_{-1}] \Big) ~.
\end{eqnarray}
The last line is the diagonal part of $V_2$ while the rest is the off-diagonal part.
At second order, we solve for $i S_2$ such that $i\ad_{S_2}(H_0) + V_2 = V_2^\text{diag}$; the solution is
\begin{eqnarray}
i S_2 &=& \frac{1}{48\Gamma^2} \Big( 3[T_{+2}, T_0] + 3[T_{-2}, T_0] + 12[T_{+1}, T_0] \nonumber \\
&+& 12[T_{-1}, T_0] - [T_{+2}, T_{+1}] - [T_{-2}, T_{-1}] \nonumber \\
&+& 9[T_{+2}, T_{-1}] + 9[T_{-2}, T_{+1}] \Big) ~.
\end{eqnarray}

We can now obtain contributions to the quasi-conserved operator as 
\begin{eqnarray}
I_1 &=& T - T_0 = T_{+2} + T_{-2} + T_{+1} + T_{-1} ~, \\
I_2 &=& -V_2^\text{diag} + i\ad_{S_1}(T_0) = \frac{1}{4\Gamma} \Big(-[T_{+2}, T_{-2}] - 2[T_{+1}, T_{-1}] \nonumber \\
&+& [T_{+2} - T_{-2}, T_0] + 2[T_{+1} - T_{-1}, T_0] \Big) ~.
\end{eqnarray}
To compare with the slowest operator approach, we calculate the component perpendicular to $H$, which can be obtained via Eq.~(\ref{eqn:I(n)perp}).
For example, we find for the leading-order SW construction,
\begin{equation}\label{eqn:I1perp}
I_1^\perp = I_1 - \frac{J^2 g^2 (g^2 + 4 h^2)}{2 (J^2+g^2 + h^2)^2(g^2+h^2)} H ~.
\end{equation}
This can be used to understand the 1-local and 2-local content of the slowest operator for large $g$, see Fig.~\ref{fig:Opprofile}.

\section{Bound on $H_{>n}$}
\label{app:boundHn}
In this Appendix, we prove the bound on $\|H_{>n}\|_\text{F}$ quoted in the main text.
We set the norm of $H_0$ as the energy unit, $\|H_0\|_\text{F} = \Gamma$, and the norm of the perturbation term as $\epsilon \|T\|_\text{F} = \epsilon \Gamma$, where $\epsilon$ is the strength of the perturbation and is used to organize the perturbative expansion.
We also assume that $H_0 \in \mathcal{T}_1$ and $T \in \mathcal{T}_2$.
Without loss of generality, we assume working in the basis such that $H_0 = \Gamma \sum_j Z_j$, since for any general $H_0 \in \mathcal{T}_1$ one can always rotate the basis to achieve this.
The results in this appendix are parallel to the results obtained in Ref.~\onlinecite{bravyi_schriefferwolff_2011} but are tailored to our definitions of norms for translationally-invariant operators and the specific SW procedure used; furthermore, our results are not restricted to effective Hamiltonians in the lowest-energy sector but are valid for the entire spectrum.

We first prove the locality of the operators $S_m$ and $V_m$ in the SW transformation procedure, Sec.~\ref{subsec:SWproc}, and of the operators $I_m$ in the quasi-conserved quantity obtained by SW transformation, Sec.~\ref{subsec:quasiIOviaSW}.
\begin{prop}
\label{prop:locality}
$V_m \in \mathcal{T}_{m+1}$, $S_m \in \mathcal{T}_{m+1}$, and $I_m \in \mathcal{T}_{m+1}$.
\end{prop}
\begin{proof}
By assumption, $H_0 \in \mathcal{T}_1$, hence $\ad_{H_0}$ maps $\mathcal{T}_m$ to $\mathcal{T}_m$.
The pseudo-inverse $[\ad_{H_0}]^{-1}$ thus also maps from $\mathcal{T}_m$ to $\mathcal{T}_m$.
Therefore, from Eq.~(\ref{eqn:solSm}) it follows that if $V_m \in \mathcal{T}_{m+1}$ then $S_m \in \mathcal{T}_{m+1}$.
Initially, $V_1 = T \in \mathcal{T}_2$ and hence $S_1 \in \mathcal{T}_2$.
Assume $V_k \in \mathcal{T}_{k+1}$ and $S_k \in \mathcal{T}_{k+1}$ hold for $k \leq m - 1$.
Now consider the first term in $V_m$ in Eq.~(\ref{eqn:Vm}); we see that $i\ad_{S_{k_p}} \dots i\ad_{S_{k_1}}(H_0) \in \mathcal{T}_{m+1}$ since $k_1 + \dots + k_p = m$.
The second term in $V_m$ is also in $\mathcal{T}_{m+1}$, by noticing that $k_1 + \dots + k_p = m - 1$ and $T \in \mathcal{T}_2$ in $i\ad_{S_{k_p}} \dots i\ad_{S_{k_1}}(T)$.
By similar argument applied to Eq.~(\ref{eqn:solIm}), we have $I_m \in \mathcal{T}_{m+1}$.
The proposition is proved by mathematical induction.
\end{proof}

Here we introduce a different norm on the operator Hilbert space $\mathcal{T}_k$ which will be technically useful in the future proofs.
Consider any operator $O \in \mathcal{T}_k$ written in the Pauli-string basis composed of $I$, $P \equiv \frac{1}{2}(X + iY)$, $M \equiv \frac{1}{2}(X - iY)$, and $Z$: 
$O = \sum_j \sum_{\mathbf{a}} o_{\mathbf{a}} Q_{j; k}^{\mathbf{a}}$, where $Q_{j; k}^{\mathbf{a}} = \sigma_j^{a_1} \dots \sigma_{j+k-1}^{a_k}$ denotes the ``$I$-$P$-$M$-$Z$'' string with support on sites $j$ to $j + k - 1$ with non-identity on the site $j$.
That is, $\sigma$ on each site other than $j$ can be one of the four operators $I$, $P$, $M$, or $Z$, while it can be only $P$, $M$, or $Z$ on the site $j$ (recall that $\mathcal{T}_k$ consists of traceless operators, and this ``gauge'' choice for writing local operators is similar to the one in the main text).
\begin{defi}
For $O = \sum_j \sum_{\mathbf{a}} o_{\mathbf{a}} Q_{j; k}^{\mathbf{a}} \in \mathcal{T}_k$, the one-norm is defined as $\|O\|_1 = \sum_{\mathbf{a}} |o_{\mathbf{a}}|$.  
\end{defi}
Such a definition of the one-norm is in fact basis-dependent, so it is crucial that our one-norm is understood in the basis such that $H_0 = \Gamma \sum_j Z_j$ and operators are expanded in the $I$-$P$-$M$-$Z$ strings.
These particular $I$-$P$-$M$-$Z$ strings are orthogonal but not normalized with the respect to the Frobernius inner product in $\mathcal{T}_k$.
In fact, $\|Q_{j; k}^{\mathbf{a}}\|_\text{F}^2 = 2^{-N_{\mathbf{a}}}$, where $N_{\mathbf{a}}$ is the number of $P$ and $M$ letters in $Q_{j; k}^{\mathbf{a}}$. 

Our one-norm can be used to bound the Frobenius norm discussed in the main text:
\begin{prop}
\label{prop:onenorm_vs_Frobenius}
For $O \in \mathcal{T}_k$, we have $\|O\|_\text{F} \leq \|O\|_1 \leq \sqrt{5 \cdot 6^{k-1}} \|O\|_\text{F}$.
\end{prop}
\begin{proof}
Indeed, writing $O$ in the $I$-$P$-$M$-$Z$ strings as $O = \sum_j \sum_{\mathbf{a}} o_{\mathbf{a}} Q_{j; k}^{\mathbf{a}}$, we have 
\begin{eqnarray*}
\|O\|_\text{F}^2 = \sum_{\mathbf{a}} |o_{\mathbf{a}}|^2  2^{-N_{\mathbf{a}}} \leq \sum_{\mathbf{a}} |o_{\mathbf{a}}|^2 \leq \left( \sum_{\mathbf{a}} |o_{\mathbf{a}}| \right)^2 = \|O\|_1^2 ~.
\end{eqnarray*}
The last inequality follows from the fact that there are more non-negative terms on the right-hand side.

For the bound on the one-norm, using Cauchy-Schwartz inequality, we have
\begin{equation*}
\sum_{\mathbf{a}} \left( |o_{\mathbf{a}}| \, 2^{-\frac{N_{\mathbf{a}}}{2}} \right) \left( 2^{\frac{N_{\mathbf{a}}}{2}} \right) \leq \sqrt{\sum_{\mathbf{a}} \left( |o_{\mathbf{a}}| \, 2^{-\frac{N_{\mathbf{a}}}{2}} \right)^2} \sqrt{\sum_{\mathbf{a}} \left( 2^{\frac{N_{\mathbf{a}}}{2}} \right)^2} ~, 
\end{equation*}
or
\begin{equation}
\|O\|_1 \leq \|O\|_\text{F} \sqrt{\sum_{\mathbf{a}} 2^{N_{\mathbf{a}}}} ~. 
\end{equation}
Remembering that the first site can only be $P$, $M$, or $Z$, a simple combinatorial exercise gives $\sum_{\mathbf{a}} 2^{N_{\mathbf{a}}} = 5 \cdot 6^{k-1}$.
\end{proof}

We now present two propositions describing key properties of our one-norm that will be used in the proof of the main bounds.
\begin{prop}
\label{prop:commutator_norm}
If $U \in \mathcal{T}_r$ and $W \in \mathcal{T}_s$, then $\|\ad_U(W)\|_1 \leq 2(r + s - 1) \|U\|_1 \|W\|_1$.
\end{prop}
\begin{proof}
By writing out $U = \sum_j \sum_{\mathbf{a}} u_{\mathbf{a}} Q^{\mathbf{a}}_{j; r}$ and $W = \sum_k \sum_{\mathbf{b}} w_{\mathbf{b}} Q^{\mathbf{b}}_{k; s}$ in the $I$-$P$-$M$-$Z$ strings, we have
\begin{equation}\label{eqn:adUW}
\|\ad_U(W)\|_1 = \| \sum_k \sum_{j = k - r + 1}^{k + s - 1} \sum_{\mathbf{a},\mathbf{b}} u_{\mathbf{a}} w_{\mathbf{b}} [Q^{\mathbf{a}}_{j; r}, Q^{\mathbf{b}}_{k; s}] \|_1 ~.
\end{equation}
Let us first consider the product $Q^{\mathbf{a}}_{j; r} Q^{\mathbf{b}}_{k; s}$ for a particular $j$ and strings $\mathbf{a}$ and $\mathbf{b}$.
By the multiplication rules among $I$, $P$, $M$, and $Z$, we note that $Q^{\mathbf{a}}_{j; r} Q^{\mathbf{b}}_{k; s}$ will ``split'' into $2^{N_{\mathbf{a}, \mathbf{b}}}$ new $I$-$P$-$M$-$Z$ strings, where $N_{\mathbf{a},\mathbf{b}}$ is the number of the positions that the letter $P$ in $Q^{\mathbf{a}}_{j; r}$ collides with $M$ in $Q^{\mathbf{b}}_{k; s}$ or $M$ in $Q^{\mathbf{a}}_{j; r}$ collides with $P$ in $Q^{\mathbf{b}}_{k; s}$, since $P M = \frac{1}{2}(I + Z)$ and $M P = \frac{1}{2}(I - Z)$.
However, each such new string will carry a factor $2^{-N_{\mathbf{a},\mathbf{b}}}$, with a plus or minus sign.
Therefore, $Q^{\mathbf{a}}_{j; r} Q^{\mathbf{b}}_{k; s}$ will generate $2^{N_{\mathbf{a},\mathbf{b}}}$ new strings carrying coefficients $\pm u_{\mathbf{a}} w_{\mathbf{b}} 2^{-N_{\mathbf{a},\mathbf{b}}}$, and likewise for $Q^{\mathbf{b}}_{k; s} Q^{\mathbf{a}}_{j; r}$.
Upon summing over $k$, each new string should be understood as ``gauge-fixed'' by shifting the position such that the first non-trivial letter is at position $k$.

Now we consider writing out the full $\ad_U(W)$ in Eq.~(\ref{eqn:adUW}) in the $I$-$P$-$M$-$Z$ strings.
Coefficient for each basis string will be some collection of the contributions described above from different $j$, $\mathbf{a}$, and $\mathbf{b}$.
Applying the triangle inequality $|x + y + \dots + z| \leq |x| + |y| + \dots + |z|$ for each such coefficient, we then have  
\begin{eqnarray}\label{eqn:adUWfinal}
\|\ad_U(W)\|_1 & \leq & 2 \sum_{j = k - r + 1}^{k + s - 1} \sum_{\mathbf{a}, \mathbf{b}} |u_{\mathbf{a}} w_{\mathbf{b}} 2^{-N_{\mathbf{a},\mathbf{b}}}| 2^{N_{\mathbf{a},\mathbf{b}}} \nonumber \\
&=& 2(r + s - 1) \|U\|_1 \|W\|_1 ~,
\end{eqnarray}
where the first factor of $2$ accounts for $Q^{\mathbf{a}}_{j; r} Q^{\mathbf{b}}_{k; s}$ and $Q^{\mathbf{b}}_{k; s} Q^{\mathbf{a}}_{j; r}$, and the factor of $r + s - 1$ comes from the counts of $j$.
\end{proof}

Equation~(\ref{eqn:solSm}) establishes the relation between $S_m$ and $V_m$, from which we deduce the following Proposition:
\begin{prop}
\label{prop:boundSm}
$\|S_m\|_1 \leq \frac{\|V_m\|_1}{2\Gamma}$.
\end{prop}
\begin{proof}
First, we note that since $[\ad_{H_0}]^{-1}$ is the pseudoinverse of $\ad_{H_0}$, it is customary to rewrite Eq.~(\ref{eqn:solSm}) as $i S_m = [\ad_{H_0}]^{-1} V_m$.
The pseudoinverse of $\ad_{H_0}$ in fact can be easily obtained as follows.
To be specific, let us consider $\ad_{H_0}$ as a map from $\mathcal{T}_{m+1}$ to $\mathcal{T}_{m+1}$, since $i S_m$ and $V_m$ belong to $\mathcal{T}_{m+1}$.
Also recall that we have rotated the Pauli basis such that $H_0 = \Gamma \sum_j Z_j$ in order to define the one-norm.
The $I$-$P$-$M$-$Z$ strings are in fact (non-normalized) eigenvectors of $\ad_{H_0}$ with eigenvalues $2(N_P - N_M) \Gamma$, where $N_P$ ($N_M$) is the number of $P$ ($M$) in the $I$-$P$-$M$-$Z$ string.
The pseudoinverse $[\ad_{H_0}]^{-1}$ is thus diagonal with eigenvalues $\frac{1}{2(N_P - N_M) \Gamma}$ if $N_P - N_M \neq 0$ and zero if $N_P - N_M = 0$.
Therefore, assuming $V_m = \sum_j \sum_{\mathbf{a}} v_{\mathbf{a}} Q^{\mathbf{a}}_{j; m+1}$ in the $I$-$P$-$M$-$Z$ strings, we have
\begin{eqnarray}
\|S_m\|_1 &=& \sum_{\mathbf{a}: N_P - N_M \neq 0} \left|\frac{v_{\mathbf{a}}}{2(N_P - N_M) \Gamma} \right| \nonumber \\
&\leq& \sum_{\mathbf{a}} \frac{|v_{\mathbf{a}}|}{2\Gamma} = \frac{\|V_m\|_1}{2\Gamma} ~.
\end{eqnarray}
\end{proof}

We are now ready to consider the SW-rotated Hamiltonian, Eq.~(\ref{Hprime_expansion}).
To remind readers, $H'$ is obtained by an exact unitary rotation using generators $i S_1, \dots, i S_n$, which we call $n$-th order SW, with specific rules for finding these generators.
Equation~(\ref{Hprime_expansion}) represents a formal expansion of $H'$ in powers of $\epsilon$.
The ``potentials'' $V_m$ in Eq.~(\ref{eqn:Vm}) for $m \leq n$ (actually, even $m \leq n+1$) are already representative of the infinite-order SW series and do not depend on $n$, while the potentials for $m > n$ that contribute to the ``remainder'' $H_{>n}$ actually depend on $n$.
Not to overburden the notation, we consider $n$ as fixed and do not put extra label on such $V_m$.
Below, we focus on convergence properties of the formal expansion in $\epsilon$ of $H_{>n}$, which will also provide a bound on its norm and inform us about locality properties of $H'$.

To obtain an upper bound on the norm of $H_{>n}$, we need some control over the $V_m$ terms, especially for $m > n$.
This is provided by the following Lemma.
\begin{lem}
\label{lem:boundVm}
In the SW construction to the $n$-th order, for $m > n$, $\|V_m\|_F \leq \Gamma (\rho_n)^{-m}$, where $\rho_n \equiv \frac{1}{263 n^2}$.
\end{lem}
\begin{proof}
It is convenient to define $v_m \equiv \|V_m\|_1$ and $s_m \equiv \|S_m\|_1$.
From Eq.~(\ref{eqn:Vm}), abbreviating $A_{k_p \dots k_1} \equiv i\ad_{S_{k_p}} \dots i\ad_{S_{k_1}}(H_0)$ and $B_{k_p \dots k_1} \equiv i\ad_{S_{k_p}} \dots i\ad_{S_{k_1}}(T)$ and using triangle inequality, we have
\begin{eqnarray*}
\|V_m\|_\text{F} &\leq& v_m \leq \sum_{p=2}^m \sum_{[k_1, \dots, k_p] = m} \!\!\! \mathfrak{f}(k_1, \dots, k_p) \| \, A_{k_p \dots k_1} \|_1 \\
&+& \sum_{p=1}^{m-1} \sum_{[k_1, \dots, k_p] = m-1} \!\!\!\! \mathfrak{f}(k_1, \dots, k_p) \| \, B_{k_p \dots k_1} \|_1 ~. \\
\end{eqnarray*}
Using Proposition~(\ref{prop:commutator_norm}) and the fact that $S_{k_\ell} \in \mathcal{T}_{k_\ell+1}$ and $k_\ell \leq n$, we have
\begin{eqnarray}
\| A_{k_p \dots k_1} \|_1 &\leq& 2^p (k_p + \dots + k_1 + 1) \dots (k_1 + 1) \nonumber \\
&\times& s_{k_p} \dots s_{k_1} \|H_0\|_1 \nonumber \\
&\leq& 2^p \left[\prod_{\ell=1}^p (\ell n + 1) \right] \times s_{k_p} \dots s_{k_1} \|H_0\|_1 \nonumber \\
&\leq& p! \left(\frac{n+1}{\Gamma} \right)^p v_{k_p} \dots v_{k_1} \|H_0\|_1 \nonumber  \\
&<& p! \left(\frac{n+2}{\Gamma} \right)^p v_{k_p} \dots v_{k_1} \|H_0\|_1 ~, \label{eqn:Abound}
\end{eqnarray}
where the last inequality is taken solely to simplify later calculations.
Similarly, we have
\begin{eqnarray*}
\| B_{k_p \dots k_1} \|_1 &\leq& 2^p (k_p + \dots + k_1 + 2) \dots (k_1 + 2) \\
&\times& s_{k_p} \dots s_{k_1} \|T\|_1 \\
&\leq& p! \left(\frac{n+2}{\Gamma} \right)^p v_{k_p} \dots v_{k_1} \|T\|_1 ~.
\end{eqnarray*}

Next, we use the relation $\sum_{[k_1, \dots, k_p] = m} \mathfrak{f}(k_1, \dots, k_p) \, (\bullet) = \frac{1}{p!} \sum_{[k_1, \dots, k_p] = m} \, (\bullet)$, where $(\bullet)$ is any summand symmetric under permutation of indices $k_1, \dots, k_p$.
We therefore obtain
\begin{eqnarray}
\label{eqn:vmbound}
v_m &\leq& \|H_0\|_1 \, \sum_{p=2}^m c^p \!\! \sum_{[k_1, \dots,  k_p] = m} \!\! v_{k_1} \dots v_{k_p} \nonumber \\
&+& \|T\|_1 \, \sum_{p=1}^{m-1} c^p \!\! \sum_{[k_1, \dots, k_p] = m-1} \!\! v_{k_1} \dots v_{k_p} ~, ~~~
\end{eqnarray}
where $c \equiv \frac{n+2}{\Gamma}$.

It is convenient to iteratively define another set of numbers, $\mu_m$, starting with $\mu_1 \equiv v_1$, and
\begin{eqnarray}
\label{eqn:mum}
\mu_m &\equiv& \|H_0\|_1 \sum_{p=2}^m c^p \sum_{k_1 + \dots + k_p = m} \mu_{k_1} \dots \mu_{k_p} \nonumber \\
&+& \|T\|_1 \sum_{p=1}^{m-1} c^p \sum_{k_1 + \dots + k_p = m - 1} \mu_{k_1} \dots \mu_{k_p} ~,
\end{eqnarray}
for $m \geq 2$.
Note that in the summation, the condition $k_\ell \leq n$ for $\ell = 1, \dots, p$ is omitted compared to Eq.~(\ref{eqn:vmbound}) but we are still requiring $1 \leq k_\ell$.
It is easy to show inductively that $v_m \leq \mu_m$ for all $m$.

We can now obtain bounds on the iteratively defined $\mu_m$ using auxiliary Taylor series
$\mu(z) \equiv \sum_{m=1}^\infty \mu_m z^m$.
It is easy to verify that $\mu(z)$ satisfies equation
\begin{eqnarray}
\mu &=& \|H_0\|_1 \left(\frac{1}{1 - c\mu} - 1 - c\mu \right) \nonumber \\
&+& \|T\|_1 \, z \, \left(\frac{1}{1 - c\mu} -1 \right) + v_1 \, z ~.
\end{eqnarray}
Indeed, by expanding the right-hand-side in powers of $\mu$, plugging in $\mu(z)$ series, and matching the coefficients of $z^m$ on both sides, we reproduce the iterative definition of $\mu_m$.
Solving for $\mu$ as a function of $z$ and noting $v_1 = \|T\|_1$, we have
\begin{equation*}
\mu(z) = \frac{1 - \sqrt{1 - 4\|T\|_1 (c + \|H_0\|_1 c^2) z}}{2(c + \|H_0\|_1 c^2)} ~,
\end{equation*}
where we have chosen the solution such that $\mu(0) = 0$.
Clearly, $\mu(z)$ is analytic in the disk $|z| \leq z_0$, where
\begin{equation}
z_0 \equiv \frac{1}{4\|T\|_1 (c + \|H_0\|_1 c^2)} \geq \frac{1}{263 n^2} \equiv \rho_n ~.
\end{equation}
Here the number $263$ is just a conservative estimation with no special meaning other than that the inequality holds for any $n \geq 1$, and we have used the fact that $\|H_0\|_1 = \Gamma$ and $\|T\|_1 \leq \sqrt{30} \|T\|_\text{F} = \sqrt{30} \Gamma$ from Prop.~\ref{prop:onenorm_vs_Frobenius}.

Furthermore, inside the disk, $|\mu(z)|$ is bounded by
\begin{equation}
|\mu(z)| \leq \frac{1}{2(c + \|H_0\|_1 c^2)} < \Gamma ~,
\end{equation}
where we have made a crude bound dropping any $n$ dependence since it will not affect considerations of the convergence of series in $m$ below.
By Cauchy's theorem,
\begin{eqnarray}
\mu_m &=& \frac{1}{2\pi i} \oint_{|z| = \rho_n} \frac{\mu(z)}{z^{m+1}} dz
\leq \frac{1}{2\pi i} \oint_{|z| = \rho_n} \left|\frac{\mu(z)}{z^{m+1}} \right |dz \nonumber \\
&\leq& \Gamma (\rho_n)^{-m} ~.
\end{eqnarray}
It follows that $\|V_m\|_F \leq v_m \leq \mu_m \leq \Gamma (\rho_n)^{-m}$.
\end{proof}

It is now easy to obtain the main bound:
\begin{thm}
\label{thm:boundHn}
If $\frac{\epsilon}{\rho_n} \leq \frac{1}{2}$, then $\|H_{>n}\|_\text{F} \leq 2\Gamma \left( \frac{\epsilon}{\rho_n} \right)^{n+1} = \mathcal{O}\left( n^2 \epsilon \right)^{n+1}$.
\end{thm}
\begin{proof}
We have
\begin{eqnarray}
\|H_{>n}\|_\text{F} &\leq& \sum_{m=n+1}^\infty \epsilon^m \|V_m\|_\text{F} \leq \Gamma \frac{(\epsilon/\rho_n)^{n+1}}{1 - \epsilon/\rho_n} \nonumber \\
&\leq& 2\Gamma \left( \frac{\epsilon}{\rho_n} \right)^{n+1} = \mathcal{O}\left( n^2 \epsilon \right)^{n+1} ~.
\end{eqnarray}
\end{proof}

This theorem also implies that for a fixed $n$, for small enough $\epsilon < \rho_n$ the local SW transformation has convergent expansion in $\epsilon$.
Since the expansion in $\epsilon$ is closely related to expansion in maximum range, we thus have such a convergent expansion in maximum range for the full SW-rotated Hamiltonian (at fixed $n$) in our definition of the $\|\bullet\|_\text{F}$ norm, or simply $U^\dagger H U$ belongs to the norm closure $\overline{\bigcup_{M \in \mathbb{N}} \mathcal{T}_M}$.

It is important that $n$ is understood as fixed since the available lower bound $\rho_n$ on the convergence radius goes to zero when $n \to \infty$.
Thus, even though we can formally define SW series developed to arbitrary order, their convergence as $n \to \infty$ is not guaranteed even for very small perturbation.
Nevertheless, bounds obtained at finite $n$ allow us to make rigorous lower bounds on the thermalization time as discussed in the main text.
We remark that while our bounds here are sufficient for a general nonquantitative discussion of prethermalization in the perturbative SW picture, we suspect that they are gross overestimates even in the spirit of such bounds.
Thus a numerical evaluation of such bounds in Appendix~\ref{app:conv_radius} suggests qualitatively tighter bounds $1/\rho_n \sim \mathcal{O}(n)$ and $\|H_{>n}\| \leq \mathcal{O}(n^n \epsilon^n)$, which would lead to a parametrically different thermalization time.\cite{ftnote}
In any case, we emphasize that all numerical calculations with the SW construction of the quasi-conserved quantity in the main text are exact and do not employ any such bounds (see also App.~\ref{app:Vmnorm}).

\section{Bound on $\ad_H(\tilde{I}^{(n)})$}
\label{app:boundadHI}
In this appendix, we give an upper bound on the squared residual norm of $\tilde{I}^{(n)}$, or $\|\ad_H(\tilde{I}^{(n)})\|_\text{F}^2$.
For the sake of simplicity, we further assume $\langle H_0, T \rangle = 0$ from now on.
Again, to bound $I_{>n}$, we need some control over the $I_m$ terms.
\begin{lem}
$\|I_m\|_\text{F} \leq \Gamma (\rho_n)^{-m}$, where $\rho_n = \frac{1}{263 n^2}$.
\end{lem}
\begin{proof}
Analogous to Lemma~\ref{lem:boundVm}, we have
\begin{eqnarray}
\|I_m\|_1 &\leq& \|H_0\|_1 \sum_{p=1}^m c^p \sum_{[k_1, \dots, k_p] = m} v_{k_1} \dots v_{k_p} \nonumber \\
&\leq& \|H_0\|_1 \sum_{p=1}^m c^p \sum_{k_1 + \dots + k_p = m} \mu_{k_1} \dots \mu_{k_p} \equiv \chi_m ~. ~~
\end{eqnarray}
Consider the auxiliary Taylor series $\chi(z) \equiv \sum_{m=1}^\infty \chi_m z^m$.
It is easy to verify that
\begin{equation}
\chi(z) = \|H_0\|_1 \left[\frac{1}{1 - c \mu(z)} - 1 \right] ~.
\end{equation}
$\chi(z)$ is analytic in the same domain as $\mu(z)$, i.e., in the disk $|z| < z_0$.
Inside the disk, $c |\mu(z)| \leq 1/2$ and $|\chi(z)| \leq \|H_0\|_1 = \Gamma$.
By Cauchy's theorem,
\begin{equation}
\chi_m = \frac{1}{2\pi i} \oint_{|z| = \rho_n} \frac{\chi(z)}{z^{m+1}} dz \leq \Gamma (\rho_n)^{-m} ~.
\end{equation}
It follows that $\|I_m\|_\text{F} \leq \|I_m\|_1 \leq \chi_m \leq \Gamma (\rho_n)^{-m}$.
\end{proof}

We can now find a bound on $I_{>n}$:
\begin{thm}
If $\frac{\epsilon}{\rho_n} \leq \frac{1}{2}$, then $\|I_{>n}\|_\text{F} \leq 2 \Gamma \left(\frac{\epsilon}{\rho_n} \right)^{n+1}$.
\end{thm}
\begin{proof}
Similarly to Theorem~\ref{thm:boundHn}, we have
\begin{eqnarray}
\|I_{>n}\|_\text{F} &\leq& \sum_{m=n+1}^\infty \epsilon^m \|I_m\|_\text{F} \leq 2 \Gamma \left( \frac{\epsilon}{\rho_n} \right)^{n+1} ~,
\end{eqnarray} 
provided $\epsilon/\rho_n \leq 1/2$.
\end{proof}
This theorem also assures that for fixed $n$ and small enough $\epsilon$, we have $\|I\|_\text{F} < \infty$; thus $I \in \overline{\bigcup_{M \in  \mathbb{N}} \mathcal{T}_M}$ under the norm $\|\bullet\|_\text{F}$.
Stated another way, for fixed $n$, the expansion in $\epsilon$ converges for small enough $\epsilon$; since this is essentially an expansion in the maximum range, the produced full $I$ is quasi-local.

We now turn to the truncation $I^{(n)}$ and its component $I^{(n)\perp}$ perpendicular to $H$ in the Frobenius inner product.
Since we want a normalized $\tilde{I}^{(n)}$, we first prove a lower bound on the norm of $I^{(n)\perp}$.
\begin{lem}
\label{lem:normInperp}
$\|I^{(n)\perp}\|_\text{F}^2 \geq \alpha \epsilon^2 \Gamma^2 + \Gamma^2 \mathcal{O}(n^6 \epsilon^3)$, where $\alpha > 0$ if $T^\text{diag} \neq 0$.
\end{lem}
\begin{proof}
From Eq.~(\ref{eqn:I(n)perp}), we have
\begin{equation}
\|I^{(n)\perp}\|_\text{F}^2 = \|I^{(n)}\|_\text{F}^2 - \frac{|\langle H, I^{(n)} \rangle|^2}{\|H\|_\text{F}^2} ~.
\end{equation}
Consider
\begin{eqnarray*}
\left|\langle I^{(n)}, H \rangle \right| &=& \left|\langle H_0, H \rangle + \sum_{m=1}^n \left( \epsilon^m \langle I_m, H_0 \rangle + \epsilon^{m+1} \langle I_m, T \rangle \right) \right| \nonumber \\
&\leq& \left|\Gamma^2 + \epsilon \langle I_1, H_0 \rangle + \epsilon^2 (\langle I_2, H_0 \rangle + \langle I_1, T \rangle) \right| \nonumber \\
&+& \sum_{m=3}^n \epsilon^m \|I_m\|_\text{F} \|H_0\|_\text{F} + \sum_{m=2}^n \epsilon^{m+1} \|I_m\|_\text{F} \|T\|_\text{F} ~,
\end{eqnarray*}
where we have used $\langle H_0, T \rangle = 0$.

The overlap between $I^{(n)}$ and $H$ can be calculated explicitly to $\mathcal{O}(\epsilon^2)$ as follows.
First, notice that $I_1 = -i\ad_{S_1}(H_0) = T - T^\text{diag} = T^\text{off-diag}$.
Therefore we have $\langle I_1, H_0 \rangle = 0$.
On the other hand, $\langle I_1, T \rangle = \|T^\text{off-diag}\|_\text{F}^2$.

Consider now $I_2 = \frac{1}{2} i\ad_{S_1} i\ad_{S_1}(H_0) - i\ad_{S_2}(H_0)$.
Since $i\ad_{S_2}(H_0) = V_2^\text{diag} - V_2 = -V_2^\text{off-diag}$, we have $\langle i\ad_{S_2}(H_0), H_0 \rangle = 0$.
Hence $\langle I_2, H_0 \rangle = -\frac{1}{2} \langle \ad_{S_1} \ad_{S_1}(H_0), H_0 \rangle = -\frac{1}{2} \langle \ad_{S_1}(H_0), \ad_{S_1}(H_0) \rangle = -\frac{1}{2} \|I_1\|_\text{F}^2 = -\frac{1}{2} \|T^\text{off-diag}\|_\text{F}^2$, where we have used 
$\langle \ad_{S_m}(A), B \rangle = \langle A, \ad_{S_m}(B) \rangle$ (which follows from hermiticity of $S_m$).

Combining the above calculations, we have
\begin{eqnarray}
\left| \langle I^{(n)}, H \rangle \right| &\leq& \Gamma^2 \left(1 + \epsilon^2 \frac{\|T^\text{off-diag}\|_\text{F}^2}{2 \Gamma^2} \right) \nonumber \\
&+& \Gamma \sum_{m=3}^\infty \epsilon^m (\|I_m\|_\text{F} + \|I_{m-1}\|_\text{F}) \nonumber \\
&\leq& \Gamma^2 \left[1 + \epsilon^2 \frac{\|T^\text{off-diag}\|_\text{F}^2}{2\Gamma^2} + 2\sum_{m=3}^\infty \left(\frac{\epsilon}{\rho_n} \right)^m \right] \nonumber \\
&\leq& \Gamma^2 \left[1 + \epsilon^2 \frac{\|T^\text{off-diag}\|_\text{F}^2}{2\Gamma^2} + 4 \left(\frac{\epsilon}{\rho_n} \right)^3 \right] ~,
\end{eqnarray}
where we have used $\rho_n < 1$ and assumed $\frac{\epsilon}{\rho_n} \leq \frac{1}{2}$.

We know $\|H\|_\text{F}^2 = \Gamma^2 (1 + \epsilon^2)$, since $\langle H_0, T \rangle = 0$.
Hence
\begin{eqnarray}
\frac{|\langle I^{(n)}, H \rangle|^2}{\|H\|_\text{F}^2} &\leq& \Gamma^2 \frac{\left[1 + \epsilon^2 \frac{\|T^\text{off-diag}\|_\text{F}^2}{2 \Gamma^2} + 4 \left(\frac{\epsilon}{\rho_n} \right)^3 \right]^2}{1 + \epsilon^2} \nonumber \\
&=& \Gamma^2 \left[1 - \alpha \epsilon^2 + \mathcal{O}(n^6 \epsilon^3) \right] ~,
\end{eqnarray}
where $\alpha \equiv 1 - \frac{\|T^\text{off-diag}\|_\text{F}^2}{\Gamma^2} > 0$ if $T^\text{diag} \neq 0$.
If $T^\text{diag} = 0$ so that $\alpha = 0$, one has to verify the negativity of the coefficient of the next order $\epsilon^3$.
While we expect this to be true, to simplify the discussion we made the assumption that $T^\text{diag} \neq 0$.

Finally, we have
\begin{eqnarray}
\|I^{(n)}\|_\text{F} &=& \|U H_0 U^\dagger - I_{>n}\|_\text{F} \geq \|H_0\|_\text{F} - \|I_{>n}\|_\text{F} \nonumber \\
&\geq& \Gamma \left[1 - 2 \left(\frac{\epsilon}{\rho_n} \right)^{n+1} \right] ~.
\end{eqnarray}
We can therefore obtain
\begin{eqnarray}
\|I^{(n)\perp}\|_\text{F}^2 &\geq& \Gamma^2 \left[1 + \mathcal{O}(n^2 \epsilon)^{n+1} \right]^2 - \Gamma^2 \left[1 - \alpha \epsilon^2 + \mathcal{O}(n^6 \epsilon^3) \right] \nonumber \\
&=& \alpha \epsilon^2 \Gamma^2 + \Gamma^2 \mathcal{O}(n^6 \epsilon^3) ~.
\end{eqnarray}
\end{proof}

We now have the ingredients for bounding $\ad_H(\tilde{I}^{(n)})$ and can prove the following theorem:
\begin{thm}
$\|\ad_H(\tilde{I}^{(n)})\|_\text{F}^2 = \frac{\|\ad_H(I^{(n)\perp})\|_\text{F}^2}{\|I^{(n)\perp}\|_\text{F}^2} \leq \mathcal{O}\left(n^{4n + 6} \epsilon^{2n} \right)$.
\end{thm} 
\begin{proof}
First, we note that 
\begin{eqnarray}
\|\ad_H(I^{(n)\perp})\|_\text{F} &=& \|\ad_H(I^{(n)})\|_\text{F} = \|\ad_H(I) - \ad_H(I_{>n})\|_\text{F} \nonumber \\
&\leq& \|\ad_H(I)\|_\text{F} + \|\ad_H(I_{>n})\|_\text{F} ~.
\label{eqn:adHInperp}
\end{eqnarray}
The first term can be bounded by
\begin{eqnarray}
\|\ad_H(I)\|_\text{F} &=& \|[H, U H_0 U^\dagger]\|_\text{F} = \|[U^\dagger H U, H_0]\|_\text{F} \nonumber \\
&=& \|[H_{>n}, H_0]\|_\text{F} \leq \|[H_{>n}, H_0]\|_1 \nonumber \\
&\leq& \sum_{m=n+1}^\infty \epsilon^m \|[V_m, H_0]\|_1 \nonumber \\
&\leq& \sum_{m=n+1}^\infty \epsilon^m 2(m+1) v_m \|H_0\|_1 \nonumber \\
&\leq& 2 \Gamma^2 \sum_{m=n+1}^\infty (m+1) \left(\frac{\epsilon}{\rho_n} \right)^m \nonumber \\
&=& 2 \Gamma^2 \frac{(n+2) \beta^{n+1}}{(1 - \beta)^2} \left(1 - \beta \frac{n+1}{n+2} \right) \nonumber \\
&\leq& 8 \Gamma^2 (n+2) \beta^{n+1} ~,
\end{eqnarray}
where we have defined $\beta \equiv \frac{\epsilon}{\rho_n}$ and used $\beta \leq 1/2$.

The second term in Eq.~(\ref{eqn:adHInperp}) can be bounded as
\begin{eqnarray}
\|\ad_H(I_{>n})\|_\text{F} &\leq& \sum_{m=n+1}^\infty \epsilon^m \|\ad_H(I_m)\|_1 \nonumber \\
&\leq& \sum_{m=n+1}^\infty \epsilon^m 2(m + 2) \|H\|_1 \|I_m\|_1 \nonumber \\
&\leq& \sum_{m=n+1}^\infty \epsilon^m 2(m + 2) \sqrt{30} \sqrt{1 + \epsilon^2} \Gamma \|I_m\|_1 \nonumber \\
&\leq& 4\sqrt{15} \Gamma^2 \sum_{m=n+1}^\infty (m + 2) \left(\frac{\epsilon}{\rho_n} \right)^m \nonumber \\
&=& 4\sqrt{15} \Gamma^2 \frac{(n+3) \beta^{n+1}}{(1 - \beta)^2} \left(1 - \beta \frac{n+2}{n+3} \right) \nonumber \\
&\leq& 16\sqrt{15} \Gamma^2 (n + 3) \beta^{n+1} ~,
\end{eqnarray}
where we have used $\|H\|_1 \leq \sqrt{30} \|H\|_\text{F} = \sqrt{30} \Gamma \sqrt{1 + \epsilon^2}$ and $\sqrt{1 + \epsilon^2} < \sqrt{2}$.

Combining the above two bounds and Lemma~\ref{lem:normInperp}, we have
\begin{eqnarray}
\frac{\|\ad_H(I^{(n)\perp})\|_\text{F}^2}{\|I^{(n)\perp}\|_\text{F}^2} &\leq& 
\frac{\left[ \Gamma^2 (a n + b) \left(\frac{\epsilon}{\rho_n} \right)^{n+1} \right]^2}{\alpha \epsilon^2 \Gamma^2 + \Gamma^2 \mathcal{O}(n^6 \epsilon^3)} \nonumber \\
&=& \Gamma^2 \mathcal{O}\left(n^{4n+6} \epsilon^{2n} \right) ~,
\end{eqnarray}
where $a = 8 + 16\sqrt{15}$, $b = 16 + 48\sqrt{15}$.
\end{proof}

\section{Better bounds on $\|V_m\|$ and the convergence radius using numerical experiments}\label{app:conv_radius}

\begin{figure}
\includegraphics[width=1.0\columnwidth]{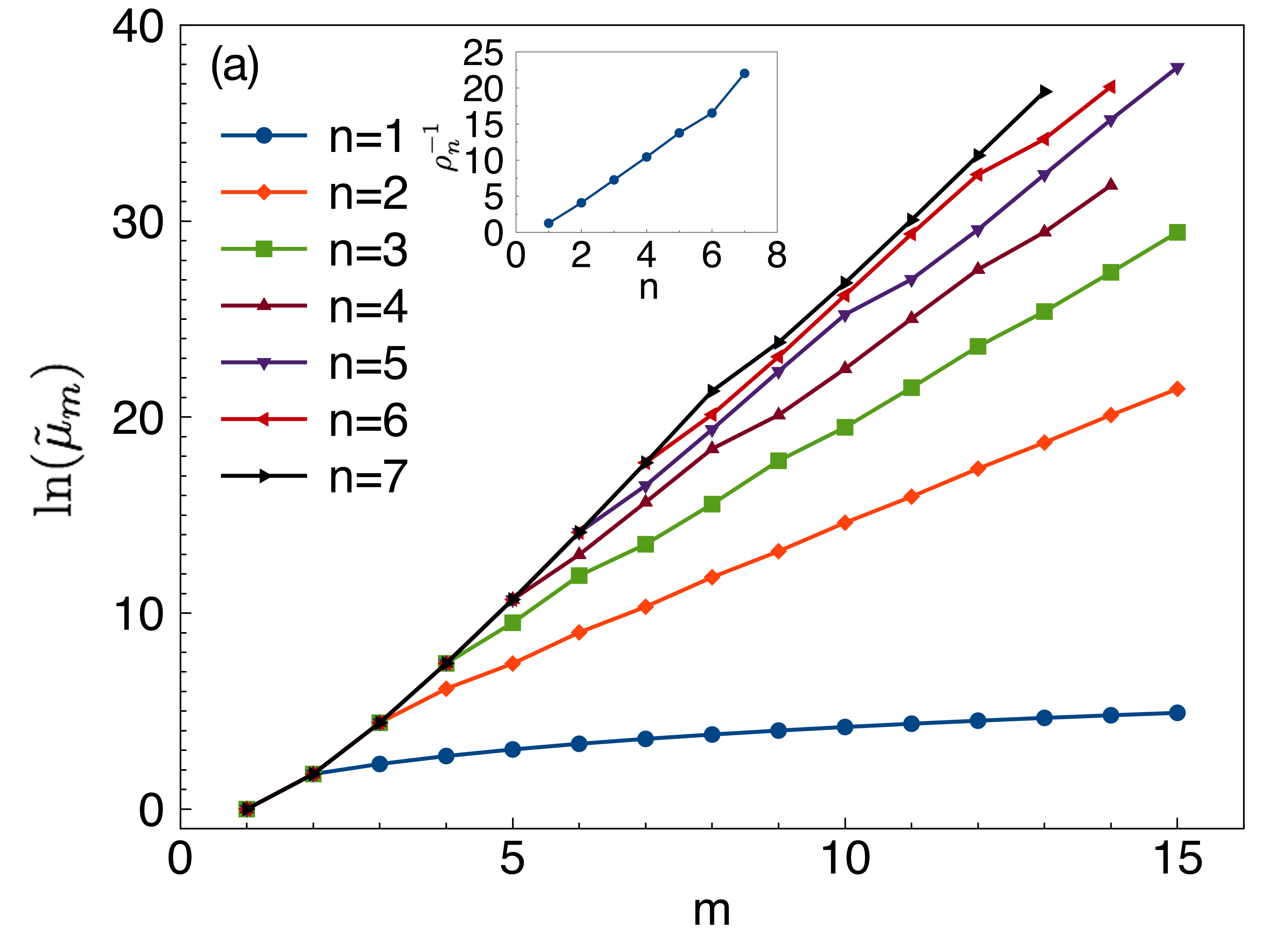}
\includegraphics[width=1.0\columnwidth]{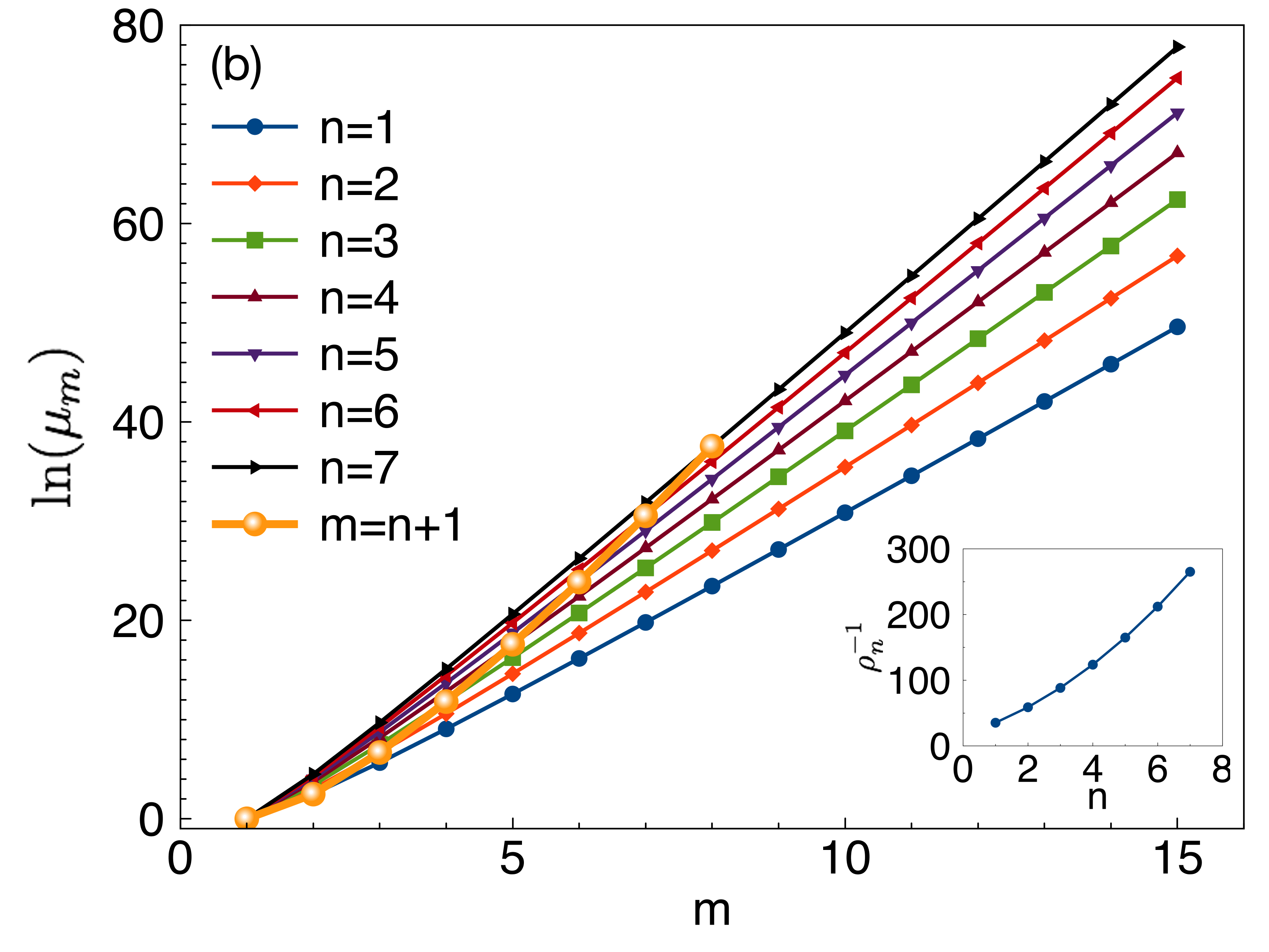}
\caption{\label{fig:conv_radius}
(color online) Numerical calculations of the iterative bounds on $\|V_m\|_1$: (a) $\tilde{\mu}_m$ generated by Eq.~(\ref{eqn:tildemum}) and (b) $\mu_m$ generated by Eq.~(\ref{eqn:mum}), for different SW order $n$.
For convenience, the one-norms of $\|H_0\|_1$ and $\|T\|_1$ are taken to be one, which does not affect the functional dependence of the convergence radius $\rho_n$ on $n$.
The curve $m = n \!+\! 1$ in (b) denotes the bound on the infinite-SW $\|V_m\|_1$ since $V_m$ does not depend on $n$ once $n \geq m \!-\! 1$.
Insets: the inverse convergence radius $\rho_n^{-1}$ as a function of $n$.
By assuming $\tilde{\mu}_m = A_n (\rho_n)^{-m}$, or $\ln(\tilde{\mu}_m) = \ln A_n - m \ln(\rho_n)$ for $m > n$, we can extract $\ln(1/\rho_n)$ from the slope of $\ln(\tilde{\mu_m})$ vs $m$ and plot $\rho_n^{-1}$ in the inset.
For $\rho_n$ extracted from $\tilde{\mu}_m$, we suspect $\rho_n^{-1} \sim n$; while for $\mu_m$, we observe $\rho_n^{-1} \sim n^2$ as expected.}
\end{figure}

In Appendix~\ref{app:boundHn}, we estimated the convergence radius $\rho_n \sim 1/n^2$, which is a lower bound.
This would give the thermalization time scale to be $\mathcal{O}(\exp(A/\sqrt{\epsilon}))$, where $\epsilon$ is the perturbation strength.
In this Appendix, we demonstrate a numerical experiment to support the conjecture that a tighter bound $\rho_n \sim 1/n$ is possible.

Recall that when bounding $v_m$, following Ref.~\onlinecite{bravyi_schriefferwolff_2011}, we used a very crude bound of $(k_p + \dots k_1 + 1) \dots (k_1 + 1) \leq p! (n+1)^p$, see Eq.~(\ref{eqn:Abound}).
We suspect that this approximation, which allowed an analytical calculation of the numbers $\mu_m$ which bound $v_m$, Eq.~(\ref{eqn:mum}), is however too crude and changes the leading behavior of the convergence radius $\rho_n$.
If we do not make this approximation, we can define another set of numbers $\tilde{\mu}_m$ which bound $v_m$:
\begin{eqnarray}
\tilde{\mu}_m &\equiv& \sum_{p=2}^m \sum_{[k_1, \dots, k_p] = m} \mathfrak{f}(k_1, \dots, k_p) \nonumber \\ 
&& \times (k_p + \dots + k_1 + 1) \dots (k_1 + 1) \, \tilde{\mu}_{k_1} \dots \tilde{\mu}_{k_p} \nonumber \\
&+& \sum_{p=1}^{m-1} \sum_{[k_1, \dots, k_p] = m - 1} \mathfrak{f}(k_1, \dots, k_p) \nonumber \\
&& \times (k_p + \dots + k_1 + 2) \dots (k_1 + 2) \, \tilde{\mu}_{k_1} \dots \tilde{\mu}_{k_p} ~, ~~~~~ \label{eqn:tildemum}
\end{eqnarray}
where we have assumed $\|H_0\|_1 = \|T\|_1 = \Gamma = 1$, without loss of generality.

Starting with $\tilde{\mu}_1 \equiv v_1 = 1$, we can iteratively calculate $\tilde{\mu}_m$ for a given $n$.
The results are shown in Fig.~\ref{fig:conv_radius}(a).
Recall that $V_m$ for $m \leq n + 1$ are already independent of $n$ (and can be viewed as representative of the infinite-order SW procedure), while $V_m$ for $m > n + 1$ describe formal expansion in powers of $\epsilon$ at fixed $n$ and form the ``remainder'' $H_{>n}$.
The same property is shared by $\tilde{\mu}_m$, i.e., $\tilde{\mu}_m$ for $m \leq n + 1$ are independent of $n$ and appear as the limiting curve in Fig.~\ref{fig:conv_radius}(a), while the data for $m > n$ determine convergence properties of the remainder $H_{>n}$.
For easy reference, we quote several numbers on the limiting curve, which are ``universal'' numbers under this bounding procedure: $\tilde{\mu}_2 = 6, \tilde{\mu}_3 = 82, \tilde{\mu}_4 = 1695, \tilde{\mu}_5 = 43995$, etc.
Focusing now on the remainder terms and assuming behavior $\tilde{\mu}_m = A_n (\rho_n)^{-m}$ for $m > n$, we can extract the convergence radius from the slope of $\ln(\tilde{\mu}_m) = \ln(A_n) - m \ln(\rho_n)$ vs $m$.
The inset shows the $n$ dependence of the inverse convergence radius $(\rho_n)^{-1}$, which in fact suggests $\rho_n^{-1} \sim n$.

As a comparison, in Fig.~\ref{fig:conv_radius}(b) we also show the same procedure applied to $\mu_m$, Eq.~(\ref{eqn:mum}), with the same normalization $\|H_0\|_1 = \|T\|_1 = \Gamma = 1$.
In this case, the inverse convergence radius $\rho_n^{-1}$ shows $n^2$ behavior, as expected from the analysis in App.~\ref{app:boundHn}.
Note that in this case we did not treat separately $m \leq n + 1 $ and $m > n + 1$, since we used the same $n$-dependent $c$ in the iteration equation for all $m$.
Of course, we know that $v_m$ no longer depends on $n$ for $m \leq n + 1$, and for each $m$ we could use $\mu_m$ from the smallest SW order $n$ satisfying this condition to bound such infinite-SW-order $v_m$;
these are indicated as ``$m = n + 1$'' curve in Fig.~\ref{fig:conv_radius}(b), and we expect such procedure to bound $v_m$ by $m^{2m}$.

To conclude, we thus suspect that the lower bound on the convergence radius can be possibly tighter than in App.~\ref{app:boundHn} and is tentatively $\rho_n \sim 1/n$, though we do not have a rigorous mathematical proof.
Related to this, the behavior of $\tilde{\mu}_m$ for $m \leq n + 1$, which bounds the infinite-SW-order $v_m$, appears to be $\ln(\tilde{\mu}_m) = m \ln(m)$ up to subdominant contributions, compared to $\ln(\mu_m) = 2 m \ln(m)$ [this could be crudely seen by noting that the vertical range in panel (a) in Fig.~\ref{fig:Vmnorm} is two times smaller than in panel (b)].

\section{Numerical results for $\|V_m\|$ in a generic model}
\label{app:Vmnorm}

\begin{figure}
\includegraphics[width=1.0\columnwidth]{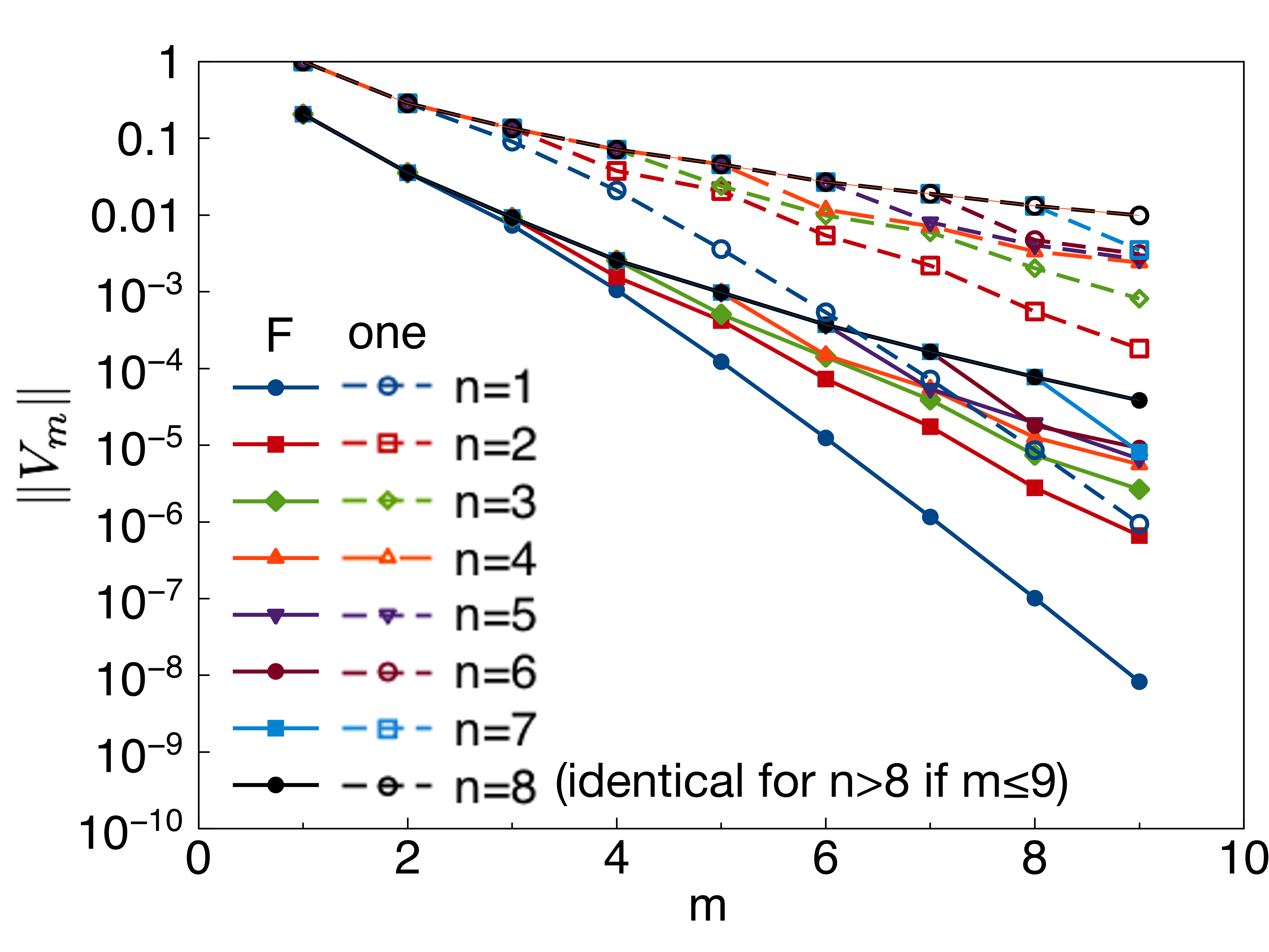}
\caption{\label{fig:Vmnorm}
(color online) Actual Frobenius norms and one-norms (denoted by ``F'' and ``one'' respectively) of operators $V_m$ that appear in the SW transformation, to be compared with bounds in Fig.~\ref{fig:conv_radius}.
The model is defined using $H_0$ in Eq.~(\ref{H0app}) with $\Gamma = 1$ and $T$ in Eqs.~(\ref{T+2})-(\ref{T0}) with $t_1 = -12/121$, $t_2 = 16/121$, $u_0 = 9/121$, and $w_0 = t_2$.
The parameters are chosen such that $\|H_0\|_1 = \|T\|_1 = 1$ and the ratios among $t_1$, $t_2$, $w_0$, $u_0$ corresponding to the case with $h = 1.5$ and $g = 2.0$ in the main text.
Note that the actual $\|V_m\|_\text{F}$ and $\|V_m\|_1$ are still decreasing for the accessible $m$, in stark contrast with the bounds that show very fast increase (at least $m^m$) starting already at $m = 1$.}
\end{figure}

In Sec.~\ref{subsec:SWproc}, we defined the local SW procedure to produce an effective Hamiltonian that commutes with $H_0$ up to order $n$.
The procedure gives ``potentials'' $V_m$ and the generator $i S_m$ is chosen to eliminate the off-diagonal part of $V_m$ for $m \leq n$.
In App.~\ref{app:boundHn}, we provided an analytical bound $\|V_m\|_\text{F} \leq \|V_m\|_1 \leq \Gamma (\rho_n)^{-m}$, where the inverse convergence radius grows as $1/\rho_n \sim n^2$.
These are bounds valid for all $m$ but are particularly used for $m > n$ bounding the terms in the remainder $H_{>n}$, while for $m \leq n$ where $V_m$ are already independent of $n$ we can bound $\|V_m\|_1 \leq \Gamma (\rho_m)^{-m}$.
From numerical experiments in App.~\ref{app:conv_radius} with more accurate bounds, we see that the bounds in App.~\ref{app:boundHn} are too crude and better bounds are possible, tentatively with $1/\rho_n \sim n$.

In this Appendix, we directly calculate $\|V_m\|_\text{F}$ and $\|V_m\|_1$, with no approximations, in a generic model to compare with these theoretical bounds.
All numerical results on the SW-generated quasi-conserved operators in the main text are also obtained with no approximations but contain all terms including all factors of $\epsilon^m$ summed up, while the purpose of this Appendix is to measure individual $V_m$ terms for direct comparisons with theoretical bounds.
Since $S_m$ is determined from $V_m$ by a relatively simple local rule and the structure of $I_m$ is similar to $V_m$, we expect the results for all these operators to be qualitatively similar and will focus on the potentials $V_m$.
Figure~\ref{fig:Vmnorm} shows the numerical values of $\|V_m\|_\text{F}$ and $\|V_m\|_1$, calculated for the SW-generated potentials for the model in App.~\ref{app:ladderformalism} taking $H_0$ in Eq.~(\ref{H0app}) with $\Gamma = 1$ and $T$ in Eqs.~(\ref{T+2})-(\ref{T0}) with $t_1 = -12/121$, $t_2 = 16/121$, $u_0 = 9/121$, and $w_0 = t_2$.
The parameters are chosen such that $\|H_0\|_1 = \|T\|_1 = 1$, while the ratios among $t_1$, $t_2$, $w_0$, $u_0$ are such that they correspond to the case with $J = 1$, $h = 1.5$, and $g = 2$ in the main text rotated to the new basis as described in App.~\ref{app:ladderformalism}; to directly connect with this data point in the main text, the appropriate $\epsilon$ is approximately $1.936$.

Recall that the $V_m$ generated by the SW procedure are independent of the perturbation parameter $\epsilon$ but contain all information needed for evaluating series for any $\epsilon$.
The above normalization of $H_0$ and $T$ is chosen such that we can directly compare with the numbers in App.~\ref{app:conv_radius}.
The best bounds in App.~\ref{app:conv_radius} are very quickly increasing already starting with $m = 1$, reaching values $e^{21} \sim 10^9$ already for $m = 8$, see top panel in Fig.~\ref{fig:Vmnorm} remembering that it plots logarithms of the bounds on $\|V_m\|_1$.
On the other hand, the actual values of $\|V_m\|_1$ are decreasing with $m$ for accessible $m$.
This suggests that even the best theoretical upper bound on $\|V_m\|_1$ is a vast overestimation.
In fact, taken at face value, the numerical results in Fig.~\ref{fig:Vmnorm} might even suggests the possibility of convergence of the SW procedure in some models.
A more conservative view is that the actual $\|V_m\|_1$ will eventually start increasing for large enough $m$, and the initial decrease is due to the chosen normalization $\|H_0\|_1 = \|T\|_1 = 1$ where the one-norm measure is somehow less fair between the 1-local and 2-local terms.
However, we emphasize that the bounds in App.~\ref{app:conv_radius} are obtained for exactly the same normalization and the comparison with the bounds in Fig.~\ref{fig:Vmnorm} is fair.
(We needed to use the one-norm in the theoretical bounds because we were not able to prove analogs of Props.~\ref{prop:commutator_norm} and \ref{prop:boundSm} for the Frobenius norm.)
The large difference between the actual norm and the theoretical bound starts already at $m = 2$, where we have verified by direct analytical calculation of the potential $V_2$ in Eq.~(\ref{eqn:V2app}) that $\|V_2\|_1 \approx 0.286$ while the bound $\tilde{\mu}_2 = 6$.

One likely source of the overestimation is that the theoretical bounds always replace the norm of a sum of a large number of terms by the sum of norms of the terms, while there can be many cancellations among the terms.
More specifically, we can trace the faster-than-exponential growth of the bounds $\tilde{\mu}_m$ to factors $(k_p + \dots k_1 + 1) \dots (k_1 + 1)$ in the second line of Eq.~(\ref{eqn:tildemum}) and $(k_p + \dots k_1 + 2) \dots (k_1 + 2)$ in the fourth line of Eq.~(\ref{eqn:tildemum}), which in turn originate from the factor $r + s - 1$ in the bound in Prop.~\ref{prop:commutator_norm} for a commutator of an operator in $\mathcal{T}_r$ and an operator in $\mathcal{T}_s$.
Examining Eq.~(\ref{eqn:adUW}) and how it is used in the proof of Prop.~\ref{prop:commutator_norm}, we see that there are $2 (r + s - 1) \cdot 3 \cdot 4^{r-1} \cdot 3 \cdot 4^{s-1}$ terms that are being collected, while the number of basis states for writing out $\ad_U(W) \in {\cal T}_{r+s-1}$ is $3 \cdot 4^{r + s - 2}$.
(Here for simplicity we ignore generation of multiple strings from products $Q^{\mathbf{a}}_{j; r} Q^{\mathbf{b}}_{k; s}$.)
Thus, an amplitude for each basis state will have roughly $6 (r + s - 1)$ contributions.
If these contributions all came with the same sign, we would indeed obtain the bound in Prop.~\ref{prop:commutator_norm}.
However, different contributions can come with different signs depending on details of various commutators.
If these signs were uncorrelated, it would be natural to replace $6 (r + s - 1)$ by $\sqrt{6 (r + s - 1)}$ when estimating a typical amplitude in the operator string basis, and such a replacement could potentially bring the bound on the growth of $\|V_m\|_1$ from $m^m$ to a much slower $m^{m/2}$.
Thus, such cancellations, while still not preventing eventual thermalization, could potentially lead to 
parametrically longer relaxation times as a function of $\epsilon$.\cite{ftnote}

Interestingly, there can be additional suppression of the growth of the bounds $\tilde{\mu}_m$ when we consider more carefully the bound in Prop.~\ref{prop:boundSm}.
Indeed, the denominator in Prop.~\ref{prop:boundSm} represents the smallest possible energy difference between the energy sectors of $H_0$.
However, at $m$-th order, $V_m$ consists of pieces that have $m$ of elementary (i.e., from the bare perturbation $T$) raising or lowering steps on the $H_0$ sector label.
We may then guess that a typical term in $V_m$ would be raising or lowering the $H_0$ sector label by roughly $\sqrt{m}$, so for estimating a typical contribution we could replace the denominator $2\Gamma$ in Prop.~\ref{prop:boundSm} with $2\Gamma \sqrt{m}$.
However, we caution that the discussed cancellations and suppressions compared to the earlier bounds implicitly assume lack of structure among the various complicated terms, hence random-walk-type estimates.
If there is a structure that would lead to some sign or magnitude bias among the terms, this could possibly arrest the discussed suppressions.
Our numerical experiment in Fig.~\ref{fig:Vmnorm} where we have not seen faster-than-exponential growth yet, together with the speculative arguments above, suggest that the convergence of the SW procedure is an open question worth further explorations.
Even if eventually the convergence radius vanishes, we clearly expect strong quantitative and perhaps qualitative modifications of how this happens, which would also have implications for estimates of the relaxation times.

%

\end{document}